\documentclass[english,thm-restate,nolineno]{socg-lipics-v2021}

\hideLIPIcs

\usepackage{amsmath,amssymb,amsfonts}
\usepackage{tcolorbox}

\usepackage{verbatim}
\usepackage{xspace}
\usepackage{edtable}

\usepackage{etoolbox} 

\usepackage{thmtools}
\usepackage{mathtools}

\usepackage{tikz-cd}

\usepackage{hyperref}
\usepackage{xcolor}
\definecolor{darkred}{RGB}{220,50,0}
\definecolor{lightblue}{rgb}{.80,.85,1}
\definecolor{darkgreen}{RGB}{0,100,0}
\definecolor{firebrick}{RGB}{178,34,34}
\definecolor{salmon}{RGB}{250,128,114}
\definecolor{turquoise}{RGB}{0,128,114}
\definecolor{turquoise2}{RGB}{0,180,140}
\definecolor{darkorchid}{rgb}{0.60,0.20,0.80}

\usepackage[normalem]{ulem}

\graphicspath{{pictures/}}

\theoremstyle{plain}
\newtheorem{Defremark2}[theorem]{\textbf{Definition and Remark}}

\usepackage[section]{placeins}

\usepackage{caption}
\usepackage[skip=0.5ex]{subcaption}

\newcommand{\Andre}[1]{{\color{turquoise2} #1}}
\newcommand{\SavedComment}[1]{}

\DeclareMathOperator*{\argmin}{argmin}
\DeclareMathOperator*{\arccosh}{arccosh}
\DeclareMathOperator{\ax}{{ax}}

\newcommand{\ud}{\mathrm{d}}
\newcommand\R{\mathbb{R}}
\newcommand{\Sphere}{\mathbb{S}}
\newcommand{\M}{\mathcal{M}}
\newcommand{\N}{\mathcal{N}}
\newcommand{\Su}{{\mathcal{S}}}

\newcommand{\lfs}{\mathrm{lfs}}
\newcommand{\Tan}{\mathrm{Tan}}
\newcommand{\Nor}{\mathrm{Nor}}
\newcommand{\rch}{\mathrm{rch}}

\newcommand{\Dual}{\mathrm{Dual}} 
\newcommand{\hanka}[1]{{\color{blue}{[#1]}}} 
 
\newcommand{\reach}{\ensuremath{\mathcal{R}}\xspace}
\newcommand{\norm}[1]{\left\|{#1}\right\|}
\newcommand{\betaAp}{\varepsilon}

\newcommand{\rchcl}{\mathrm{rch}^{\operatorname{cl}}}

\newcommand{\curvlowbnd}{{\Lambda_{\ell}}}
\newcommand{\curvlowbndn}{{\abs{\curvlowbnd}}}

\newcommand{\Ninteger}{\mathbb{N}}
\newcommand{\Hull}{\mathcal{CH}}
\newcommand{\MinAtZero}{X_{\min}^0}
\newcommand{\CplOffset}{\mathsf{C} }
\newcommand{\abs}[1]{{\left\lvert #1 \right\rvert}}

\newcommand{\pairf}[1]{{\tiny $p_{#1}$}}
\newcommand{\pairs}[1]{{\tiny $\tilde{p}_{#1}$}}

\newcommand{\defunder}[1]{\underset{\text{def.}}{#1} \:}

\newenvironment{Altproof}[1][{}]{
  \begin{trivlist}\item[]\textit{Alternative proof #1}\quad}%
  {\hfill\hspace*{\fill}~$\square$\end{trivlist}}

\newenvironment{SketchProof}[1][{}]{
  \begin{trivlist}\item[]\textit{Sketch of proof #1}\quad}%
  {\hfill\hspace*{\fill}~$\square$\end{trivlist}}

\title{
Tight Bounds for the Learning of Homotopy {\`a} la Niyogi, Smale, and Weinberger for Subsets of Euclidean Spaces and of Riemannian Manifolds
}
\titlerunning{Learning Homotopy in Euclidean Spaces and Riemannian Manifolds
}

\author{Dominique Attali}{Universit{\'e} Grenoble Alpes, CNRS, Grenoble INP, GIPSA-lab\\{[Grenoble, France]} }{Dominique.Attali@grenoble-inp.fr}{}{}

\author{Hana Dal Poz Kou\v{r}imsk\'a}{IST Austria \\{[Klosterneuburg, Austria]}}{hana.kourimska@ist.ac.at}{https://orcid.org/0000-0001-7841-0091}{} 

\author{Christopher Fillmore}{IST Austria \\{[Klosterneuburg, Austria]}}{christopher.fillmore@ist.ac.at}{https://orcid.org/0000-0001-7631-2885}{}

\author{Ishika Ghosh}{IST Austria \\{[Klosterneuburg, Austria]} \\ Michigan State University \\ {[East Lansing, USA]} }{ghoshis3@msu.edu}{https://orcid.org/0000-0002-7901-5912}{}

\author{Andr{\'e} Lieutier}{No affiliation\\{[Aix-en-Provence, France]}}{andre.lieutier@gmail.com }{}{}

\author{Elizabeth Stephenson}{IST Austria \\{[Klosterneuburg, Austria]}}{elizabeth.stephenson@ist.ac.at}{https://orcid.org/0000-0002-6862-208X}{}

\author{
Mathijs Wintraecken}{Inria Sophia Antipolis, Universit{\'e} C{\^o}te d'Azur\\{[Sophia Antipolis, France]}  }{m.h.m.j.wintraecken@gmail.com}{https://orcid.org/0000-0002-7472-2220}{Supported by the European Union's Horizon 2020 research and innovation programme under the Marie Sk{\l}odowska-Curie grant agreement No. 754411, the Austrian science fund (FWF) grant No. M-3073, and the welcome package from IDEX of the Universit{\'e} C{\^o}te d'Azur. }

\funding{This research has been supported by the European Research Council (ERC), grant No.\ 788183, by the Wittgenstein Prize, Austrian Science Fund (FWF), grant No.\ Z 342-N31, and by the DFG Collaborative Research Center TRR 109, Austrian Science Fund (FWF), grant No.\ I 02979-N35.
}

\acknowledgements{
We thank Jean-Daniel Boissonnat, Herbert Edelsbrunner, and Mariette Yvinec for discussion. 
}

\authorrunning{
	D. Attali, H. Dal Poz Kou\v{r}imsk\'a, C. Fillmore, I. Ghosh, A. Lieutier, E. Stephenson, and M. Wintraecken
} 

\Copyright{
Dominique Attali, Hana Dal Poz Kou\v{r}imsk\'a, Christopher Fillmore, Ishika Ghosh, Andr{\'e} Lieutier, Elizabeth Stephenson, and Mathijs Wintraecken
} 



\EventEditors{Wolfgang Mulzer and Jeff M. Phillips}
\EventNoEds{2}
\EventLongTitle{40th International Symposium on Computational Geometry (SoCG 2024)}
\EventShortTitle{SoCG 2024}
\EventAcronym{SoCG}
\EventYear{2024}
\EventDate{June 11-14, 2024}
\EventLocation{Athens, Greece}
\EventLogo{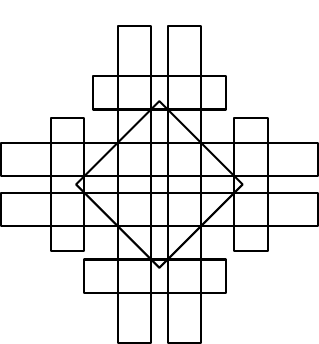}
\SeriesVolume{293}
\ArticleNo{9}

\begin{document}
\maketitle

\begin{abstract} 
In this article we extend and strengthen the seminal work by Niyogi, Smale, and Weinberger on the learning of the homotopy type from a sample of an underlying space.  
In their work, Niyogi, Smale, and Weinberger studied samples of $C^2$ 
manifolds with positive reach embedded in $\mathbb{R}^d$. We extend their results in the following ways:

	\begin{itemize}
		\item As the ambient space we consider both $\mathbb{R}^d$ and Riemannian manifolds with lower bounded sectional curvature.
		\item In both types of ambient spaces, we study sets of positive reach --- a significantly more general setting than $C^2$ manifolds --- as well as general manifolds of positive reach.
		\item The sample $P$ of a set (or a manifold) $\Su$ of positive reach may be noisy. We work with two one-sided Hausdorff distances --- $\varepsilon$ and $\delta$ --- between $P$ and $\Su$. We provide tight bounds in terms of $\varepsilon$ and $\delta$, that guarantee that there exists a parameter $r$ such that the union of balls 
		of radius $r$ centred at the sample $P$ deformation-retracts 
		to $\Su$. We exhibit their tightness by an explicit construction. 
	\end{itemize}

	We carefully distinguish the roles of $\delta$ and $\varepsilon$. This is not only essential to achieve tight bounds, but also sensible in practical situations, since it allows one to adapt the bound according to sample density and the amount of noise present in the sample separately. 
\end{abstract}

\section{Introduction}\label{section:introduction}

Can we infer the topology of a set if we are only given partial geometric information about it? 
Under which conditions is such inference possible?

These questions were first motivated by the shape reconstruction of objects in 3-dimensional Euclidean space. There, the partial geometric information was represented by a finite, in general noisy, set of points obtained from photogrammetric or lidar measurements
 \cite{amenta2000simple,berger2017survey,boissonnat1984geometric,boissonnat1988shape,chazal2008smooth}.
 
More recently, the same questions have arisen 
in the context of learning and topological data analysis (TDA). In these fields, one seeks to recover a (relatively) low-dimensional support of 
a probability measure in a high-dimensional 
space, given a (finite) data set drawn from this probability measure \cite{boissonnat2018geometric, chazal2011geometric, edelsbrunner2010computational, dey2022computational}. 
Assuming the support is a manifold, one calls this process {manifold learning} \cite{pless2009survey}.

In \cite{niyogi2008}, Niyogi, Smale, and Weinberger showed that, given a $C^2$ manifold of positive reach\footnote{ 
We recall that the reach of a closed subset in Euclidean space is the distance from the set to its medial axis. In turn, the medial axis of a set consists of those points in Euclidean space that do not have a unique closest point on the set. Both notions are defined in Definition~\ref{def:medial_axis}.
} {embedded in Euclidean space} and a sufficiently dense point sample on (or near) the manifold, the union of balls of certain radii centred on the point sample captures the homotopy type of the manifold. 
By the nerve theorem~\cite{edelsbrunner2010computational}, the homotopy type of the union of balls is shared by the \v{C}ech complex  \cite{bjorner1995topological,edelsbrunner1994triangulating} 
and $\alpha$-complex \cite{edelsbrunner2011alpha}
of the point sample. From these complexes we can then learn the topological information such as the homology groups of the underlying manifold. 
Niyogi, Smale, and Weinberger's homotopy learning result has led to numerous generalizations including \cite{antasik2022sampling, ATTALI2013448, chazal2009sampling, kim2019homotopy, WANG2020101606}. 

In this article, we revisit the work of Niyogi, Smale, and Weinberger, generalizing the settings of their work in various ways.
	
The first generalization is in terms of ambient space --- we consider both the Euclidean space $\mathbb{R}^d$ and Riemannian manifolds with bounded sectional curvature. To this end, we introduce a new version of the reach in the Riemannian setting inspired by the cut locus (see Definition~\ref{definition:ReachCutLocus}).

The second generalization lies in the types of sets we study --- we consider sets of positive reach and manifolds of positive reach. Sets of positive reach need not be manifolds --- in fact, they can have varying dimensions (see for example Figure~\ref{fig:set_covering}).
Manifolds with positive reach are $C^{1,1}$ smooth\footnote{Topologically embedded manifolds with positive reach are $C^{1,1}$ embedded \cite{Federer, lytchak2004geometry, lytchak2005almost, rataj2019curvature, StructureRataj}.}, i.e., differentiable with Lipschitz derivative. This is a significantly larger family of sets in comparison to $C^2$ manifolds with positive reach, considered by Niyogi, Smale, and Weinberger.

As in the work of Niyogi, Smale, and Weinberger, our settings consist of a set (or a manifold) $\Su$ of positive reach and its sample $P$. We distinguish two sample quality parameters --- sample density $\varepsilon$ and sample noisiness $\delta$, which we encode using one-sided Hausdorff distances between $P$ and $\Su$. We provide explicit conditions on $\varepsilon$ and $\delta$, under which there exists a parameter $r$ such that the union of balls of radius $r$ centred at the sample $P$ deformation-retracts to $\Su$. This result expands on the work of Niyogi, Smale, and Weinberger, who considered the cases $\delta=0$ and $\delta = \varepsilon$ only, and only achieved tight bounds in the latter case (see Figure~\ref{fig:graphCcr}).

Furthermore, given a set of positive reach $\Su$ and its sample $P$, we identify an interval of radii $r$ (equation \eqref{EQ:InvervalrSetPosReach}) for which the union of balls of radius $r$ centred at the sample $P$ deformation-retracts to $\Su$. Thus, we provide a \emph{guarantee} for a successful homotopy inference of the set $\Su$ from the sample $P$. Moreover, we show that for a specific choice of $\Su$ and $P$ (see Propositions \ref{prop:counterexample_set}, \ref{prop:counterexample_mfld}, \ref{prop:counterexample_set_Riemann}, and \ref{prop:counterexample_mfld_Riemann}), the homotopy of $\Su$ is not inferrable from $P$ if our conditions on $\varepsilon$ and $\delta$ are not satisfied, proving that our bounds are, in terms of $\varepsilon$ and $\delta$, tight.

\begin{figure}[!h]
  \begin{center}
    \def\svgwidth{0.495\linewidth}
\begingroup%
  \makeatletter%
  \providecommand\color[2][]{%
    \errmessage{(Inkscape) Color is used for the text in Inkscape, but the package 'color.sty' is not loaded}%
    \renewcommand\color[2][]{}%
  }%
  \providecommand\transparent[1]{%
    \errmessage{(Inkscape) Transparency is used (non-zero) for the text in Inkscape, but the package 'transparent.sty' is not loaded}%
    \renewcommand\transparent[1]{}%
  }%
  \providecommand\rotatebox[2]{#2}%
  \newcommand*\fsize{\dimexpr\f@size pt\relax}%
  \newcommand*\lineheight[1]{\fontsize{\fsize}{#1\fsize}\selectfont}%
  \ifx\svgwidth\undefined%
    \setlength{\unitlength}{223.23493903bp}%
    \ifx\svgscale\undefined%
      \relax%
    \else%
      \setlength{\unitlength}{\unitlength * \real{\svgscale}}%
    \fi%
  \else%
    \setlength{\unitlength}{\svgwidth}%
  \fi%
  \global\let\svgwidth\undefined%
  \global\let\svgscale\undefined%
  \makeatother%
  \begin{picture}(1,0.79824365)%
    \lineheight{1}%
    \setlength\tabcolsep{0pt}%
    \put(0,0){\includegraphics[width=\unitlength,page=1]{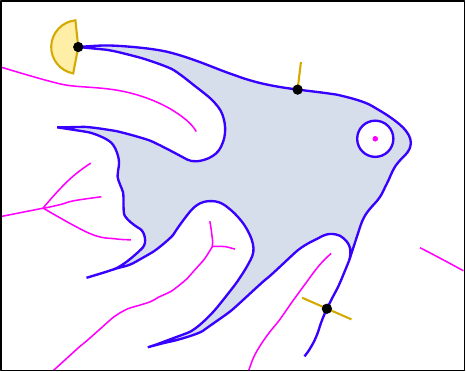}}%
    \put(0.18568681,0.73956671){\makebox(0,0)[lt]{\lineheight{1.25}\smash{\begin{tabular}[t]{l}$p_0$\end{tabular}}}}%
    \put(0.65869666,0.68499685){\makebox(0,0)[lt]{\lineheight{1.25}\smash{\begin{tabular}[t]{l}$p_1$\end{tabular}}}}%
    \put(0.65086666,0.38710738){\makebox(0,0)[lt]{\lineheight{1.25}\smash{\begin{tabular}[t]{l}$p_2$\end{tabular}}}}%
    \put(0.77375199,0.11942352){\makebox(0,0)[lt]{\lineheight{1.25}\smash{\begin{tabular}[t]{l}$p_3$\end{tabular}}}}%
    \put(0,0){\includegraphics[width=\unitlength,page=2]{fish-like-shape-info.pdf}}%
  \end{picture}%
\endgroup%
\hfill
    \def\svgwidth{0.495\linewidth}
    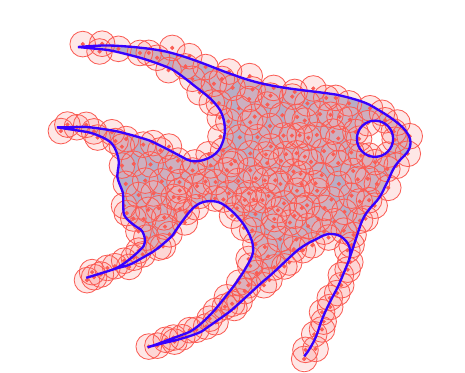
	\end{center}
	
	\caption{{\small Left: A fish shaped set $\Su$ of positive reach (in
			blue). Its medial axis (in purple) is at a positive
			distance. For $0 \leq i \leq 3$, we also represent the
			normal cone of $p_i$ with respect to $\Su$ (after an
			intersection with a small disk and a translation to $p_i$). The normal cone of the point $p_2$ is $p_2$ itself. 
			Right:
			The set $\Su$ with a sample $P$ and a thickening of $P$. We see that the thickening has the same homotopy type as $\Su$.
      }
	}
	\label{fig:set_covering}
\end{figure}

\section{State-of-the-art}

\subsection{Sets of positive reach} 

Our extension of Niyogi, Smale, and Weinberger's result to sets of positive reach --- as well as improvement of their results on manifolds --- relies on the work of Federer \cite{Federer}, which Niyogi, Smale, and Weinberger have not cited. In particular, we use Federer's generalization of normal spaces to normal cones (see Figure \ref{fig:set_covering} (left)
for a pictorial introduction and Appendix~\ref{sec:Euclidean_definitions} for a full definition) and his different characterizations of the normal cone as a key building block.
We recall the relevant results from Federer's work in Appendix~\ref{sec:Euclidean_definitions}.
 
We note that the reach can be estimated from a sample \cite{EddieOptimal, EstReach, berenfeld2022estimating,cholaquidis2023universally,cotsakis2024computable}.   
 
Subsets of positive reach of Riemannian manifolds were studied extensively by Kleinjohann \cite{kleinjohann1980convexity,kleinjohann1981nachste} and Bangert \cite{bangert1982sets} in generalization of Federer's theory \cite{Federer} for subsets of Euclidean space. 
Boissonnat and Wintraecken investigated yet another definition of the reach for subsets of Riemannian manifolds in \cite{ReachSubmanifolds}.

\subsection{Homotopy learning}

For some particular cases, the best previously known bounds on the distance between a manifold (or a set) of positive reach and its sample that guarantee successful homotopy inference, can be found in \cite{ATTALI2013448} and \cite{niyogi2008}.
Attali \emph{et al.} \cite{ATTALI2013448}, Chazal \emph{et al.} \cite{chazal2009sampling}, and Kim \emph{et al.} \cite{kim2019homotopy} expanded homotopy learning to 
even more general{ subsets of Euclidean space, such as subsets with positive $\mu$-reach}. 
Their proofs are, however, different from ours, more involved, and {their bounds are not shown to be tight.}

\subsection{Manifold and stratification learning} 

Although this article focuses on homotopy learning, our work should also be seen as part of recent developments in manifold learning \cite{EddieManSTOC, aamari2018stability, fefferman2018fitting, fefferman2019fitting, fefferman2020reconstruction, sober2020manifold}. 
The goal of this field is to reconstruct a manifold from a `reasonable' sample lying on or near it --- at least up to a homeomorphism, but usually an ambient isotopy. 

At the moment work is ongoing to expand {this strategy} to more general spaces --- see {for example the work of Aamari \emph{et al.}} \cite{EddieBoundary} on manifolds with boundary. 

Although inferring the homotopy of a manifold is simpler than manifold {learning}, the 
sets 
we consider are more general than manifolds or manifolds with boundary. The extension of 
learning 
from subsets of Euclidean space to subsets of Riemannian manifolds also departs from the usual track. We are only aware of one 
work in computational geometry and topology which operates  
within this context, namely \cite{chazal2013persistence}. {These are the first steps in the developing field of stratification learning.} Homotopy inference in the hyperbolic space was considered in \cite{antasik2022sampling}.

\section{Contribution} 
\subsection{Subsets of Euclidean space}
Let $\M$ denote a manifold of positive reach, $\Su$ a set of positive
reach and let $P$ be a sample.  All sets are assumed to be compact
unless stated otherwise. We denote the reach of a set $\mathcal{X}$ by
$\rch(\mathcal{X})$ and let $\reach$ be a non-negative real number such
that $\reach \leq \rch(\Su)$ (resp. $\reach \leq \rch(\M)$). 

We denote the bound on the one-sided Hausdorff distance\footnote{We recall that the one sided Hausdorff distance from $X$ to $Y$, denoted by $d_H^o(X;Y)$, is the smallest $\rho$ such that {$Y$ is covered by the union of balls of radius $\rho$ centred at $X$, that is, $Y \subseteq \bigcup_{x \in X} B(x,\rho)$.} } from $P$ to $\Su$ (resp. $\M$) by $\varepsilon$, and the one-sided Hausdorff distance from $\Su$ (resp. $\M$) to $P$ by $\delta$.

{In this article we establish} 
conditions on $\varepsilon$ and $\delta$ which, if satisfied, guarantee the existence of a radius $r>0$ such that the union 
of balls of radius $r$ centred at the sample $P$ deformation-retracts onto $\M$ (resp. $\Su$). 
{The set of pairs $(\varepsilon,\delta)$ that satisfy these conditions is depicted in Figure \ref{fig:graphCcr} {on the left}. The precise conditions are given in Propositions \ref{theorem:HomotopyNoiselessPositiveReach} and \ref{theorem:DeformRetractsTheoremForManifolds}. } 

\begin{figure}[!h]
\begin{center}
\includegraphics[width=0.70\textwidth]{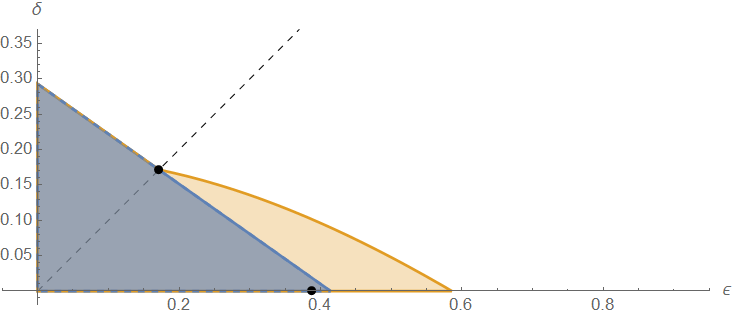}
\includegraphics[width=0.70\textwidth]{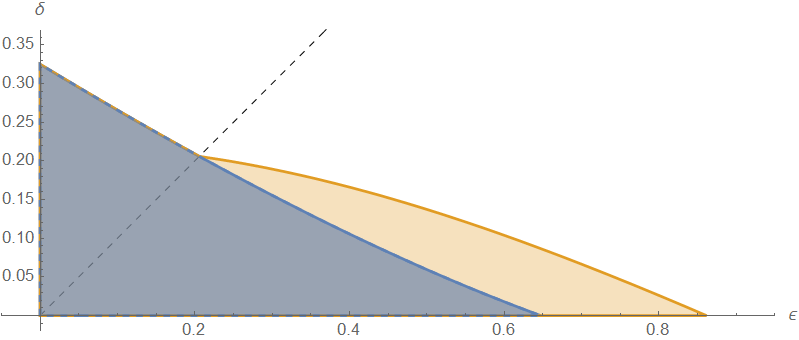}
\includegraphics[width=0.70\textwidth]{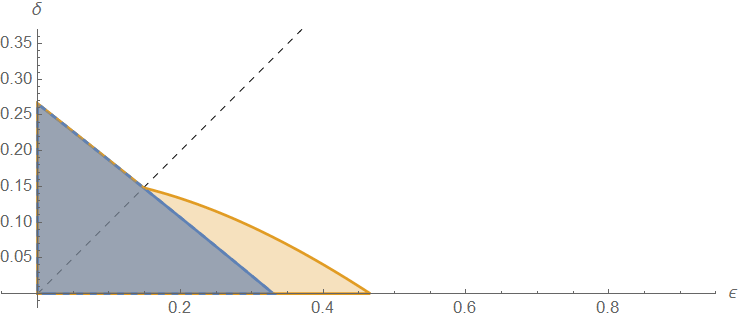}
\end{center}
\caption{\small
  The blue-gray region bounded by the blue dashed curve represents the
  set of pairs $(\varepsilon,\delta)$ for which there exists a radius
  $r$ such that the union of balls of radius $r$ centred at $P$
  captures the homotopy type of a set of positive reach $\reach =
  1$. The equivalent region for a manifold of reach $\reach = 1$ is
  depicted in yellow and is a superset of the previous one. The two
  regions coincide above the diagonal $\delta = \varepsilon$.
  The bounds for the Euclidean setting are indicated on top, for an
  ambient manifold with positive curvature bound ($\curvlowbnd =+2$)
  in the middle, and for an ambient manifold with negative curvature
  bound ($\curvlowbnd =-2$) bottom.  In the top picture, the black
  points indicate the bounds that were known to Niyogi, Smale, and
  Weinberger.}
\label{fig:graphCcr}
\end{figure}

Distinguishing the two one-sided Hausdorff distances seems natural to the authors, because in measurements one would expect the measurement error $\delta$ (with the exception of some small number of outliers) to be often smaller than the sampling density $\varepsilon$. Similar assumptions seem to be common in the learning community, see e.g. \cite{levin1998approximation}. Niyogi, Smale, and Weinberger \cite{niyogi2008} also made similar assumptions on the support of the measure from which they sampled.

We only consider samples for which we have precise bounds on $\varepsilon$ and $\delta$. In \cite{niyogi2008}, the authors also consider a setting where the point sample is drawn from a distribution centred on the manifold. They still recover the homotopy type of the underlying manifold with high probability. Our results can be applied to improve the bounds also in this context. However, we have not discussed this in detail, since combining both results is straightforward.

We stress that in \cite{boissonnat2018geometric, niyogi2008}, and \cite{WANG2020101606}, the authors use $\varepsilon/2$ instead of our $\varepsilon$. We also stress that $\varepsilon$ and $\delta$ have precisely opposite meanings in \cite{kim2019homotopy} compared to this paper.

Our conditions on $\varepsilon$ and $\delta$ 
are optimal for sets of dimension at least $2$ in the following sense: if the  conditions are not satisfied, 
we can construct a set of positive reach $\Su$ (resp. manifold $\M$) and a sample $P$, such that
{there is no $r\geq 0$ for which }the union of balls of radius $r$ centred at $P$ would have the same homology as $\Su$ (resp. $\M$). These constructions are explained in Section~\ref{sec:Euclidean_tightness}.

We would like to emphasize that for noiseless samples, (that is, when $\delta=0$,)  both the constant $\left(\sqrt{2}-1\right)$ (for general sets of positive reach), and the constant $\left(2-\sqrt{2}\right)$ (for manifolds) compare favourably with the previously best known constant  $\tfrac{1}{2} \sqrt{\tfrac{3}{5}}$ from \cite{niyogi2008}  for manifolds.\footnote{It should be noted that in \cite{niyogi2008} $r$ was not considered as a variable, but set equal to $2 \varepsilon$, which (at least partially) explains the suboptimal result in that paper.} 

In Proposition 7.1 of \cite{niyogi2008}, one encounters the condition $\varepsilon < (3-\sqrt{8})  \reach$ for a particular case of the setting we consider,
namely when the sampling condition is expressed through an upper bound $\varepsilon$ on the Hausdorff distance ($\delta= \varepsilon$ in our setting).
The same constant $3-\sqrt{8}$ appears independently in  \cite[Theorem 4]{attali:hal-00427035} for general sets of positive reach.
Our results (Propositions \ref{prop:counterexample_set} and \ref{prop:counterexample_mfld}) show that this bound is optimal when $\delta= \varepsilon$, both for general sets of positive reach and for manifolds.

To contrast the two related results in \cite{niyogi2008}, for $\delta= 0$ and $\delta= \varepsilon$ respectively, with our bounds, we portray them as black dots in Figure \ref{fig:graphCcr}.

Homotopy reconstruction of manifolds with boundary has been studied in \cite[Theorem 3.2] {WANG2020101606}, assuming lower bounds on both the reach of the manifold and the reach of its boundary.
We also improve on this result by treating a manifold with boundary as a particular case of a set of positive reach, while 
our bounds only depend on the reach of the set itself and not the one of its boundary.

{
\subsection{Subsets of Riemannian manifolds}
In the second part of this article we extend the homotopy reconstruction results to sets $\Su$
and manifolds $\M$ 
of positive reach embedded in a Riemannian manifold 
whose sectional curvatures\footnote{We recall (one of) the (equivalent) definition(s) of sectional curvatures of the Riemannian manifold {$\N$}:
For a point $p\in \N$ let $\Pi \subseteq T_p\N$ be a two dimensional plane in the  tangent space to $p$ at $\N$.
If $U\subseteq \Pi$ is a sufficiently small neighbourhood of $p$ in $\Pi$, then $\exp_p(U)$ is a surface. The Gauss curvature of this surface at $p$ is the sectional curvature of $\N$ at $p$ for the directions that span $\Pi$.} 
are bounded. 

Also in this Riemannian setting we find tight\footnote{When the curvature of the ambient manifold is positive we face a subtle issue because the manifold has a small volume. In that case, the meaning of optimality becomes less straightforward.} 
	bounds on the one-sided Hausdorff distances $\varepsilon$ and $\delta$ between $\Su$ (resp. $\M$) and its sample $P$. The set of pairs $(\varepsilon,\delta)$ that satisfy these conditions is depicted in Figure \ref{fig:graphCcr} (centre and right). The precise bounds are given in Propositions \ref{theorem:HomotopyPositiveReachRiemannian} and \ref{theorem:DeformRetractsTheoremForManifolds_Riemann}.

The main pillar of this part of our work is comparison theory. We recall the most essential definitions and results in Appendix \ref{sec:RecapToponogov}, and refer to \cite{berger2003panoramic, buser1981, chavel2006riemannian, cheeger2008comparison, Gromoll,  Karcher2} for further reading. 

For the extension to the Riemannian setting we also formulate a new generalization of the reach. To establish some of its properties, we use results on the gradient of the distance function \cite{albano2013singular}, see also \cite{LIEUTIERhomotopytype}. These results in turn require non-smooth analysis \cite{Clarke1990} and semi-concave functions \cite{albano1999structural}. We refer to Appendix \ref{section:CutLocusIsSingularSet} for discussion. 

In computer vision, many papers have argued in favour of using Riemannian manifolds as the main setting without embedding the Riemannian manifold in Euclidean space. 
In particular, symmetric positive definite matrices and Grassmannians form the natural stage for some data~\cite{vemulapalli2015riemannian, zhu2018towards}. Symmetric positive definite matrices occur as diffusion tensors \cite{pennec2006riemannian} (used in e.g. magnetic resonance imaging), in image segmentation \cite{goh2008clustering, rathi2007segmenting}, and in texture classification \cite{tuzel2006region}, while Grassmanians are used in image matching and recognition \cite{hamm2008grassmann, harandi2011graph}. {Although it is possible to embed these manifolds in Euclidean space, it} 
 is not natural and would increase the dimensionality significantly. 
  In \cite{zhang2016efficient}, time-series obtained from observations of dynamical systems are encoded as positive semi-definite matrices, produced by forming Hankel matrices and taking their Gram matrices. Thus, the problem of analysing time-series data is transformed into the problem of analysing point set data on a Riemannian manifold, namely the one formed by semi-positive definite matrices.

\section{Results for subsets of the Euclidean space}\label{sec:OverviewBoundsEuclideanCase}
\subsection{Setting}
We denote the closed ball in Euclidean space centred at a point $p$ with radius $r$ by $B(p,r)$. 

	\begin{tcolorbox} The thickening of a set $A \subseteq \R^d$ by parameter $r>0$ is denoted by $A^{\boxplus r}$, that is, 
		\[ A^{ \boxplus r} := \bigcup_{a\in A} B(a,r).\] 
	\end{tcolorbox}

\begin{remark}
	We use the notation $A^{ \boxplus r}$ to remind the reader of the Minkowski sum. It is indeed true that in $\R^d$, $A^{ \boxplus r} = A\oplus B(0,r)$. However, the above notation is also well-defined for subsets of manifolds, whereas the Minkowski sum is not.
\end{remark}

While working with subsets of the Euclidean space (Section \ref{sec:OverviewBoundsEuclideanCase} and Appendix \ref{sec:Euclidean_setting}) we assume the following:
\begin{tcolorbox} 
\begin{restatable}{assumption1}{assumptionEuclideanSetting}
\label{assumption}
	We work with a closed set $\Su \subseteq \R^d$ with positive reach $\rch(\Su)$, and let $\reach>0$ be a constant satisfying $\reach \leq \rch(\Su)$. 
	Furthermore, we consider a set $P\subseteq \R^d$, such that the one-sided Hausdorff distance from $P$ to $\Su$ is at most $\delta$, and the one-sided Hausdorff distance from $\Su$ to $P$ is at most $\varepsilon$. 
	That is,
	\begin{equation*} 
		\Su \subseteq P^{\boxplus \varepsilon}  
		\qquad \text{and} \qquad P \subseteq \Su^{\boxplus \delta} . 
	\end{equation*} 
	We assume that $\delta, \varepsilon <\reach$.
	If the set $\Su$ is a submanifold of $\R^d$, we denote it by $\M$.
\end{restatable}
\end{tcolorbox} 
For most applications the assumption $\delta \leq \varepsilon$ seems natural, but we do not need this. However, when $\Su = \M$, we achieve better bounds when $\delta \leq \varepsilon$. See Remark \ref{remark.WhenDeltaGreaterThanEpsilon} for more details.

\subsection{The geometric argument}\label{sec:Euclidean_geometric_argument}

We show that if the thickening $P^{\boxplus r}  = \bigcup_{p \in P } B(p, r)$ 
covers a sufficiently large thickening of $\Su$ --- quantified by parameter $\alpha$ --- and the parameter $r$ is not too big, $P^{\boxplus r} $  
deformation-retracts to $\Su$.

We start by recalling that the normal cone at a point $p$ of a set of positive reach is the set of directions such that if you move from $p$ in that direction the closest point projection will remain $p$. For a definition we refer to Definition \ref{def:4.3and4.4Fed}.  

\begin{theorem}\label{theorem:geometric_argument}
	Assume that a parameter $\alpha>0$ is small enough, so that the $\alpha$-neighbourhood $\Su^{\boxplus \alpha} $ of the set $\Su$ is contained in $P^{\boxplus r} $. In other words,
	\begin{align}  
	\Su^{\boxplus \alpha}  \subseteq P^{\boxplus r}. 
	\label{eq:Tube} 
	\end{align} 
If, moreover, 
\begin{equation}\label{equation:R2TooSmallToCNormalLineS}
r^2 \leq  (\reach -  \delta)^2 - (\reach - \alpha)^2,
\end{equation}
then, for any point $q\in\Su$, the intersection $(q+\Nor(q, \Su)) \cap B(q,\reach) \cap  P^{\boxplus r} $ of the normal cone $q+\Nor(q, \Su)$, the ball $B(q,\reach)$, and the union of balls $P^{\boxplus r}$, is star-shaped, with the point $q$ as its `centre'. Furthermore, $P^{\boxplus r}$ deformation-retracts onto $\Su$ along the closest point projection.
\end{theorem}

\begin{remark}
The statement of Theorem~\ref{theorem:geometric_argument} does not use the hypothesis $\Su \subseteq P^{\boxplus \varepsilon}$ from the Universal Assumption \ref{assumption}. 
\end{remark}

We refer to Figure \ref{fig:OverviewProof} for a pictorial overview of the proof of Theorem \ref{theorem:geometric_argument}. Further in the paper, we express the parameter $\alpha$ in terms of $r$ and the quality parameters $\varepsilon$ and $\delta$. The expression differs depending on whether $\Su$ is a set or a manifold of positive reach. Inserting the appropriate expression into bound~\eqref{equation:R2TooSmallToCNormalLineS} yields the final bounds on $\varepsilon$ and $\delta$ (see Propositions \ref{theorem:HomotopyNoiselessPositiveReach} and \ref{theorem:DeformRetractsTheoremForManifolds}).

\begin{figure}[h!]
	\centering
	\begin{subfigure}[b]{0.445\textwidth}
		\centering
		\includegraphics[width=\textwidth]{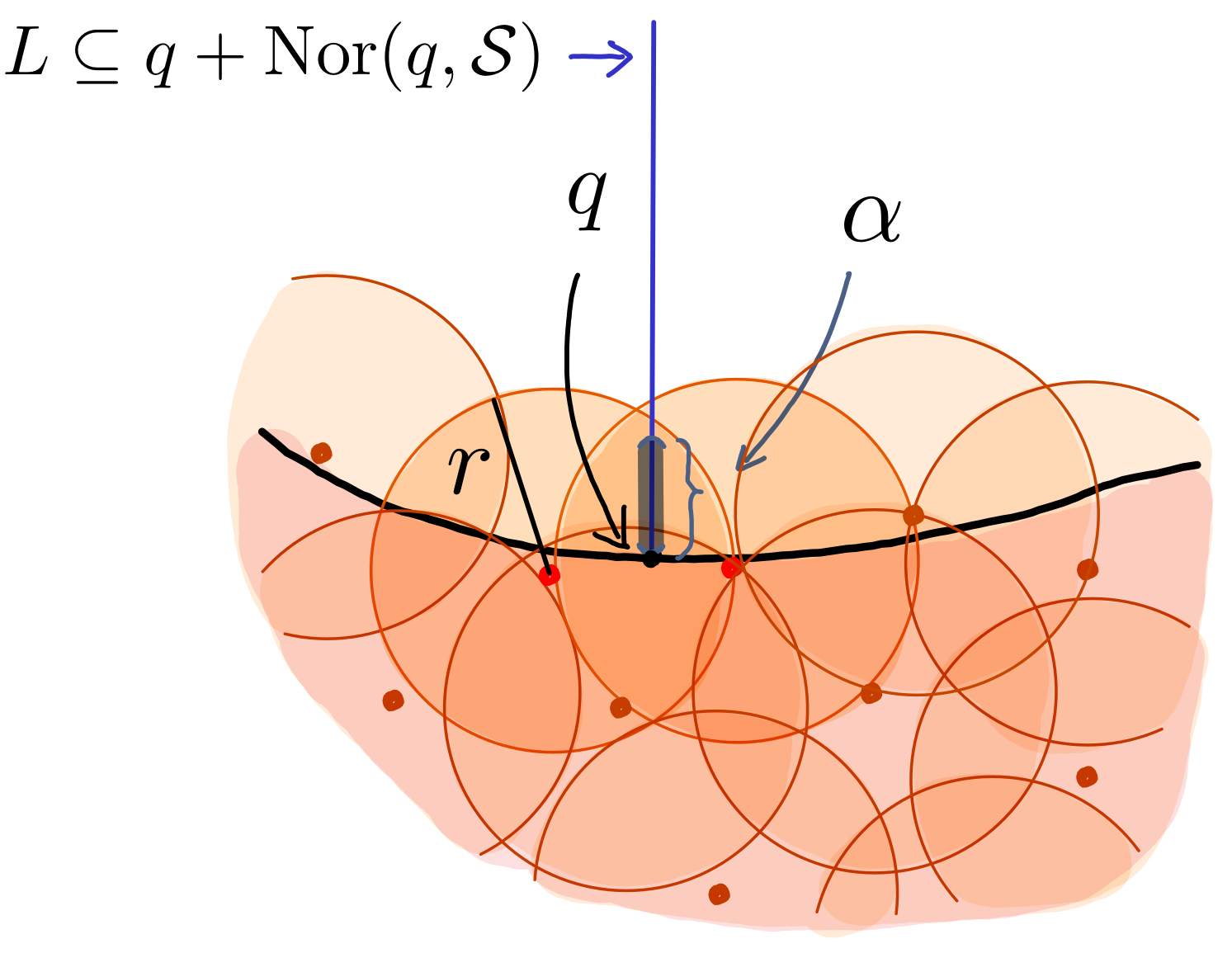}
		\subcaption{Any point in $q+\Nor(q,\Su)$ a distance less than $\alpha$ from $\Su$ is covered by $P^{\boxplus r}$.
		}
		\label{fig:2a}
	\end{subfigure}
	\hfill
	\begin{subfigure}[b]{0.505\textwidth}
		\centering 
		\includegraphics[width=0.92\textwidth]{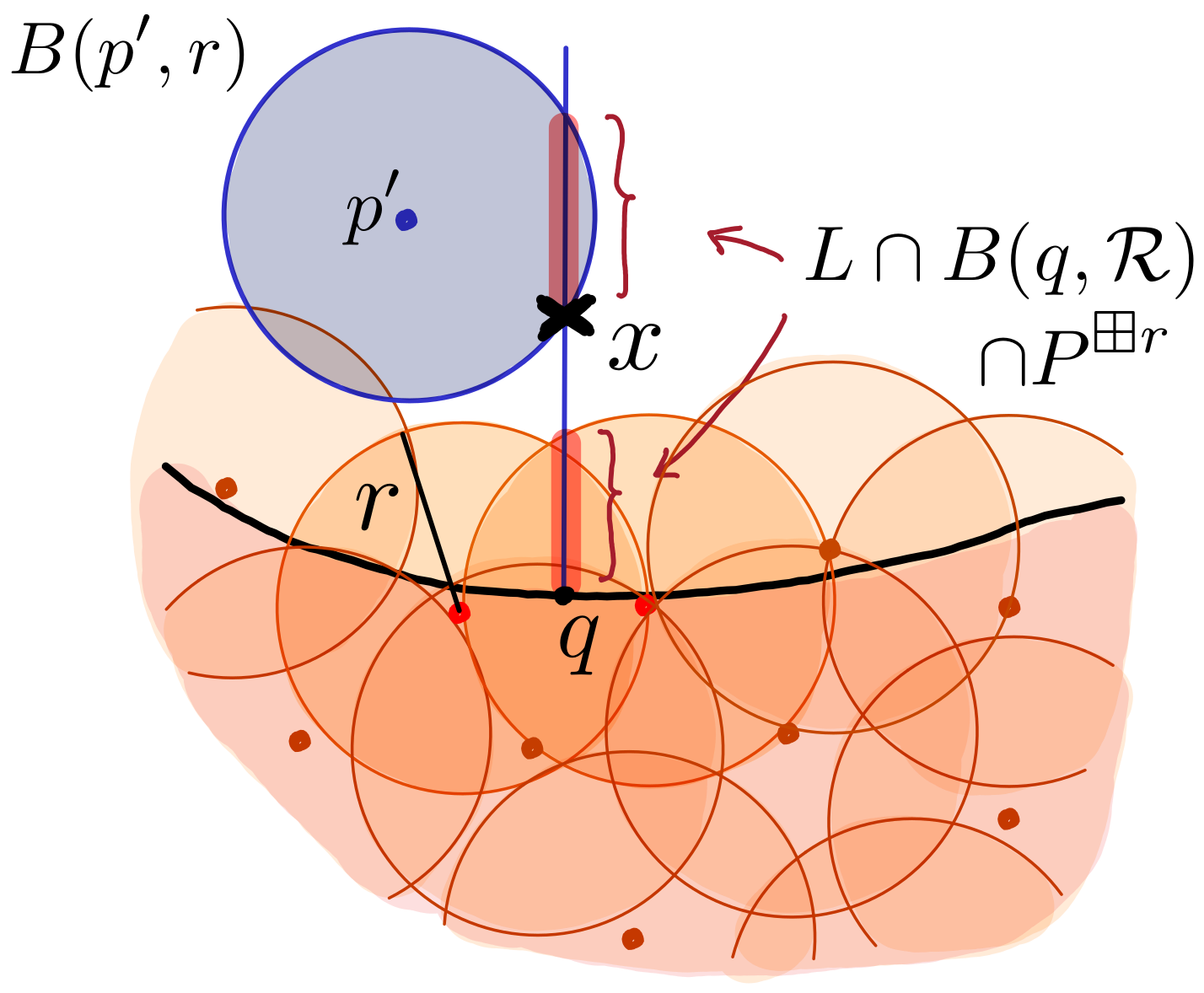}
		\subcaption{If $(q+ \Nor(q,\Su)) \cap P^{\boxplus r}$ is not star-shaped there exists a point $x$ where the segment $L$ reenters a ball $B(p',r)$ (in blue) in $P^{\boxplus r}$ after having left $P^{\boxplus r}$ closer to $q$. 
	}
		\label{fig:2b}
	\end{subfigure}
	\vskip\baselineskip
	\begin{subfigure}[b]{0.445\textwidth} 
		\centering 
		\includegraphics[width=\textwidth]{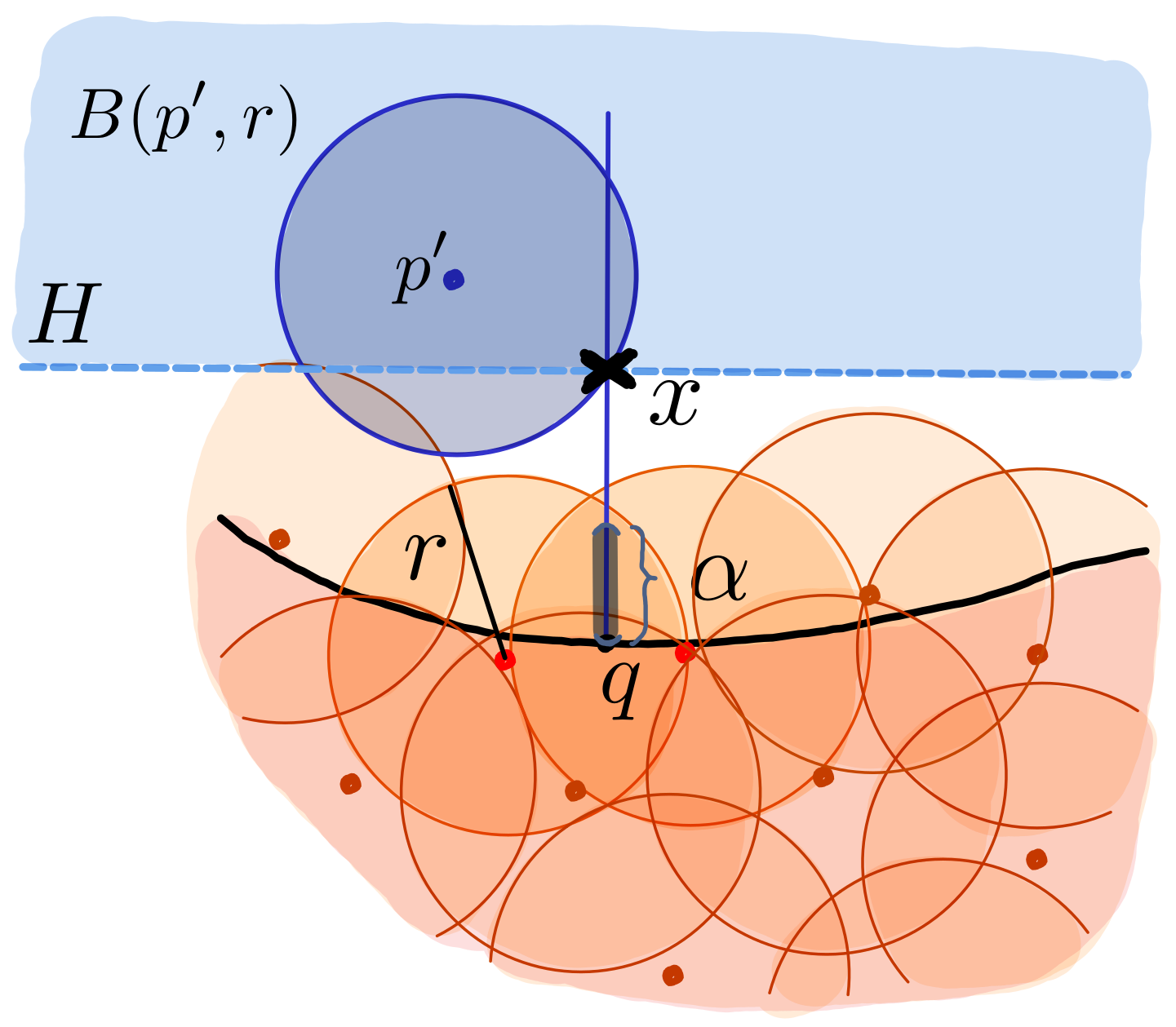}
		\subcaption{The centre $p'$ of the ball $B(p',r)$ lies inside the half-space~$H$. The half-space~$H$ lies at least a distance $\alpha$ from $q$.}
		\label{fig:2c}
	\end{subfigure}
	\hfill
	\begin{subfigure}[b]{0.505\textwidth}   
		\centering 
		\includegraphics[width=\textwidth]{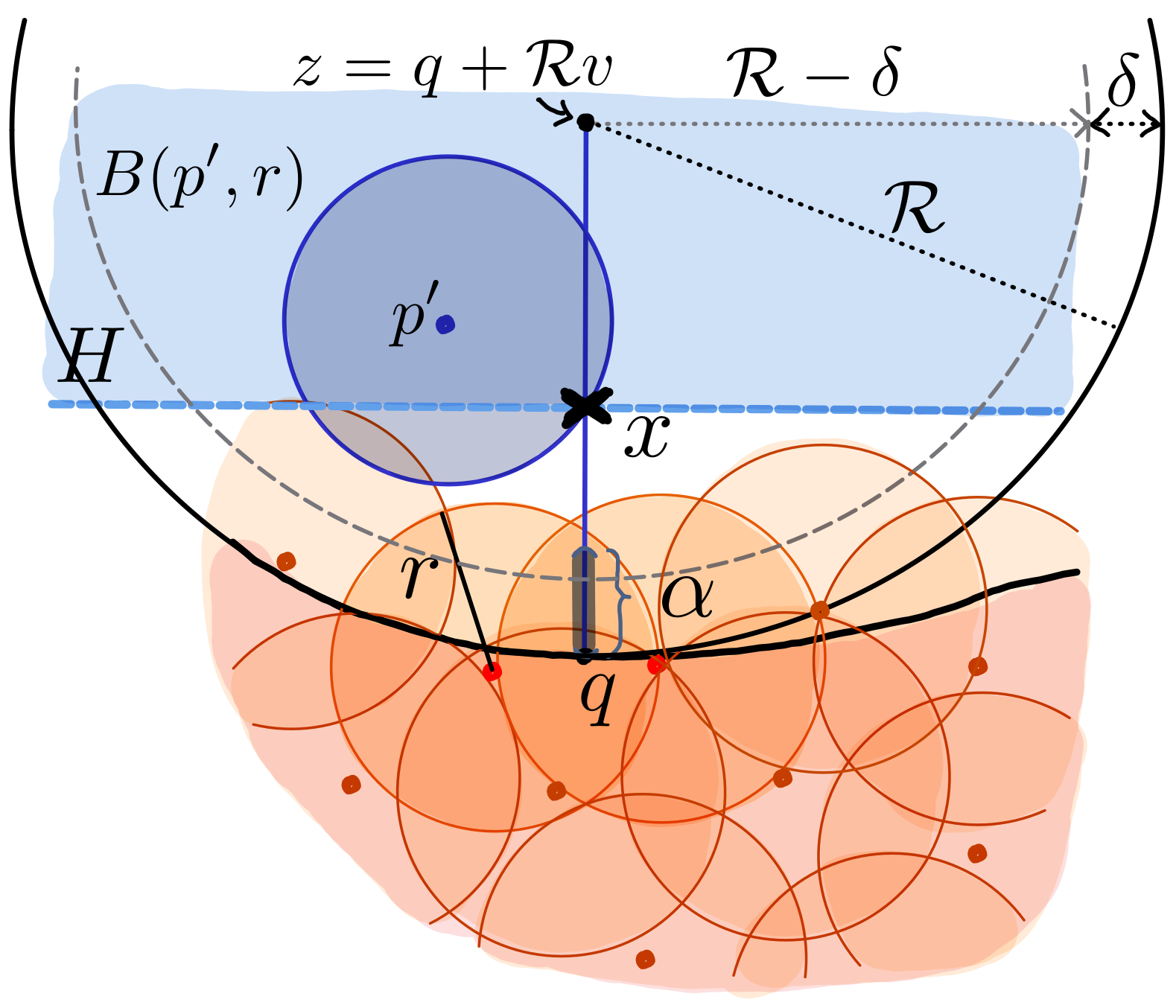}
		\subcaption{The ball of radius $\reach$ is `tangent' to the set~$\Su$, thus it cannot contain any point of~$\Su$ in its interior. {Since the distance between $p'$ and $\Su$ is bounded by $\delta$, $p'$ has to lie outside of the ball of radius $\reach-\delta$.}
		This contradicts the fact that $p'$ lies in the half-space~$H$ and is not too far from the normal space $q+\Nor(q,\Su)$.}%
		\label{fig:2d}
	\end{subfigure}
	\caption[]
	{\small A pictorial overview of the proof. The pink shaded region represents a part of the set $\Su$, the union of balls $P^{\boxplus r}$ 	is coloured orange. The thickened blue segment shows those points of the segment $L$ that lie a distance less than $\alpha$ from $\Su$. Per assumption, this segment is contained in the union of balls $P^{\boxplus r}$. } 
	\label{fig:OverviewProof}
\end{figure}
	
\begin{proof}[Proof of Theorem \ref{theorem:geometric_argument}]
	We prove the claim by contradiction. 
	For any point $q\in\Su$, the set $(q+\Nor(q, \Su))\cap \left(\Su^{\boxplus \alpha } \right)$ is contained in the union of balls $P^{\boxplus r}$. In Figure~\ref{fig:2a}, we illustrate this for the case where the set $q+\Nor(q, \Su)$ consists of one ray.
	Assume that there exists a point $q \in \Su$ and a vector $v \in \Nor(q, \Su)$, with $\|v\|=1$, such that the intersection 
	of $P^{\boxplus r}$	with the segment 
	\[
	L\defunder{=} \{ q+ \lambda v \mid \lambda \in [0, \reach)\}
	\]
	consists of several connected components (as illustrated in Figure \ref{fig:2b}). Thanks to Equation~\eqref{eq:Tube}, the connected component that contains $q$ has length at least $\alpha$. Let $x$ be first point along $L$, seen from $q$, lying inside a connected component of $ \left(P^{\boxplus r} 
	\right) \cap L$ that does not contain $q$. 
	Then $x$ lies at the intersection of the segment $L$ and a ball $B(p', r)$, 
	with $p'\in P$. {We have $\|x-q\| \geq \alpha$.}
	Furthermore, the point $p'$ is 
	contained in the open half-space $H$ orthogonal to the vector $v$, that does not
	contain $q$, and whose boundary contains $x$. We stress that if $p'$ lies on the boundary of $H$ then the line $L$ is tangent to the sphere $\partial B(p', r)$, which is still compatible with star-shapedness.
	The situation is illustrated in Figure \ref{fig:2c}.

{Let $z=q+\reach v$ be the open endpoint of $L$.} Since, by Theorem~\ref{Fed4.8.12} (\cite[Theorem 4.8 (12)]{Federer}), the intersection $\Su \cap  B(z, \reach)^\circ$ is empty {and the distance between $p'$ and $\Su$ is bounded by $\delta$}, we know that $p' \notin B(z, \reach- \delta)^\circ$. Thus,
	\[
	p' \in  A \defunder{=} H \cap (\R^d \setminus B(z, \reach- \delta)^\circ).
	\]
	\begin{figure}[!h]
		\begin{center}
			\includegraphics[width=0.55\textwidth]{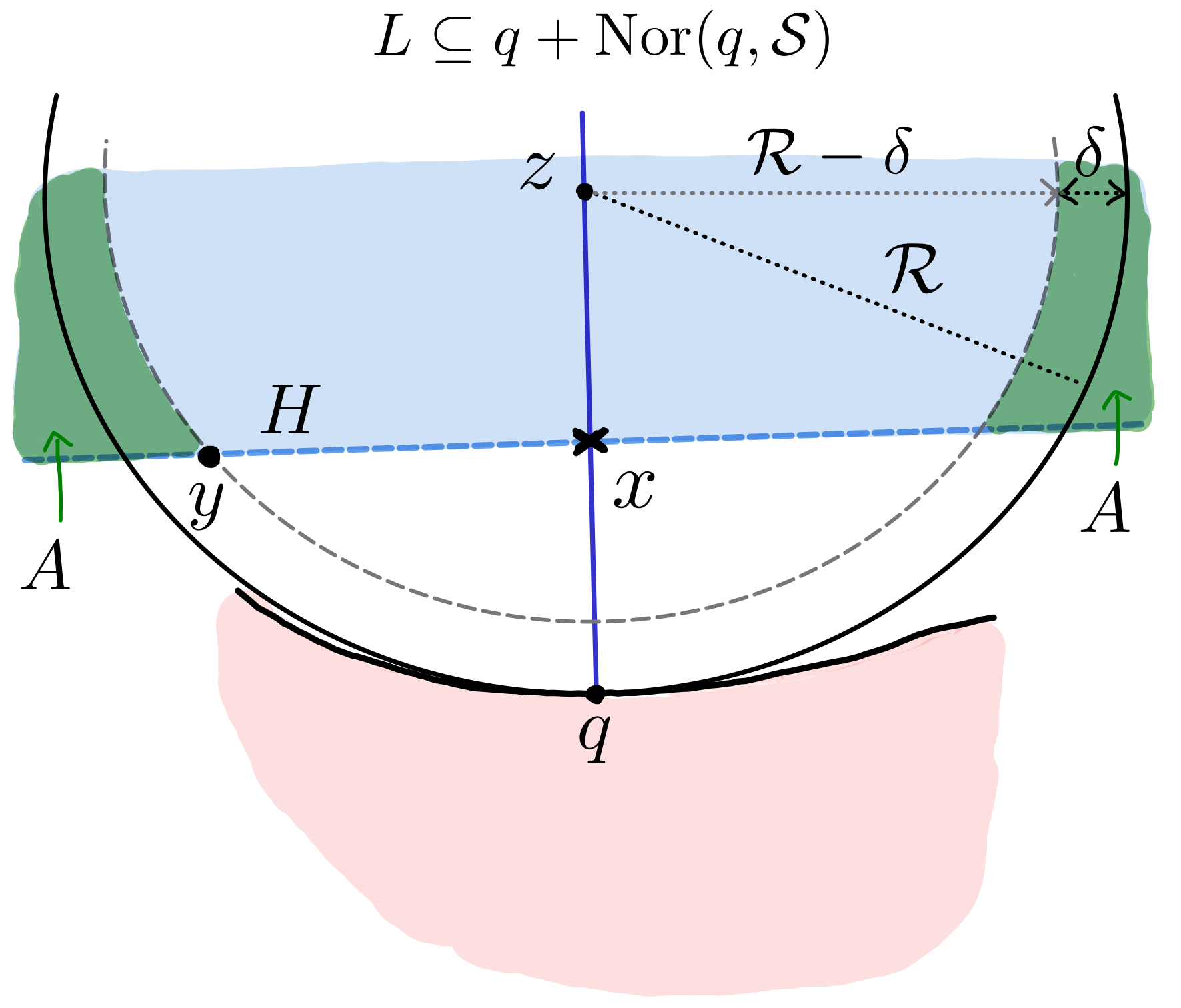}
		\end{center}
		\caption{
			The centre of the ball creating a new connected component along one direction in the normal cone $q+ \Nor(q,\Su)$ (in blue) is constrained to belong to the set $A$ (in green). The set $\Su$ is coloured pink, the half-plane $H$ in light blue.
		} 
		\label{fig:OverviewProofNoisy}
	\end{figure}
 The sphere $\partial B(z,\reach-\delta)$ has a
      non-empty intersection with the plane $\partial H$. Indeed, the
      sphere passes through point $q + \delta v$ which does not belong
      to $H$ while its centre $z$ belongs to $H$; see Figure
      \ref{fig:2d}. We can thus pick a point $y$ 
			in the
      intersection $\partial H \cap \partial B(z,\reach-\delta)$.  By the
    Pythagorean theorem, the minimal squared distance between $A$ and
    $L$ is:
   {
    \[
      \inf_{\substack{a \in A \\ \ell\in L}} \|a-\ell\|^2 = \|x-y\|^2 = \|z-y\|^2 - \left( \|z-q\| - \|x-q\| \right)^2 \geq (\reach -  \delta)^2 - (\reach - \alpha)^2,
    \]}
	as illustrated in Figure \ref{fig:OverviewProofNoisy}. Hence, if
	\begin{equation}
	r^2  \leq   (\reach -  \delta)^2 - (\reach - \alpha)^2,
	\tag{\ref{equation:R2TooSmallToCNormalLineS}} 
	\end{equation}
	the ball $B(p', r)$ does not intersect $L$. Therefore, $L \cap  (P^{\boxplus r} 
	) $ cannot have more than one connected component. The set $(q+\Nor(q, \Su)) \cap B(q,\reach) \cap  (P^{\boxplus r} 
	) $ is thus star-shaped with centre $q$.

Since $r$ satisfies Equation \eqref{equation:R2TooSmallToCNormalLineS},
we deduce that $\delta + r < \reach$, and thus
\[
P^{\boxplus r}  \subseteq \left(\Su^{\boxplus \reach}\right) ^\circ. \]
Thanks to this, the fact that the set $(q+\Nor(q, \Su)) \cap B(q,\reach) \cap  (P^{\boxplus r} 
) $ is star-shaped with centre $q$, and Theorem~\ref{Fed4.8.12}, the map
\begin{align*} 
\mathcal{H}: & [0,1]\times (P^{\boxplus r}) \to P^{\boxplus r},
\\
&(t, x) \mapsto (1-t) x + t \pi_{\Su} (x),
\end{align*}
is well-defined.

Furthermore, since $\Su$ has positive reach, then, thanks to Theorem~\ref{Fed4.8.8} (\cite[Theorem 4.8 (8)]{Federer}), the projection $\pi_{\Su}$ is (Lipschitz) continuous. Thus, the map $\mathcal{H}$ is a deformation retract from the union of balls $P^{\boxplus r}$ to the set $\Su$.	
\end{proof}
In Appendix \ref{app:Alternative_proof}, we provide an alternative proof of Theorem \ref{theorem:geometric_argument}, similar to an argument presented in \cite{chazal2008smooth}.

\subsection{Bounds on the sampling parameters}
Recall that throughout the paper we assume the Universal Assumption \ref{assumption}.
For sets of positive reach, we obtain the following bounds on the quality parameters $\varepsilon$ and $\delta$:
\begin{restatable}{proposition}{HomotopyNoiselessPositiveReach}
\label{theorem:HomotopyNoiselessPositiveReach}
If $\varepsilon$ and $ \delta$ satisfy 
\begin{equation}\label{equation:BoundOnR0}
\varepsilon + \sqrt{2} \, \delta \leq  (\sqrt{2}  - 1) \reach,
\end{equation}
there exists a radius $r>0$ such that the union of balls  {$P^{\boxplus r} = \bigcup_{p\in P} B(p,r)$} deformation-retracts onto $\Su$ along the closest point projection. In particular, $r$ can be chosen as:
\begin{equation} 
 r \in  \left  [ \frac{1}{2} \left(\reach+\varepsilon  -  \sqrt{\Delta}\right),  
\frac{1}{2} \left(\reach+\varepsilon +  \sqrt{\Delta}\right)  \right  ],
\label{EQ:InvervalrSetPosReach} 
\end{equation}
where $ \Delta=  2(\reach-\delta)^2 - (\reach+\varepsilon)^2 $.
\end{restatable}

\begin{restatable}{remark}{ExtendedInterval}\label{remark:ExtendedInterval}
The interval for $r$ as given in \eqref{EQ:InvervalrSetPosReach} can be slightly extended to 
  \begin{equation}
    \label{equation:intervalForPositiveReach}
    r \in  \left  [ \frac{1}{2} \left(\reach+\varepsilon  -  \sqrt{\Delta}\right),  
      \sqrt{\frac{1}{2} (\reach-\delta)^2 + \frac{1}{2} (\reach+\varepsilon) \sqrt{\Delta}}  \right  ],
  \end{equation} as we show in an alternative proof of Proposition \ref{theorem:HomotopyNoiselessPositiveReach} in Appendix~\ref{app:Alternative_proof}. It is not obvious that even this improved bound is tight.
\end{restatable}

If the set is a manifold, the bounds on $\varepsilon$ and $\delta$ can be improved as follows:
\begin{restatable}{proposition}{DeformRetractsTheoremForManifolds}
\label{theorem:DeformRetractsTheoremForManifolds}
	If $\varepsilon$ and $  \delta$ satisfy
\begin{equation}\label{equation:BoundOnEpsilon2}
(\reach - \delta)^2 - \varepsilon^2   \geq   \left(4\sqrt{2}  - 5\right) \reach^2
\end{equation}
and $\delta \leq \varepsilon$, there exists a radius $r>0$ such that the union of balls {$P^{\boxplus r} $} deformation-retracts onto $\M$ along the closest point projection. 
The radius $r$ can be chosen as in \eqref{Bounds_r_For_Manifold_Case}. 
\end{restatable}

\begin{figure}[h!]
\begin{subfigure}[t]{0.99\textwidth}
\centering 
\includegraphics[width=0.8\textwidth]{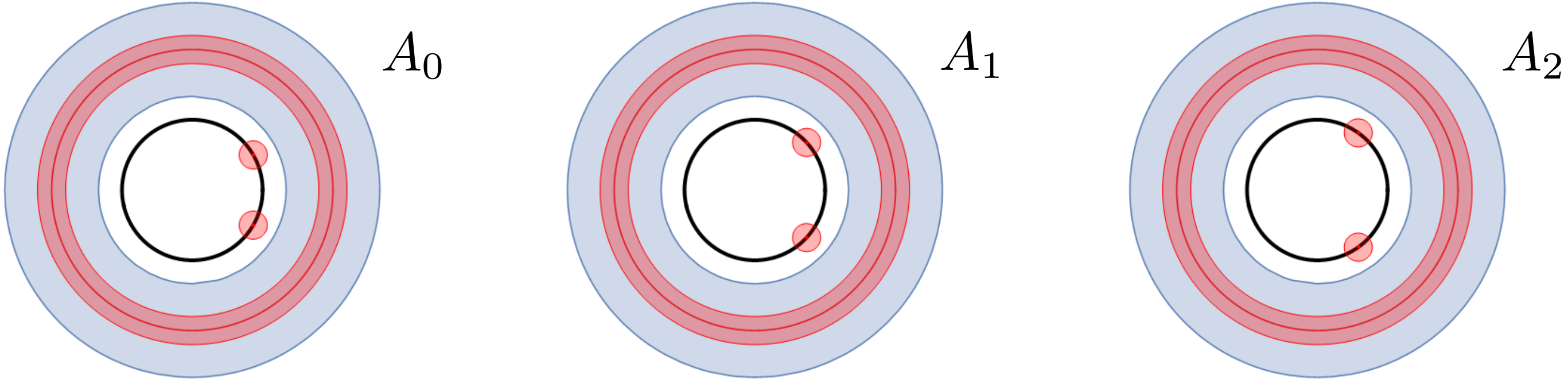}
\subcaption{
{At first, the thickening of the sample has three connected components per annulus. The thickening thus has three times as many connected components as the set $\Su$.}
}
\end{subfigure}
\newline
\begin{subfigure}[t]{0.99\textwidth}
\begin{center}
\includegraphics[width=0.8\textwidth]{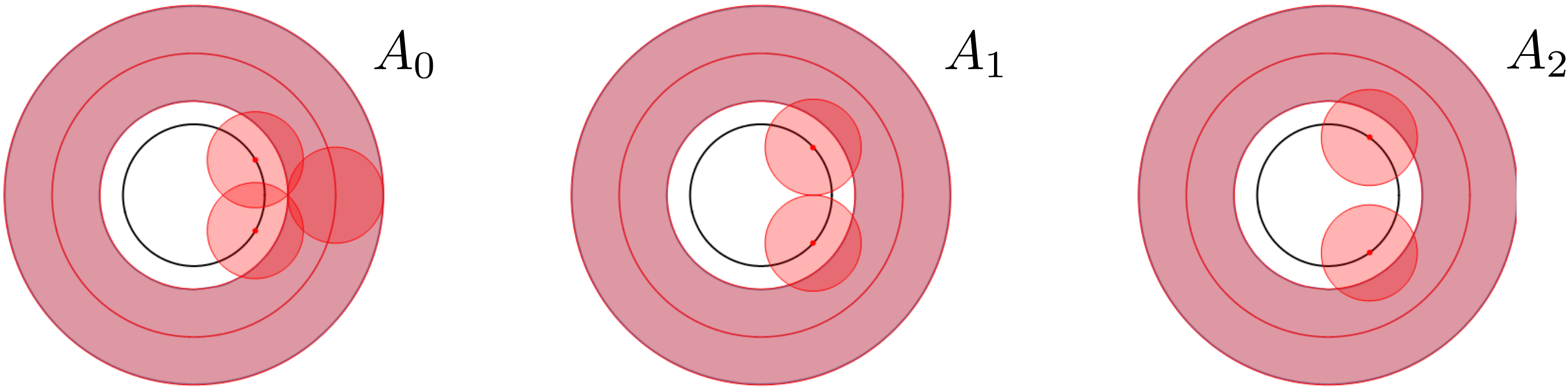}
\end{center}
\subcaption{
{	As the radius of the thickening grows, the connected components merge. However, at all times there exists an additional cycle at one of the annuli (annulus $A_1$ in this case).}
}
\end{subfigure}
\newline
\begin{subfigure}[t]{0.99\textwidth}
\begin{center}
\includegraphics[width=0.8\textwidth]{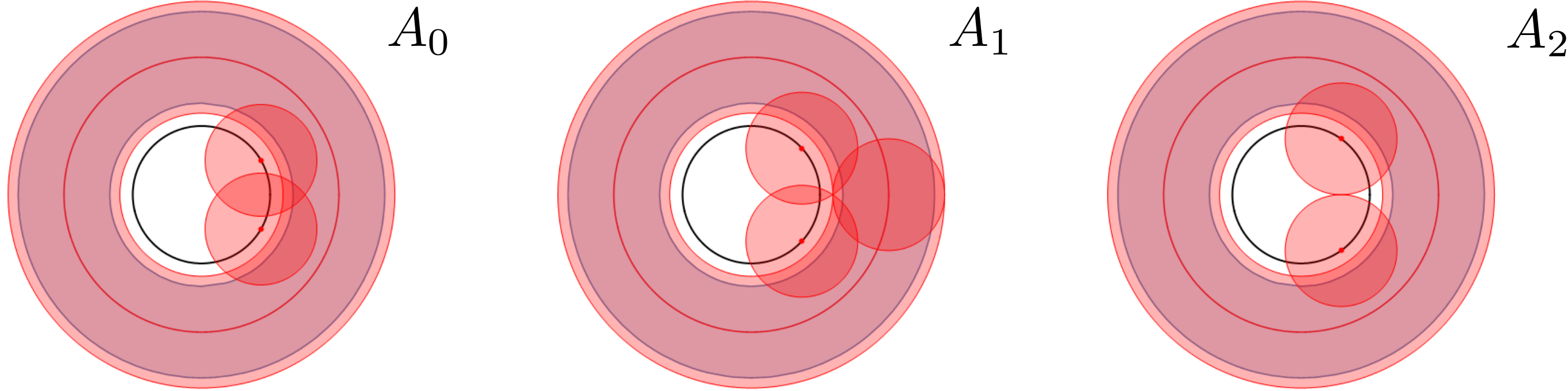}
\end{center} 
\subcaption{
{At the moment when the cycle at annulus $A_1$ vanishes, another cycle is formed at annulus $A_2$.}
}
\end{subfigure}
	\caption{A pictorial explanation of why $P^{\boxplus r}$ {never has} the 
	  homotopy type {of the set $\Su$}. We only depict three annuli in the sequence of $A_i$s.	The set $\Su$ is in blue, the sample $P$ in red, and the thickening {of $P$} in pink.
      {The black circles indicate the location of the two isolated sample points of $P$ associated to each annulus.}
    }
	\label{fig:tight_set_nbhdsTeaser}
\end{figure}

{Both in Propositions \ref{theorem:HomotopyNoiselessPositiveReach} and \ref{theorem:DeformRetractsTheoremForManifolds}, the interval for $r$ tends to $[0,\reach]$ as $\varepsilon$ and $\delta$ tend to zero. }

\subsection{Tightness of the bounds on the sampling parameters}
\label{sec:Euclidean_tightness}

Our sampling criteria for homotopy inference of sets of positive reach are tight in the following sense: 
\begin{restatable}{proposition}{counterexampleSet}
\label{prop:counterexample_set}
Suppose that the dimension $d$ of the ambient space $\R^d$ satisfies $d\geq 2$, and the one-sided Hausdorff distances $\varepsilon$ and $\delta$ fail to satisfy bound~\eqref{equation:BoundOnR0}. Then there exists a set $\Su$ of positive reach and a sample $P$ that satisfy Universal Assumption~\ref{assumption}, while the homology of the union of balls {$P^{\boxplus r}$}  
does not equal the homology of $\Su$ for any $r$.
\end{restatable}


We construct the set $\Su$ and the sample $P$ explicitly {in
  $\R^2$}.  The set $\Su$ consists of {a finite family of annuli
  $A_i$, the first three of which are depicted in Figure
  \ref{fig:tight_set_nbhdsTeaser}}. The sample $P$ is the union of a
circle and two points for every annulus.
In Figure \ref{fig:tight_set_nbhdsTeaser}, we illustrate that the thickening of the sample never captures the homotopy type of the set $\Su$.
The details of the construction and the proof 
{of Proposition \ref{prop:counterexample_set} are 
{provided} in Section \ref{sec:TightnessSetsOfPosReach}.

\begin{figure}[!ht]
  	\centering
\begingroup%
  \makeatletter%
  \providecommand\color[2][]{%
    \errmessage{(Inkscape) Color is used for the text in Inkscape, but the package 'color.sty' is not loaded}%
    \renewcommand\color[2][]{}%
  }%
  \providecommand\transparent[1]{%
    \errmessage{(Inkscape) Transparency is used (non-zero) for the text in Inkscape, but the package 'transparent.sty' is not loaded}%
    \renewcommand\transparent[1]{}%
  }%
  \providecommand\rotatebox[2]{#2}%
  \newcommand*\fsize{\dimexpr\f@size pt\relax}%
  \newcommand*\lineheight[1]{\fontsize{\fsize}{#1\fsize}\selectfont}%
  \ifx\svgwidth\undefined%
    \setlength{\unitlength}{261.98081533bp}%
    \ifx\svgscale\undefined%
      \relax%
    \else%
      \setlength{\unitlength}{\unitlength * \real{\svgscale}}%
    \fi%
  \else%
    \setlength{\unitlength}{\svgwidth}%
  \fi%
  \global\let\svgwidth\undefined%
  \global\let\svgscale\undefined%
  \makeatother%
  \begin{picture}(1,0.41384019)%
    \lineheight{1}%
    \setlength\tabcolsep{0pt}%
    \put(0,0){\includegraphics[width=\unitlength,page=1]{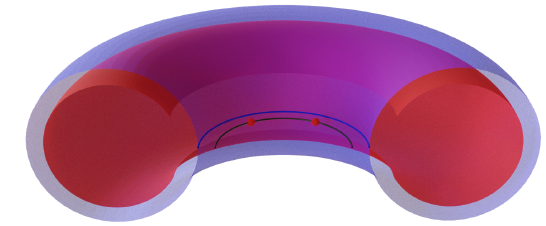}}%
    \put(-0.00260001,0.12808967){\color[rgb]{0.21568627,0,1}\makebox(0,0)[lt]{\lineheight{1.25}\smash{\begin{tabular}[t]{l}$T_i$\end{tabular}}}}%
    \put(0.1069469,0.12808967){\color[rgb]{1,0,0}\makebox(0,0)[lt]{\lineheight{1.25}\smash{\begin{tabular}[t]{l}$C_i$\end{tabular}}}}%
    \put(0.44699666,0.15100328){\color[rgb]{1,0,0}\makebox(0,0)[lt]{\lineheight{1.25}\smash{\begin{tabular}[t]{l}$p_i$\end{tabular}}}}%
    \put(0.56437818,0.14982875){\color[rgb]{1,0,0}\makebox(0,0)[lt]{\lineheight{1.25}\smash{\begin{tabular}[t]{l}$\tilde{p}_i$\end{tabular}}}}%
    \put(0.38436367,0.10229101){\makebox(0,0)[lt]{\lineheight{1.25}\smash{\begin{tabular}[t]{l}$C_i'$\end{tabular}}}}%
  \end{picture}%
\endgroup%

	\caption{
		The (half of the) torus $T_i$ depicted in blue; the sample --- the set $C_i$ and the points $p_i$ and $\tilde{p}_i$ --- in red.
		In black we indicate the circle $C'_{i}$ {on which the points $p_i$ and $\tilde{p}_i$ lie}. The closest point projection of this circle onto $\M$ is indicated in blue. 
	}
	\label{fig:TightManifold1NTeaser}
\end{figure}
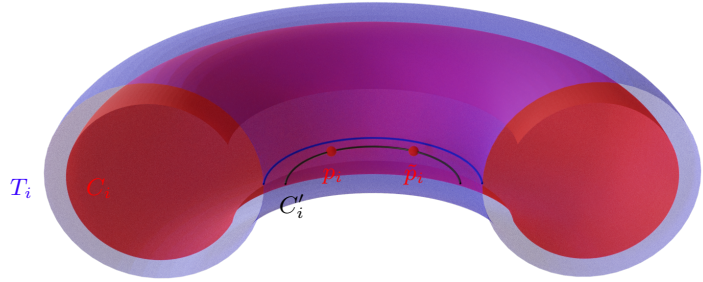


\begin{restatable}{proposition}{counterexampleMfld}
\label{prop:counterexample_mfld}
Suppose that the dimension $d$ of the ambient space $\R^d$ satisfies $d\geq 3$, the one-sided Hausdorff distances $\varepsilon$ and $\delta$ fail to satisfy bound~\eqref{equation:BoundOnEpsilon2}, and $\delta \leq \varepsilon$. Then there exists a manifold $\M$ of positive reach and a sample $P$ that satisfy Universal Assumption~\ref{assumption}, while the homology of the union of balls $P^{\boxplus r}  $ 
does not equal the homology of $\M$ for any~$r$.
\end{restatable}

We again construct the manifold $\M$ and the sample $P$ explicitly, {this time in $\R^3$}. The manifold $\M$ is {the union of a finite family of} tori {$T_i$}.
{The sample $P$ consists of one set $C_i$ and one pair of points $\{p_i,\tilde{p}_i\}$ for each torus $T_i$. The set $C_i$ is constructed by taking a copy of $T_i$, decreasing the minor radius and cutting out a part close to the axis of revolution.}
{We illustrate the manifold {$\M = \bigcup_{i} T_i$} together with the sample {$P = \bigcup_{i} C_i \cup \{p_i,\tilde{p}_i\}$} in Figure~\ref{fig:TightManifold1NTeaser}, {and sketch why the underlying homology is not captured in Figure \ref{Fig:CycleManifold}.} The proof of Proposition~\ref{prop:counterexample_mfld} as well as details on the construction are provided in Section~\ref{sec:optimality_mflds}.}

A video animating our construction has been submitted to the Media Exposition at Computational Geometry Week 2024 \cite{NSWVideoSoCG2024}.

\begin{figure}[h!]
  \begin{subfigure}[t]{0.99\textwidth}
    \def\svgwidth{0.99\linewidth}
    \centering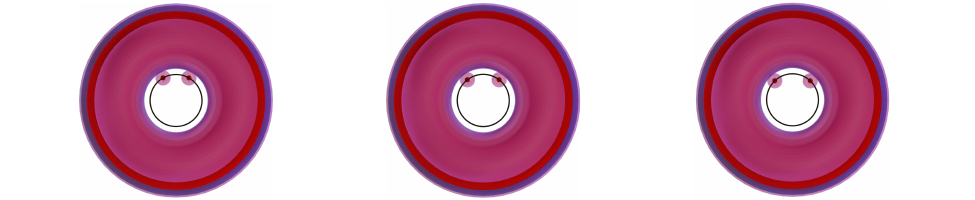
\subcaption{At first, the balls around the points $p_i$ and $\tilde{p}_i$ do not intersect the thickening of the set $C_i$, and thus the number of connected components of the thickening (in pink) of $P$ is different from the number of components of the manifold. }
\end{subfigure}
\newline
\begin{subfigure}[t]{0.99\textwidth}
  \def\svgwidth{0.99\linewidth}
    \centering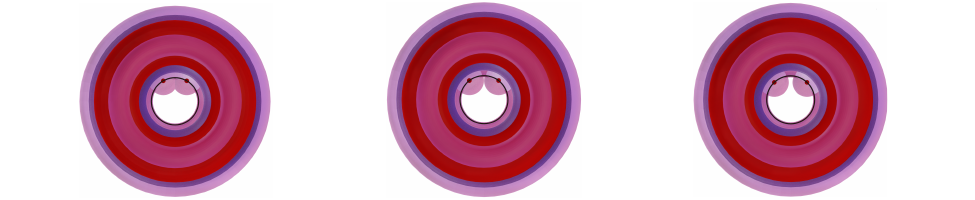
\subcaption{Then we create a (or possibly multiple) spurious cycle(s) for the first torus in the sequence (on the left).  }
\end{subfigure}
\newline
\begin{subfigure}[t]{0.99\textwidth}
  \def\svgwidth{0.99\linewidth}
  \centering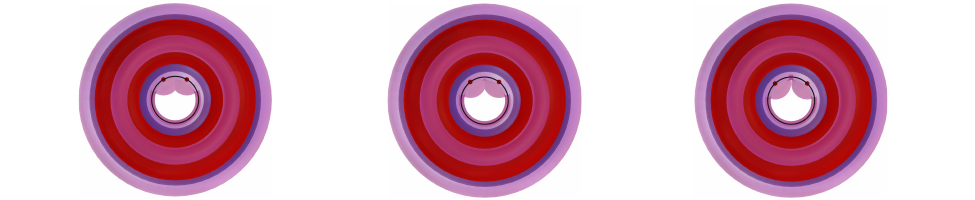
\subcaption{{By the time the spurious cycles at the first torus have disappeared, others have been created at the second torus.} This process is then repeated for all tori in the sequence as $r$ increases.} 
\end{subfigure}
\caption{The construction for manifolds imitates the construction for general sets of positive reach as much as possible. The manifold $\M$ is depicted in blue, the sample $P$ in red, and the thickening in pink. {We only display the part of objects below a horizontal clipping plane.}
}
\label{Fig:CycleManifold}
\end{figure}  

\begin{remark}
For simplicity, the sets constructed, see Figures \ref{Fig:CycleManifold} and \ref{fig:tight_set_nbhdsTeaser} (or Examples \ref{example:set} and \ref{example:mfld} in the appendix for details), are not connected. However, in each construction one can glue the connected components together in a way that preserves the reach, and the resulting examples still yield Propositions 
\ref{prop:counterexample_set} and \ref{prop:counterexample_mfld}. See Figure \ref{fig:OneConnectedComponent} for a sketch of the modification needed.
\end{remark}

\begin{remark}
	Propositions \ref{prop:counterexample_set} and \ref{prop:counterexample_mfld} show that the bounds~\eqref{equation:BoundOnR0} and~\eqref{equation:BoundOnEpsilon2} are tight in (ambient) dimensions $d\geq 2$, resp. $d\geq 3$. We did not construct similar examples in lower dimensions. Nevertheless, our intuition is that, in these cases, the bounds~\eqref{equation:BoundOnR0} and~\eqref{equation:BoundOnEpsilon2} can be improved further.
\end{remark}

	\begin{figure}[!h]
		\begin{center}
	\includegraphics[width=0.35\textwidth]{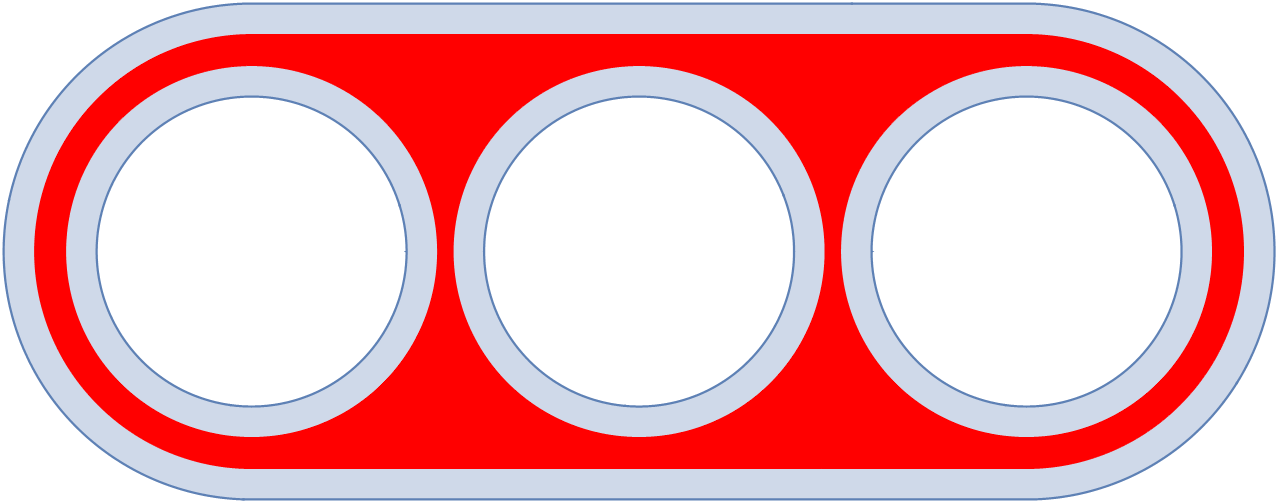}
	$\phantom{44444}$		
	\includegraphics[width=0.49\textwidth]{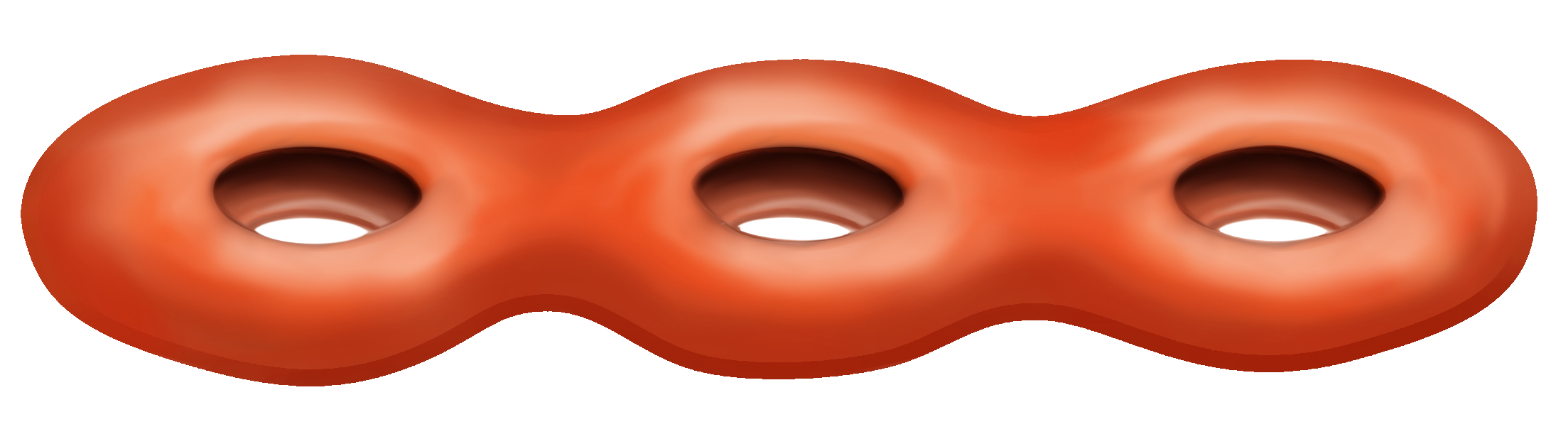}
		\end{center}
		\caption{The connected variants of our sets $\Su$ and $\M$ are a topological disc with $k$ holes and a genus $k$ surface. On the left we sketch both the sample and the set of positive reach, on the right we only give the sample for the manifold setting because of visualization constraints.  
		}
		\label{fig:OneConnectedComponent}
	\end{figure}

\section{Results for subsets of Riemannian manifolds}
\label{sec:OverviewBoundsRiemannianCase}
\subsection{Setting}

In the second part of this paper we consider subsets of a ($C^2$)
Riemannian manifold $\N$.  {In this Riemannian} setting we denote
(geodesic) balls with radius $r>0$ centred at a point $p\in \N$ by
$B(p,r)$, and write $A^{\boxplus r} = \bigcup_{a \in A} B(a,r)$ for
the union of (geodesic) balls of radius $r$ centred at a subset
$A\subseteq \N$. {Similarly, the one-sided Hausdorff distance from $X
  \subseteq \N$ to $Y \subseteq \N$ is defined as the smallest $\rho$
  such that the union of (geodesic) balls of radius $\rho$ centered at
  $X$ covers $Y$.}

To be able to proceed as in the Euclidean setting and state tight bounds on the sampling parameters, we need a notion of the reach in the Riemannian setting. To this end, we introduce a new definition, inspired by the cut locus (which is defined for example in \cite{berger2003panoramic}):
\begin{restatable}[Cut locus]{definition}{SetCutLocus}
\label{definition:SetCutLocus}
Given a closed subset $\Su \subseteq \N$, the {\em cut locus} of $\Su$ is the set $\operatorname{cl}_{\N}(\Su)$
of points $p \in \N$ for which there are at least $2$ geodesics of minimal length from $p$ to some point in $\Su$.
\end{restatable}

\begin{restatable}[Cut locus reach]{definition}{ReachCutLocus}
\label{definition:ReachCutLocus}
The {\em cut locus reach} $\rchcl_{\N}(\Su)$ of a closed set $\Su \subseteq \N$ is the infimum of distances between $\Su$ and its cut locus $\operatorname{cl}_{\N}(\Su)$:
\[
\rchcl_{\N}(\Su) \defunder{=} \inf_{\substack{ p \in \Su, \\ q  \in \operatorname{cl}_{\N}(\Su), }} d_\N(p,q).
\]
\end{restatable}

Our definition is a refinement of the notion 
used by Bangert and Kleinjohann~\cite{bangert1982sets, kleinjohann1980convexity, kleinjohann1981nachste}, as well as the reach defined in \cite{ReachSubmanifolds}. We explain why our new definition is appropriate for the learning of topological features in Appendix~\ref{sec:reach_history}.
Using the new extension of the reach we assume the following conditions, which resemble the ones in the Euclidean setting closely: 
	\begin{tcolorbox} 
\begin{restatable}{assumption2}{assumptionRiemannianSetting}
\label{assumption:RiemannianSetting}
We work with a closed set $\Su \subseteq
	\N$ with positive cut locus reach $\rchcl_{\N} (\Su)$, and let $\reach>0$ be a constant
	satisfying $\reach \leq \rchcl_{\N} (\Su)$. Furthermore, we consider a
	set $P\subseteq \N$, such that the one-sided Hausdorff distance 
	from $P$ to $\Su$ is at most
	$\delta$, and the one-sided Hausdorff distance 
	from $\Su$ to $P$ is at most
	$\varepsilon$. That is,
	$\Su \subseteq P^{\boxplus \varepsilon}$  
	and 
	$P \subseteq \Su^{\boxplus \delta}$.  
	We assume that $\delta, \varepsilon <\reach$.
    
	We also assume that the sectional curvatures of the manifold $\N$ are lower bounded {by a constant $\curvlowbnd \in \mathbb{R}$.} 
	When $\curvlowbnd>0$ and $\Su= \M$ is a manifold, 
        we can safely assume, thanks to Lemma~\ref{lemma:ManifoldHaveReachLessThanPiOverTwo}, that $\reach \leq \frac{\pi}{2 \sqrt{\curvlowbnd}}$. 
\end{restatable}
\end{tcolorbox} 
{This assumption is used in Section~\ref{sec:OverviewBoundsRiemannianCase} and Appendix~\ref{sec:Riemannian_setting}.}

\subsection{Bounds on the sampling parameters}
Also in the Riemannian setting we 
provide (tight) bounds that the sample $P$ needs to satisfy in order to be able to infer homotopy. For sets of positive (cut locus) reach, we obtain the following bounds on $\varepsilon$ and $\delta$:  
\begin{restatable}{proposition}{HomotopyPositiveReachRiemannian}
\label{theorem:HomotopyPositiveReachRiemannian}
If $\varepsilon$ and $ \delta$ satisfy
\begin{align}
2 \cos \left(\sqrt{\curvlowbnd}  ( \reach - \delta)\right) - \cos  \left(\sqrt{\curvlowbnd}  ( \reach + \varepsilon)\right)  & \leq  1 && \text{if $\curvlowbnd> 0$},
\nonumber \\
{\sqrt{2} (\reach - \delta)  - (\reach + \varepsilon) } &\leq 0 && \text{if $\curvlowbnd= 0$},
\label{equation:BoundOnR0Riemannian} \\
2 \cosh \left(\sqrt{\curvlowbndn}  ( \reach - \delta)\right) - \cosh  \left(\sqrt{\curvlowbndn}  ( \reach + \varepsilon)\right)   &\geq  1 && \text{if $\curvlowbnd< 0$}, 
\nonumber
\end{align}
there exists a radius $r>0$ such that the union of balls $P^{\boxplus r}$ 
deformation-retracts onto $\Su$ along the closest point projection. In particular, $r$ can be chosen as:
\begin{equation} 
 r = \frac{1}{2} \left(\reach+\varepsilon  \right).
\label{EQ:InvervalrSetPosReach_Riemann} 
\end{equation}
\end{restatable}

If the set is a manifold, the bounds on $\varepsilon$ and $\delta$ can be improved as follows:
\begin{restatable}{proposition}{DeformRetractsTheoremForManifoldsRiemann}
  \label{theorem:DeformRetractsTheoremForManifolds_Riemann}
  Let $\tilde x = \sqrt{\curvlowbndn} x$. For $\delta \leq \varepsilon$ satisfying
  \begin{align}
    & \left(2\cos \tilde{\varepsilon}\cos \tilde{\reach}-3\cos \left(\tilde{\reach}-\tilde{\delta} \right)\right)^2\leq  \left( \frac{\cos \tilde{\varepsilon}  - \cos \left(\tilde{\reach}-\tilde{\delta}\right) \cos \tilde{\reach}}{\sin \tilde{\reach}} \right)^2+ \cos^2 \left(\tilde{\reach}-\tilde{\delta}\right)    
		\nonumber 
	 \\
     & \omit\hfill \textrm{\emph{if $\curvlowbnd > 0$,}\hspace{7mm}}  \label{equation:BoundOnEpsilon2_Riemann}
    \\
    & (\reach - \delta)^2 - \varepsilon^2   \geq   \left(4\sqrt{2}  - 5\right) \reach^2 \nonumber \\
     & \omit\hfill \textrm{\emph{if $\curvlowbnd = 0$,}\hspace{7mm}} 
		\tag{\ref{equation:BoundOnEpsilon2}}
		\\
    & 2 \cosh \tilde{\varepsilon}\cosh \tilde{\reach}\leq 3\cosh \left(\tilde{\reach}-\tilde{\delta}\right) \qquad\text{and}\qquad\nonumber\\
    & \cosh^2 \left(\tilde{\reach}-\tilde{\delta}\right)\leq\left(\frac{\cosh \tilde{\varepsilon}  - \cosh \left(\tilde{\reach}-\tilde{\delta}\right) \cosh \tilde{\reach}}{\sinh \tilde{\reach}}\right)^2+ \left(2\cosh \tilde{\varepsilon}\cosh \tilde{\reach}-3\cosh \left(\tilde{\reach}-\tilde{\delta} \right)\right)^2
	\nonumber
   \\  & \omit\hfill \textrm{\emph{if $\curvlowbnd < 0$,}\hspace{7mm}}
	\label{equation:BoundOnEpsilon2_RiemannHyper}
  \end{align}
  there exists a radius
  $r>0$ such that $P^{\boxplus r}$ deformation-retracts onto $\M$
  along the (geodesic) closest point projection $\pi_\M$.
The interval from which $r$ can be chosen can be recovered from \eqref{eq:OtherFprmForF}, \eqref{Bounds_r_For_Manifold_Case}, and \eqref{eq:OtherFprmForG} respectively. 
\end{restatable}

{ 
The computation of \v{C}ech complexes in a Riemannian manifold can be difficult (depending on the manifold). Fortunately, we can avoid this step and still recover the homology: 
\begin{remark}
The results of Chazal and co-authors~\cite{chazal2008towards} on the interleaving between the \v{C}ech and Rips complexes extend to the Riemannian setting. By combining their results with the results of this paper, one can recover the homology type of a subset of positive reach of a Riemannian manifold using persistent homology of Rips complexes. 
\end{remark} 
The Rips complex is easier to calculate than the \v{C}ech complex, since the calculation only involves distances between pairs of points. 
}

\subsection{Tightness of the bounds on the sampling parameters}
We exhibit the tightness of the bounds on $\varepsilon$ and $\delta$ from Propositions \ref{theorem:HomotopyPositiveReachRiemannian} and \ref{theorem:DeformRetractsTheoremForManifolds_Riemann} by constructions of examples in (simply connected) spaces of constant curvature. 
These constructions are similar to the Euclidean setting --- they also consist of annuli and tori, see Figure \ref{fig:sequenceOfAnnuliOnTheSphereTeaser}. However, due to the curvature of the ambient manifold, the proof of the tightness of the bounds is significantly more involved (see Appendix~\ref{sec:Riemannian_tightness}). 

\begin{figure}[!h]
		\begin{center}
		  \includegraphics[width=0.5\textwidth]{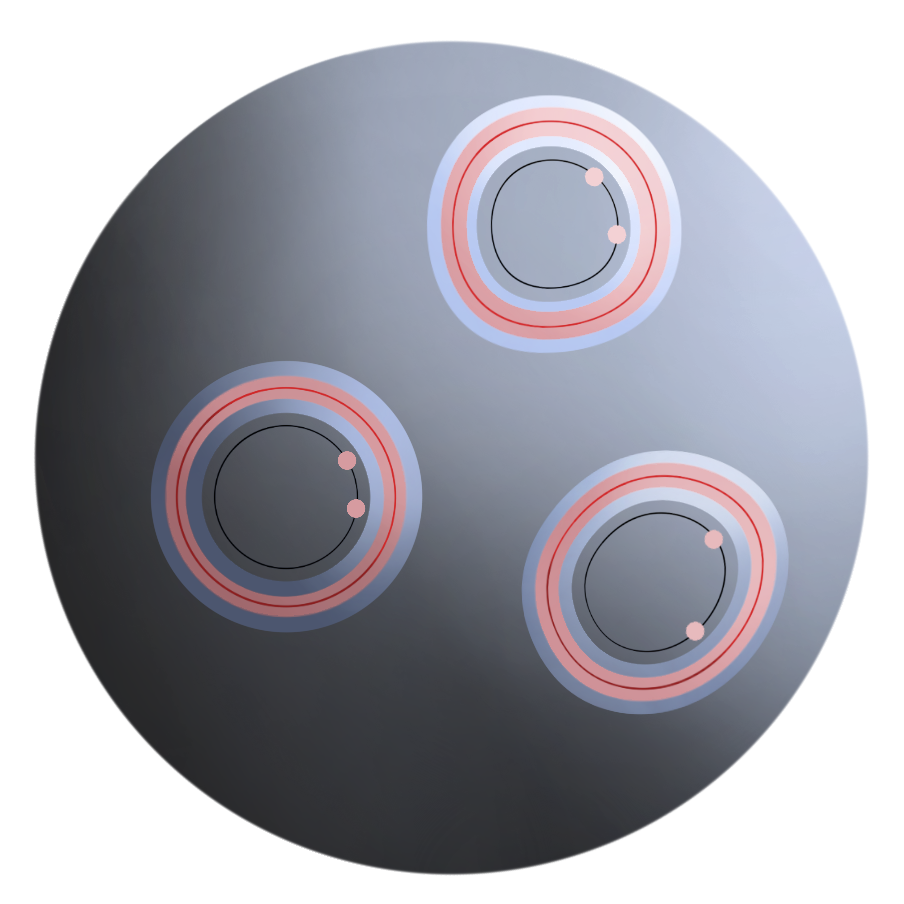} 
		\end{center}
		\caption{
	  The construction for sets of positive reach on a manifold with (constant) positive curvature (the sphere). For a detailed version of the figure see Figure \ref{fig:sequenceOfAnnuliOnTheSphere}.
		}
		\label{fig:sequenceOfAnnuliOnTheSphereTeaser}
	\end{figure}

\section{Future work} 
This article leaves several important questions unanswered. We mention three.

First of all, we consider the union of balls {centered on a sample $P$} 
 whose homotopy type is equal to that of the {\v{C}ech} complex {of $P$}
and, when the ambient space is a Riemannian manifold, the radius of balls 
is smaller than the convexity radius. 
} 

It would be interesting to see if our work would help understanding the same question for Rips complexes. For related work see e.g. \cite{adamaszek2017vietoris,HenryRips1, hausmann1995vietoris, latschev2001vietoris}. 

Second, we consider sets embedded in Riemannian manifolds whose sectional curvature is lower bounded. 
A natural question is under which conditions do our results generalize to a larger class of metric spaces with lower bounded curvatures. 

{
The generalized gradient of the distance function and its flow have been used to generalize results on subsets of positive reach in Euclidean space to subsets with positive $\mu$-reach and weak feature size \cite{chazal2009sampling, cl2005lambda, chazal2005weak, chazal2008towards}. Our work on the cut locus reach makes it possible to extend the notations of positive $\mu$-reach and  weak feature size to Riemannian manifolds. It is expected that many of the main results from the Euclidean setting still hold with minor modifications in this more general context. 
}

\phantomsection
\addcontentsline{toc}{section}{Bibliography}
\bibliography{geomrefs}

\appendix
\section*{Appendix I: The technical statements and proofs}

{The two sections in this part of our paper are structured in the same way. In the first section (Section \ref{sec:Euclidean_setting}), we deal with subsets of Euclidean space, in the second (Section \ref{sec:Riemannian_setting}) with subsets of Riemannian manifolds. In each section, we first introduce necessary definitions and recall the general setting (Sections \ref{sec:Euclidean_definitions} and \ref{sec:Riemannian_definitions}). In Sections \ref{sec:Euclidean_geometric_argument} and \ref{sec:Riemannian_geometric_argument}
we consider a set $\Su$, its sample $P$, and use a geometric argument to establish a condition on the thickening parameter $r>0$ that guarantees that the thickening $P^{\boxplus r}$ of the sample $P$ deformation-retracts to the set $\Su$. In the following sections (Sections \ref{sec:bounds} and \ref{sec:boundsRiemannian}) we show that if the sampling parameters $\varepsilon$ and $\delta$ of the sample satisfy certain bounds, the condition on the thickening parameter is never satisfied. We carefully distinguish between subsets (Sections~\ref{sec:bounds_Euclidean_sets} and~\ref{sec:bounds_Riemannian_sets}) and submanifolds (Sections~\ref{sec:bounds_Euclidean_mflds} and~\ref{sec:bounds_Riemannian_mflds}), for which we obtain sharper bounds. Finally (Sections \ref{section:optimality} and \ref{sec:Riemannian_tightness}), we construct explicit counterexamples to prove that our bounds on the sampling parameters are tight. 
}

\section{Subsets of the Euclidean space}
\label{sec:Euclidean_setting}
\subsection{Definitions and setting}
\label{sec:Euclidean_definitions}


In this section we revise the notions and results by Federer~\cite{Federer}. We assume that $\Su \subset \mathbb{R}^d$ is a closed set, and denote the closest point projection on $\Su$ by $\pi_\Su$.



At first, we define the medial axis, the local feature size, and the reach of the set $\Su$:
\begin{definition}\label{def:medial_axis}
	The medial axis $\ax(\Su)$ of a closed set $\Su$ is the set of points in the ambient Euclidean space that do not have a unique closest point on $\Su$. The distance from a point $p$ to the medial axis is called the local feature size $\lfs(p)$. Finally, the (minimal) distance between $\ax(\Su)$ and $\Su$ is the reach $\rch(\Su)$ of $\Su$:
	\[ \lfs(p) = \inf_{q\in \ax(\Su)} \norm{p-q}, \qquad  \rch(\Su) = \inf_{p\in\Su}\lfs(p).\]
\end{definition} 
For example, the medial axis of an ellipse in the Euclidean plane is the (open) segment connecting the two focal points, and the reach is the distance from (one of) the focal point(s) to the ellipse.

Next, we introduce the normal cone. We denote the scalar product in $\R^d$ by $\langle .,. \rangle$.
\begin{definition}[Definitions 4.3 and 4.4 of \cite{Federer}] \label{def:4.3and4.4Fed} 
	If $\Su \subseteq \R^d$ and $p\in \Su$, then the generalized
	tangent space $\Tan(p,\Su)$
	is the set of all tangent vectors of $\Su$ at $p$. 
	It consists of all those $u \in \R^d$, such that either $u=0$ or for every $\varepsilon>0$ there exists a point $q \in  \Su$ with
	\begin{align}
		0<&\|q-p\|<\varepsilon &\textrm{and} & &\left\| \frac{q-p}{\|q-p\|}- \frac{u}{\|u\|} \right\| < \varepsilon.
		\nonumber 
	\end{align}
	The normal cone of $\Su$ at $p$ is the set 
	\begin{align*}
		\Nor(p,\Su)
	\end{align*}
	of all vectors $v \in \R^d$ such that  $\langle v,u \rangle \leq 0$ for all
		$u \in \Tan(p,\Su)$.
\end{definition}

We illustrate the medial axis and a few normal cones in Figure~\ref{fig:set_covering} (left). The normal cone is indeed a cone, geometrically speaking: 


\begin{Defremark2}[{\cite[Remark 4.5]{Federer}}] 
	A subset $C \subseteq \R^d$ is a \emph{convex cone} if and only if for all $x,y \in C$ and $\lambda>0$ we have $x+y \in C$ and $\lambda x \in C$. For every set $A \subseteq \R^d$, its dual
	\[ 
	\Dual (A) = \{ v \mid   \langle v,u \rangle \leq 0 \textrm{ for all } u \in A \}, 
	\]   
	is a closed convex cone. The double dual, $\Dual (\Dual (A))$, is the smallest closed convex cone that contains the set $A$.  
	The~set $\Nor(p,\Su)$ 
	is therefore a {convex} cone. 
	 
	The~generalized tangent space $\Tan (p , \Su)$, on the other hand, is only closed and positively homogeneous, but not necessarily convex. That is, if $v \in \Tan (p , \Su)$, $ \lambda v \in \Tan (p , \Su)$ for all $\lambda \in \R_{\geq 0}$. The~space $\Tan (p , \Su)$ is a convex cone if the set $\Su$ has positive reach, as we will see below. 
\end{Defremark2}

With these definitions in place we present the following two theorems, that form the core of the proof of our statement on deformation retraction of the set $\Su$ (Theorem~\ref{theorem:geometric_argument}).
\begin{theorem}[{\cite[Theorem 4.8 (8)]{Federer}}]
	\label{Fed4.8.8} 
	Let $\ell$ and $\reach$ satisfy $0 < \ell < \reach<\infty$ and $\rch(\Su) \geq \reach$. Then any points $x,y \in \R^d \backslash \ax (\Su)$ with
\[d(x,\Su) \leq \ell \qquad\text{and}\qquad  d(y,\Su) \leq \ell\]
	satisfy  
	\begin{align} 
		\| \pi_\Su (x) - \pi_\Su (y) \| \leq \frac{\reach}{\reach - \ell}  \| x-y \| . 
		\nonumber
	\end{align} 
\end{theorem} 


\begin{theorem}[{\cite[Theorem 4.8 (12)]{Federer}}] 
	\label{Fed4.8.12}
Let $p \in \Su$. Then for any number $\ell$ satisfying $\lfs (p)> \ell > 0 $, the normal cone equals
\begin{align}
	\Nor ( p,\Su) = \{ \lambda v \mid \lambda\geq 0 , \| v \| = \ell ,  \pi_{\Su} (p+v) =p\} .
	\nonumber
\end{align}
$\Tan(p,\Su)$ is the convex cone dual to $\Nor(p,\Su)$, and, for any vector $u\in {\Tan(p,\Su)}$,
\begin{align}
	\lim_{t \to 0^+} t^{-1} d(p +t u, \Su) =0.
	\nonumber
\end{align}

\end{theorem} 

Finally, we recall the setting we assume for the remainder of Section~\ref{sec:Euclidean_setting}:

	\begin{tcolorbox} 
\assumptionEuclideanSetting*
\end{tcolorbox}

\subsection{Bounds on the sampling parameters} 
\label{sec:bounds}

In this section we first compute the bounds on the size $\alpha$ of the neighbourhood $\Su^{\boxplus \alpha}$ 
covered by the union of balls $\bigcup_{p \in P } B(p,r) =  P^{\boxplus r}$ 
in terms of $\varepsilon, \delta$, and $r$. 
We then combine these bounds with Equation \eqref{equation:R2TooSmallToCNormalLineS} to infer (optimal) upper bounds on $\varepsilon$ and $\delta$, for which there exists a radius $r$ such that the deformation retract from $P^{\boxplus r }$
to $\Su$ is possible.
We do so first for sets of positive reach and then for manifolds. Somewhat counter-intuitively, it turns out to be easier to determine the bounds for sets of positive reach.

\subsubsection{Sets of positive reach}\label{sec:bounds_Euclidean_sets}
 
For sets of positive reach, the bound on $\alpha$ is almost trivial. Nevertheless, it is tight, as we will see in Section \ref{section:optimality}.


\begin{lemma}\label{lem:Bounds_alpha_sets_pos_reach} 
Suppose that $\Su \subseteq P^{\boxplus \varepsilon}$ 
for some $\varepsilon \geq 0$. 
Then, for all $\alpha \leq r - \varepsilon$, the $\alpha$-neigbourhood $\Su^{\boxplus \alpha} $ 
of $\Su$ is contained in the union of balls $P^{\boxplus r}$. 
		That is,
    \[
    \Su^{\boxplus \alpha} 
		\subseteq P^{\boxplus r}. 
    \]
  \end{lemma}

\begin{proof} 
{ 
	Writing out the definition we see that the $\boxplus$ operation is additive. For any set $A\subseteq\R^d$: 
	\begin{align}
		(A^{\boxplus r_1})^{\boxplus r_2} &= \bigcup_{a' \in  A^{\boxplus r_1} } B(a', r_2)
		\nonumber
		\\
		& = \bigcup_{a' \in  \bigcup_{a \in A} B(a ,r_1 ) } B(a', r_2)
		\nonumber
		\\
		&\subseteq \bigcup_ {a \in A} B(a , r_1+r_2) \tag{by the triangle inequality}
		\\
		&= A^{\boxplus (r_1+r_2)}.
		\label{eqq:additivityBoxPlus}
	\end{align}
	So indeed,
	\begin{align}
		\Su^{\boxplus \alpha} &\subseteq (P ^{\boxplus \varepsilon })^{\boxplus \alpha} 
		\tag{because $\Su  \subseteq P ^{\boxplus \varepsilon }$}
		\\ 
		&\subseteq P^{\boxplus (\varepsilon+\alpha)} \tag{by \eqref{eqq:additivityBoxPlus}} 
		\\
		&\subseteq P^{\boxplus r} .
		\tag{because by assumption $\alpha \leq r - \varepsilon$}
	\end{align} 
}
\end{proof}

\begin{remark}\label{remark:OffsetISAddiiveInGeodesicSpaces}
The statement of Lemma \ref{lem:Bounds_alpha_sets_pos_reach} holds in any metric space. Writing $B(a, r)$ for a metric ball with radius $r$ centred at a point $a$, and $A^{\boxplus r} = \bigcup_{a \in A} B(a, r)$ for the thickening of a set $A$ in the metric space, we see from the proof of Lemma \ref{lem:Bounds_alpha_sets_pos_reach} that
\[  (A^{\boxplus r_1})^{\boxplus r_2} \subseteq A^{\boxplus (r_1+r_2)},\]	
 with equality if the metric space is geodesic.
\end{remark}
	
From Lemma~\ref{lem:Bounds_alpha_sets_pos_reach}, we derive the bounds on $\varepsilon$ and $\delta$ (in terms of $\reach$).

\HomotopyNoiselessPositiveReach*

\begin{proof} 
We combine the bound from Lemma \ref{lem:Bounds_alpha_sets_pos_reach} with the conditions of Theorem \ref{theorem:geometric_argument}. More precisely,  
inserting $\alpha= r -\varepsilon $ in Equation \eqref{equation:R2TooSmallToCNormalLineS} yields that
\begin{equation} \label{eq:triangle_ineq}
r^2 +(\reach - r +\varepsilon)^2 \leq  (\reach -  \delta)^2.
\end{equation} 

 Using the abc-formula for quadratic equations, this is equivalent to 
\[ r \in  \left [\frac{1}{2} \left(\reach+\varepsilon  -  \sqrt{\Delta}\right),  
\frac{1}{2} \left(\reach+\varepsilon +  \sqrt{\Delta}\right)  \right ],
\] 
where 
\[ \Delta= 2 \delta ^2+\reach^2-4 \delta  \reach-2 \reach \varepsilon -\varepsilon ^2 = 2(\reach-\delta)^2 - (\reach+\varepsilon)^2 \]  
is the discriminant. 
This interval is non-empty if the discriminant is non-negative, that is, if $\varepsilon + \sqrt{2} \, \delta \leq  (\sqrt{2}  - 1) \reach$.   
\end{proof} 

\ExtendedInterval*

\begin{remark}
The parameter $\delta$ is not necessarily smaller than $\varepsilon$, even if this would be natural in most applications.
\end{remark}

\subsubsection{Manifolds with positive reach}
\label{sec:bounds_Euclidean_mflds}
In this section, we show that the bounds from Proposition \ref{theorem:HomotopyNoiselessPositiveReach} can be improved further if the set of positive reach is a manifold.
In Lemma \ref{lem:Bounds_alpha_sets_pos_reach}, we used the triangle inequality to set~$\alpha = r- \varepsilon$.

If~$\Su$ is a manifold, however,  
the parameter $\alpha$ 
can be increased using more subtle arguments than the triangle inequality:
Manifolds with positive reach are $C^{1,1}$ smooth\footnote{Topologically embedded manifolds with positive reach are $C^{1,1}$ embedded \cite{Federer, 
lytchak2004geometry, lytchak2005almost, rataj2019curvature, StructureRataj}.}, i.e., differentiable with Lipschitz derivative.  
Moreover, Federer's normal cone $\Nor(q,\M)$ (Definition \ref{def:4.3and4.4Fed}) coincides at every point $q \in \M$ with the `classical' normal space $N_q\M$ of an $n$-dimensional submanifold $\M$ of $\R^d$.
In particular, the tangent and normal cones of manifolds of positive reach are $n$- and $(d-n)$-dimensional linear spaces, respectively, that are not only dual, but also orthogonal.

In Lemma \ref{lemma:BoundOnTubularCoverForManifolds}, we establish a lower bound for the parameter $\alpha$ in the case that $\Su=\M$ is a manifold. This bound is tight, as we will see in Section \ref{section:optimality}. 
 
\begin{lemma}\label{lemma:BoundOnTubularCoverForManifolds}
Suppose that $\M \subseteq P^{\boxplus \varepsilon}$ 
for some
$\varepsilon \geq 0$. Then, for any $r\geq\alpha\geq 0$ satisfying
{
		\begin{equation}\label{equation:ConditionTubularCoverForManifold}
		r^2 \geq    \alpha^2 + \frac{\alpha}{\reach} \left( \reach^2 + \varepsilon^2 - (\reach- \delta)^2 \right)   + \varepsilon^2,
		\end{equation}
}
		the $\alpha$-neighbourhood $\M^{\boxplus \alpha}$  
		of $\M$ is contained in the thickening $P^{\boxplus r}$. 
		That is,
		\[
		\M^{\boxplus \alpha } 
		\subseteq P^{\boxplus r}.  
		\]
		\end{lemma}
\begin{proof}
Given a point $q \in \M$, the tangent{ space} 
$T_q\M$ and the normal{ space}
 $N_q\M$ are orthogonal vector spaces satisfying $T_q\M \times N_q\M = \R^d$, where $\times$ denotes the 
direct product.
Since $\M \subseteq P^{\boxplus \varepsilon}$, the intersection $P\cap B(q,\varepsilon)$ is non-empty. Let $p \in P\cap B(q,\varepsilon)$.

\begin{figure}[!h]
\begin{center}
\includegraphics[width=0.5\textwidth]{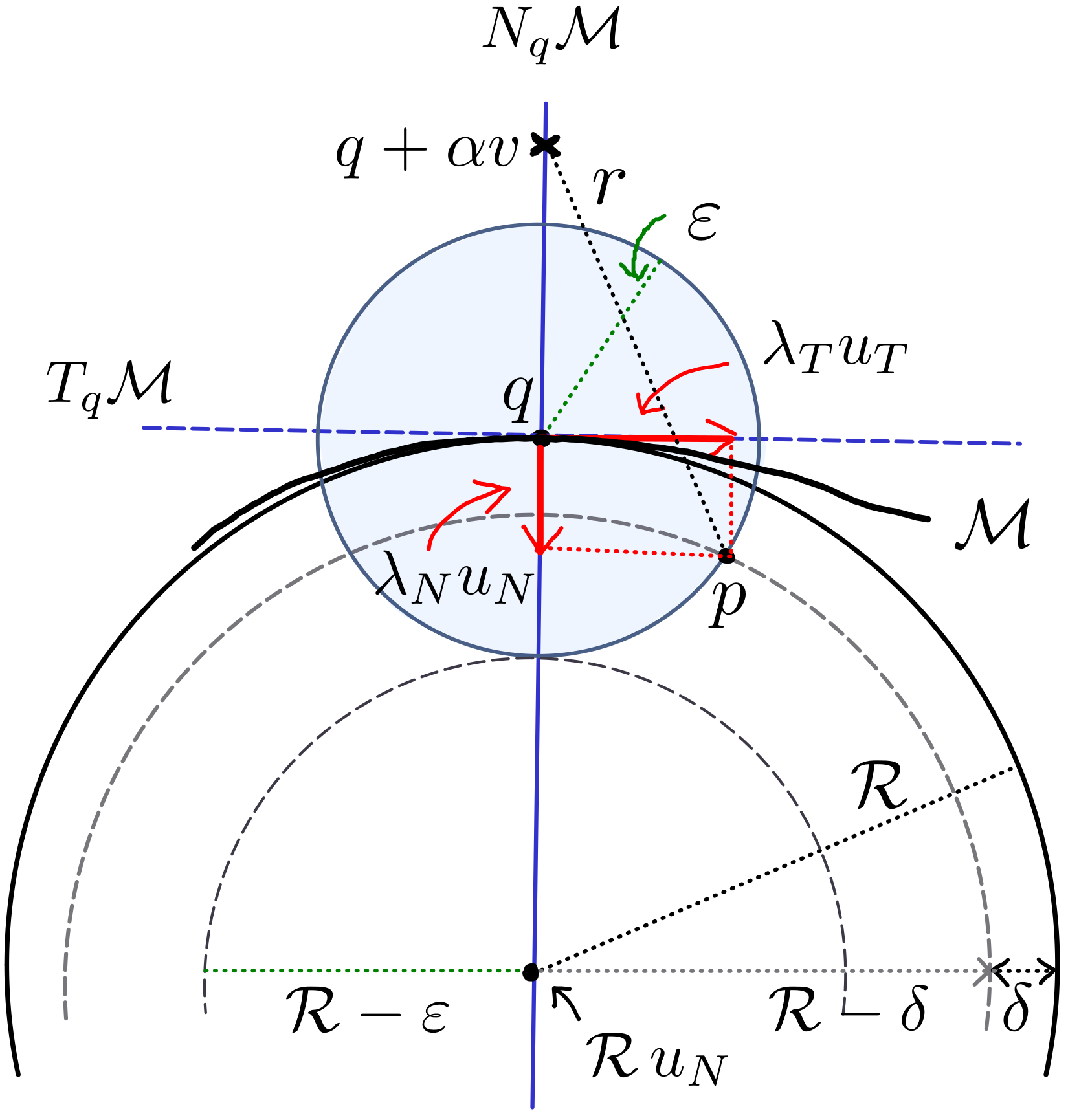}

\end{center}
\caption{Overview of the notation used in the proof of Lemma  \ref{lemma:BoundOnTubularCoverForManifolds}.
} 
\label{fig:ProofManifoldCase1}
\end{figure}

The vector $p-q$ decomposes uniquely as
\[
p-q = \lambda_T u_T + \lambda_N u_N,
\]
with $u_T \in T_q\M$,  $u_N \in N_q\M$, $\|u_T\|= \|u_N\| = 1$, and $\lambda_T, \lambda_N \geq 0$ (see Figure \ref{fig:ProofManifoldCase1}).

Since $\|p-q\| \leq \varepsilon$, 
\begin{equation}\label{equation:NormPMinusQLessThanEpsilon}
\lambda_T^2 + \lambda_N^2  \leq \varepsilon^2.
\end{equation}

Thanks to \cite[Theorem 4.8 (12)]{Federer} (Theorem \ref{Fed4.8.12}), the sets $\M$ and \\$B(q  + \reach\cdot u_N, \reach)^\circ$ do not intersect, and thus:
\[
p \notin B(q  + \reach \cdot u_N, \reach- \delta)^\circ. 
\]
Hence,
$(p - q  - \reach \cdot u_N ) ^2 \geq (\reach- \delta)^2$.
Applying the decomposition of $p-q$ we obtain
\[
\left( \lambda_T u_T +  (\lambda_N  - \reach)  u_N \right)^2 \geq (\reach- \delta)^2, 
\]
which implies that
\[
  \lambda_T^2   +  ( \reach -\lambda_N )^2 \geq (\reach- \delta)^2.
\]
Combining this result with Equation \eqref{equation:NormPMinusQLessThanEpsilon} implies that
\[
\varepsilon^2 -  \lambda_N^2   +  ( \reach -\lambda_N )^2 \geq (\reach- \delta)^2,
\]
which can be rewritten as
\begin{equation}\label{equation:BoundOnLambda_N}
2 \reach \lambda_N \leq \reach^2 + \varepsilon^2 - (\reach- \delta)^2.
\end{equation}
Choose a vector $v \in \Nor(q, \M)$ with $\|v\|= 1$, and let $\alpha \geq 0$. Then, 
\begin{align}
\left( p - (q + \alpha v)  \right)^2 &= \left( ( \lambda_N u_N - \alpha v ) +  \lambda_T u_T \right)^2
\nonumber
\\
& = \left(  \lambda_N u_N - \alpha v   \right)^2 + \left(\lambda_T u_T \right)^2 
\nonumber
\\
& \leq ( \lambda_N  + \alpha )^2 + \lambda_T^2  
\nonumber
\\
& \leq  ( \lambda_N  + \alpha )^2 + \varepsilon^2 - \lambda_N^2 
\tag{by \eqref{equation:NormPMinusQLessThanEpsilon}} 
\\
& = \alpha^2 + 2 \alpha \lambda_N + \varepsilon^2.
\nonumber
\end{align}
Using inequality \eqref{equation:BoundOnLambda_N} to substitute $2 \lambda_N$,  we further obtain:
\[
\left( p - (q + \alpha v)  \right)^2 \leq  \alpha^2 + \frac{\alpha}{\reach} \left( \reach^2 + \varepsilon^2 - (\reach- \delta)^2 \right)   + \varepsilon^2.
\]

Thus, if 
\[r^2 \geq    \alpha^2 + \frac{\alpha}{\reach} \left( \reach^2 + \varepsilon^2 - (\reach- \delta)^2 \right)   + \varepsilon^2  ,\]
then the point $q+ \alpha v$ lies in $B(p, r) \subseteq P^{\boxplus r}$. 
Since this inclusion holds for any $q \in \M$ and $v \in N_q\M$ {with $\|v\| = 1$}, 
$\M^{\boxplus \alpha} \subseteq P ^{\boxplus r}$. 
\end{proof}

As in Proposition \ref{theorem:HomotopyNoiselessPositiveReach}, we now derive a bound on $\varepsilon$.
\DeformRetractsTheoremForManifolds*
The bound is illustrated in Figure \ref{fig:graphCcr}.

\begin{proof}
We combine the bound from Lemma \ref{lemma:BoundOnTubularCoverForManifolds} with the conditions of Theorem \ref{theorem:geometric_argument}. More precisely,  
combining Equations \eqref{equation:R2TooSmallToCNormalLineS} 
and \eqref{equation:ConditionTubularCoverForManifold}
yields the following sufficient condition for  $L \cap \left(P^{\boxplus r} 
\right)$ to be connected:
\begin{equation}\label{equation:BoundOnrForManifolds}
{ \alpha^2 + \frac{\alpha}{\reach} \left( \reach^2 + \varepsilon^2 - (\reach- \delta)^2 \right)   + \varepsilon^2   \leq   \, r^2  \leq \,  (\reach -  \delta)^2 - (\reach - \alpha)^2.}
\end{equation}
{The inequality between leftmost and rightmost members of \eqref{equation:BoundOnrForManifolds},
 which needs to be satisfied for a non-empty range of values for $r$ to exist, can be rearranged as: }
\[
0  \geq \varepsilon^2 - \left(\reach-\delta\right)^2 + \reach^2 + \alpha \, \frac{1}{\reach}\left(\varepsilon^2 - \reach^2- \left(\reach-\delta\right)^2\right) +2\alpha^2.
\]
Using the abc-formula for quadratic equations, the above inequality is satisfied if $\alpha \in  \left[ \alpha_{\min} , \alpha_{\max} \right] $, with 
\begin{align} 
\alpha_{\min} =  \frac{1}{4} \left(\frac{(\reach-\delta)^2 +\reach^2 -\varepsilon ^2}{\reach} - \sqrt{ \Delta }\right), 
\alpha_{\max} =  
\frac{1}{4} \left(\frac{(\reach-\delta)^2 +\reach^2 -\varepsilon ^2}{\reach} +\sqrt{ \Delta }\right), 
\label{eq:DefAlphaminMax} 
\end{align} 
where the discriminant $\Delta$ is
\[
\Delta = \frac{1}{\reach^2}\left(\varepsilon^2 - \left(\reach-\delta\right)^2\right)^2 - 10\left(\varepsilon^2 - \left(\reach-\delta\right)^2\right) - 7\reach^2.
\]
The discriminant can be viewed as a polynomial in $y=\varepsilon^2 - \left(\reach-\delta\right)^2$. Solving $\Delta(y)=0$ with respect to $y$ yields $y= \reach^2\left(5\pm4\sqrt{2}\right)$. This in turn implies that $\Delta$ is non-negative if either $\varepsilon^2 - \left(\reach-\delta\right)^2 \leq \reach^2\left(5-4\sqrt{2}\right)$ or $\varepsilon^2 - \left(\reach-\delta\right)^2\geq \reach^2\left(5+4\sqrt{2}\right)$. Thanks to Assumption~\ref{assumption}, we are only interested in the case where $\varepsilon,\delta<\reach$, and thus we can ignore the second inequality. 
Hence the interval $\left[ \alpha_{\min} , \alpha_{\max} \right]$ is non-empty if 
\[\varepsilon^2 - \left(\reach-\delta\right)^2 \leq \left(5 - 4 \sqrt{2} \right) \reach^2.
\]
Substituting the bounds on $\alpha$ (Equation \eqref{eq:DefAlphaminMax}) into  Equations \eqref{equation:R2TooSmallToCNormalLineS} and \eqref{equation:ConditionTubularCoverForManifold} yields 
\begin{align} \left(1 +  \frac{\alpha_{\min} }{\reach} \right) \varepsilon^2 + \alpha_{\min}^2 + \frac{\alpha_{\min}}{\reach} \left( \reach^2 - (\reach - \delta)^2\right) \leq r^2 &\leq  (\reach -  \delta)^2 - (\reach - \alpha_{\max})^2 .
\label{Bounds_r_For_Manifold_Case}
\end{align} 
\end{proof}

\begin{remark}
We restricted ourselves to the case where $\delta \leq \varepsilon$, because if $\delta>\varepsilon$, the fact that the set of positive reach is a manifold no longer helps. The geometric reason for this is that $p$ in Figure~\ref{fig:ProofManifoldCase1} may lie in $N_q \M$. 
\end{remark} 

\subsection{Tightness of the bounds on the sampling parameters} 
\label{section:optimality}
In this section we prove that the bounds provided in Section \ref{sec:bounds} are optimal in the following sense:

\counterexampleSet*

\counterexampleMfld*

To prove Propositions \ref{prop:counterexample_set} and \ref{prop:counterexample_mfld}, we construct the set $\Su$, the manifold $\M$, and the corresponding samples in Examples~\ref{example:set} and~\ref{example:mfld}, respectively. Due to rescaling it suffices to construct sets of reach $\rch(\Su)=\reach=1$.

\begin{remark}
  {For the proof of Proposition \ref{prop:counterexample_set}, we
    construct a set $\Su$ that is a subset of $\R^2$. For the
    proof of Proposition \ref{prop:counterexample_mfld}, the set $\M$
    is a surface in $\R^3$. Incidentally, both sets are
    two-dimensional.}  As mentioned in the introduction, we expect
  that better bounds than \eqref{equation:BoundOnR0} and
  \eqref{equation:BoundOnEpsilon2} can be obtained for one-dimensional
  sets {in $\R^d$ with $d \geq 2$}, i.e., curves, possibly with boundary.
\end{remark}

\begin{remark}\label{remark.WhenDeltaGreaterThanEpsilon}
When $\delta \geq \varepsilon$, which in Figure~\ref{fig:graphCcr} corresponds to the area above the diagonal $\delta=\varepsilon$,
the same bound is  optimal whether the set is assumed to be a manifold or not.
Indeed, in this case the union of annuli  $\Su$ in Example~\ref{example:set} can be replaced by a union of circles, namely the inner boundaries of the annuli.
Thus, the bound is tight for manifolds,
 including one-dimensional submanifolds in $\R^2$.
\end{remark}


\begin{remark} \label{Rem:DiscreteSample}
To simplify the analysis, the samples $P$ in our examples are continuous and therefore have an infinite number of points. However these samples can be approximated arbitrarily well by finite sets because they are compact.
\end{remark} 
\begin{SketchProof}[of Remark \ref{Rem:DiscreteSample}]
To pass to a finite sample, we first note that failing the bounds on the sampling parameters (in Propositions 5.2, 5.3, 6.4 and 6.5) is an open condition, i.e. for every $(\varepsilon, \delta)$ we can find an $(\varepsilon', \delta)$ with $\varepsilon' < \varepsilon$ such that  $(\varepsilon', \delta)$ still fail the bounds. To construct an example for a given $(\varepsilon, \delta)$ we take the example (Example  \ref{example:set}, \ref{example:mfld}, \ref{example:set_Riemann}, and \ref{example:mfld_Riemann} respectively) for $(\varepsilon', \delta)$ and take a subsample of $P$ that is so dense that the one-sided Hausdorff distance is $\varepsilon$.  Using the notation introduced in the Examples \ref{example:set}, \ref{example:mfld}, \ref{example:set_Riemann}, and \ref{example:mfld_Riemann} we can give a more precise description of the finite sample. For sets of positive reach the finite subsample can be chosen as follows: The circle $C_i$ should be densely subsampled such that the subsample contains $q_i$. The points $p_i$ and $\tilde{p}_i$ can remain as is. For the manifolds the finite subsample of can be chosen as follows: The trimmed torus $C_i$ should be densely subsampled such that the subsample contains $q_i$ and $\tilde{q}_i$. Also in this context, the points $p_i$ and $\tilde{p}_i$ can remain as is. Because the samples contain $p_i$ and $\tilde{p}_i$ and $q_i$ ($q_i$ and $\tilde{q}_i$ respectively) (most of) the spurious cycles we examined in the Examples \ref{example:set}, \ref{example:mfld}, \ref{example:set_Riemann}, and \ref{example:mfld_Riemann} remain the same. The only change that may occur for large $r$ in the proof of Proposition 5.5 is that for the interval $r \in [r_{i-1} ,r_i )$ spurious $2$-cycles may be interchanged for spurious $1$-cycles. Of course for small $r$ there are many more connected components and cycles because of the discrete approximation  than in the continuous examples.  With these finite subsamples $P$ one still finds that the homology of $P^{\boxplus r}$ is never the same as the underlying space.
\end{SketchProof}

\subsubsection{Sets of positive reach}\label{sec:TightnessSetsOfPosReach}

The construction of the set proving Proposition
\ref{prop:counterexample_set} goes as follows.

\begin{example}\label{example:set}
We define $\Su$ to be a 
union of annuli $A_i$ in
$\R^2$, each of which has inner radius $1$ and outer radius
$1+2 \varepsilon$. We lay the annuli in a row at distance at least 2 away from each other. 
{Due to this assumption, the reach of the set $\Su$ equals 1.} We number the annuli from $i=0$. Later we will see that the number of annuli that we need for the construction is finite. 

The sample $P$ consists of circles $C_i$ of radius $1+\varepsilon$ lying
in the middle of the annuli ($C_i \subseteq A_i$), and pairs of points
$\{p_i,\tilde{p}_i\}$.
Each pair $\{p_i,\tilde{p}_i\}$ lies
in the disk inside the annulus $A_i$, at a distance $\delta$ from
$A_i$, and the two points lie at a distance $2r_i$ from each other; see Figure \ref{fig:acute}, left. The bisector of $p_i$
and $\tilde{p}_i$ intersects the circle $C_i$ in two points. We let
$q_i$ be the intersection point that is closest to $p_i$ (and thus
$\tilde{p}_i$). We denote the circumradius of $p_i \tilde{p}_i
q_i$ by $R_i$ and note that $R_i \geq r_i$.
{Before explaining how we pick the sequence of $r_i$, we state a lemma which is
key for the construction:}

\begin{lemma}
  \label{lemma:acute}
  If $\varepsilon$ and $\delta$ fail to satisfy bound~\eqref{equation:BoundOnR0}, that is, $\varepsilon + \sqrt{2} \, \delta >  (\sqrt{2}  - 1),$ then
  \begin{itemize}
  \item the triangle $p_i \tilde{p}_i q_i$ is strictly acute;
  \item there exists a constant $c > 0$, depending only on $\delta$ and $\varepsilon$, such that $R_i - r_i \geq c \, r_i$.
  \end{itemize}
\end{lemma}

\begin{remark}
	Indeed, the triangle $p_i \tilde{p}_i q_i$ is per construction strictly acute if and only if $\varepsilon + \sqrt{2} \, \delta >  (\sqrt{2}  - 1)$.
\end{remark}
	
\begin{proof}
  The situation is illustrated in Figure \ref{fig:acute}, right. Let $z_i$ be the centre of $C_i$ and let
  $C'_i$ be the circle centred at $z_i$ with radius $1-\delta$. By
  construction, $C_i'$ passes through $p_i$ and $\tilde{p}_i$, while
  $C_i$ passes through $q_i$. Without loss of generality, we may
  assume that $p_i$ and $\tilde{p}_i$ lie on a vertical line, with
  $p_i$ above the segment $z_iq_i$ and $\tilde{p_i}$ below it.  Since $p_i\tilde{p}_i q_i$ is an isosceles triangle, it
  is acute if $\angle z_i q_i p_i < \frac{\pi}{4}$. The angle $\angle
  z_i q_i p_i$ is maximized when $p_i$ reaches the position $p^*_i$ on
  $C'_i$ --- in this position, the line through $q_i$ and $p_i$ is tangent to the
  circle $C'_i$. Using Condition~\eqref{equation:BoundOnR0}, we obtain that
  \[
  \sin \angle z_i q_i p_i \leq \sin \angle z_i q_i p^*_i = \frac{1-\delta}{1+\varepsilon} < \frac{1}{\sqrt{2}} = \sin \frac{\pi}{4}.
  \]
  Thus, $\angle z_i q_i p_i < \frac{\pi}{4}$ and therefore $p_i
  \tilde{p}_i q_i$ is acute. Because of the strict inequality in
  the above equation, we can find a small angle, say $\varphi = 2\left( \frac{\pi}{4} - \arcsin  \frac{1-\delta}{1+\varepsilon}  \right) > 0$, such that
  $\angle z_i q_i p_i \leq \frac{\pi}{4} - \frac{\varphi}{2}$. Since $R_i
  = \frac{r_i}{\sin \angle p_i q_i \tilde{p}_i }$, we deduce that
  \[
  R_i - r_i = \left( \frac{1}{\sin\angle p_i q_i \tilde{p}_i} - 1 \right) r_i \geq \left( \frac{1}{\cos \varphi} - 1 \right) r_i,
  \]
  which, after setting $c= \frac{1}{\cos \varphi} - 1$, proves the second item of the lemma.
\end{proof}

\begin{figure}[!h]
  \def\svgwidth{1.3\linewidth}
  \centering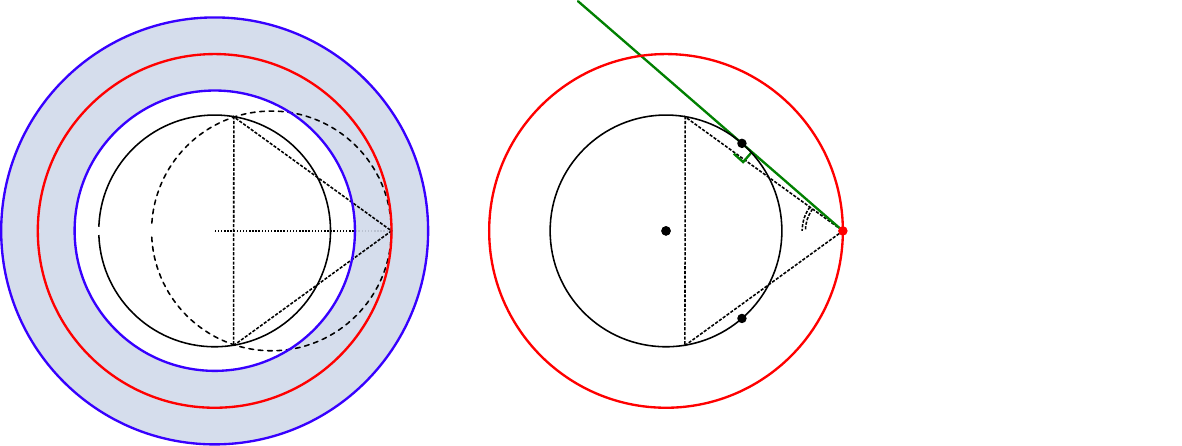
  \caption{Left: Each annulus $A_i$  is sampled by a circle $C_i$ and a pair of points $\{p_i,\tilde{p}_i\}$. Right: Notation for the proof of Lemma \ref{lemma:acute}. If
    Condition~\eqref{equation:BoundOnR0} 
		fails,
		then $\sin \angle
    z_iq_ip_i \leq \sin \angle
    z_iq_ip_i^* = \frac{1-\delta}{1+\varepsilon} < \frac{1}{\sqrt{2}}
    = \sin \frac{\pi}{4}$, and triangle $p_i \tilde{p}_i q_i$ is
    guaranteed to be acute.
    \label{fig:acute}}
\end{figure}

{We are now ready to} define the distance between each
pair of points $p_i$ and $\tilde{p}_i$ in an inductive manner.  We set $r_0 = \frac{\delta+
  \varepsilon}{2}$ and, for $i \geq 0$, 
	\[r_{i+1}=\begin{cases}
R_i, &\text{if }R_i < 1-\delta,\\
1-\delta, &\text{otherwise}.
\end{cases}\]
We stop the
sequence at the first value of $i$ such that $r_i = 1-\delta$.

Assume that $\varepsilon$ and $\delta$ fail to satisfy bound~\eqref{equation:BoundOnR0}. By Lemma~\ref{lemma:acute}, $r_{i+1} - r_{i}$ is lower bounded by a positive constant
that only depends on $\delta$ and $\varepsilon$,
\[
r_{i+1} - r_{i} = R_i-r_i\geq c \, r_i\geq c\, r_0.
\]
Hence, the sequence
of $r_i$ reaches the value $1-\delta$ in a finite number of steps.
Let $k$ be the index at which $r_k =
1-\delta$. Our constructed set $\Su$ consists of the finitely many annuli
$A_0 \cup A_1 \cup \ldots \cup A_k$ and our sample $P$ is defined as
$\bigcup_{0 \leq i \leq k} C_i \cup \{p_i,\tilde{p}_i \}$.
\end{example}

\begin{figure}[h!]
	\centering
	\begin{minipage}[b]{0.8\textwidth}
		\centering
		\includegraphics[width=0.99\textwidth]{pictures/CombinedA4label}
		\subcaption{For all $r < r_0$, the union of balls $(C_i \cup \{p_i,\tilde{p}_i\})^{\boxplus r} 
		$ has three connected components.}%
		\label{fig:tight_set_tiny_r}
	\end{minipage}
	\newline
	\begin{minipage}[b]{0.8\textwidth}
		\centering
		\includegraphics[width=0.99\textwidth]{pictures/Crit1label}
		\subcaption{At radius $r_1$, the cycle in the union of balls $(C_0 \cup \{p_0,\tilde{p}_0\})^{\boxplus r}$ at the annulus $A_0$ dies, while a cycle is created in the union of balls $(C_1 \cup \{p_1,\tilde{p}_1\})^{\boxplus r}  
		$ at the annulus $A_1$. }%
		\label{fig:tight_set_small_r}
	\end{minipage}
	\newline
	\begin{minipage}[b]{0.8\textwidth}  
		\centering 
		\includegraphics[width=0.99\textwidth]{pictures/Crit2label}
		\subcaption{At radius $r_2$, the cycle in the union of balls at the annulus $A_1$ dies, while a cycle is created in the union of balls at the annulus $A_2$.}%
		\label{fig:tight_set_large_r}
	\end{minipage}
	\newline
	\begin{minipage}[b]{0.38\textwidth}
		\centering
		\includegraphics[width=0.65\textwidth]{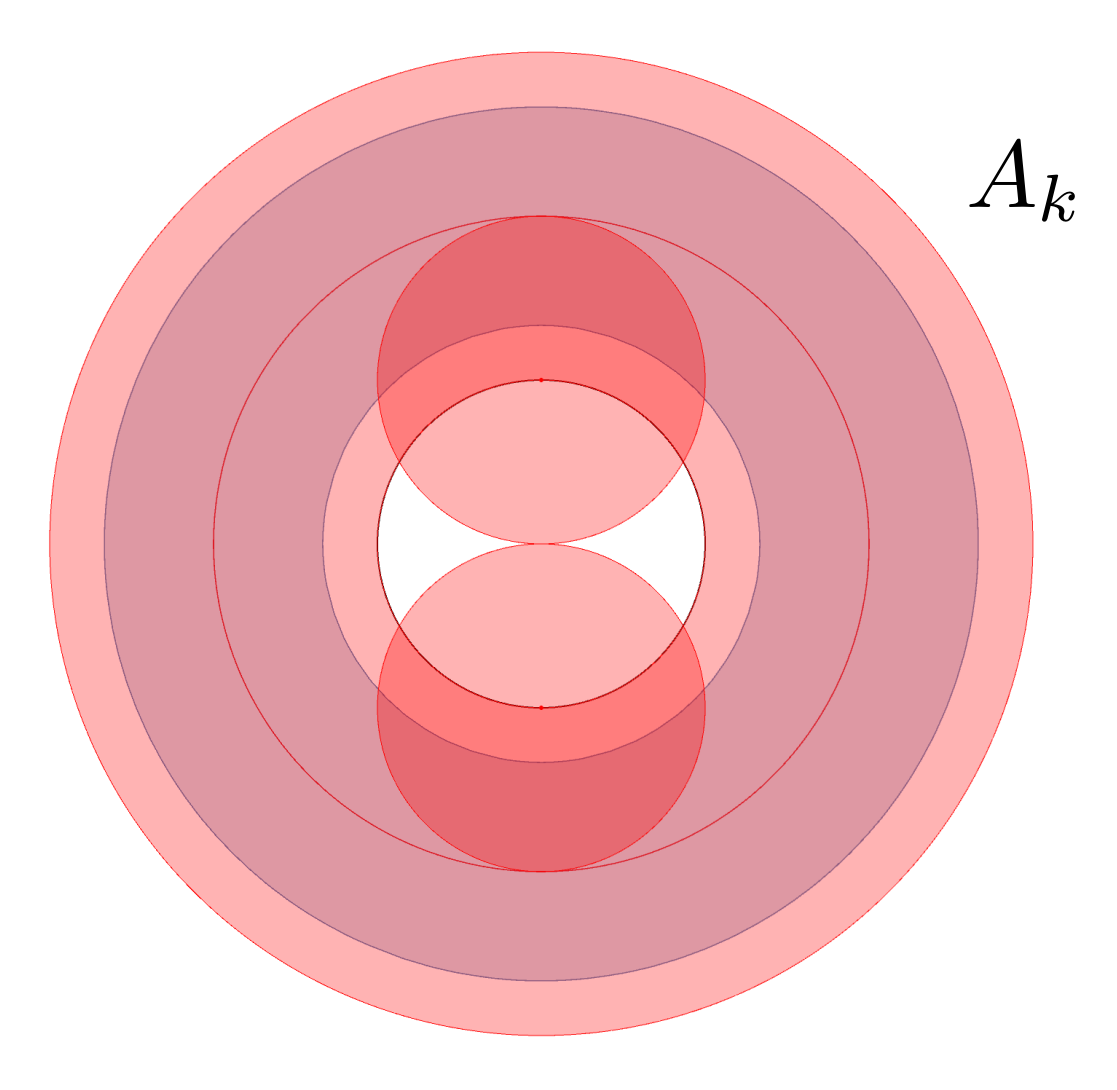}
		\subcaption{The set $(C_k \cup \{p_k,\tilde{p}_k\})^{\boxplus r}$ 
		at radius $r_k=1-\delta$. The two `gaps' are identical.}%
		\label{fig:A_k_1}
	\end{minipage}
	\quad
	\begin{minipage}[b]{0.38\textwidth}
		\centering
		\includegraphics[width=0.65\textwidth]{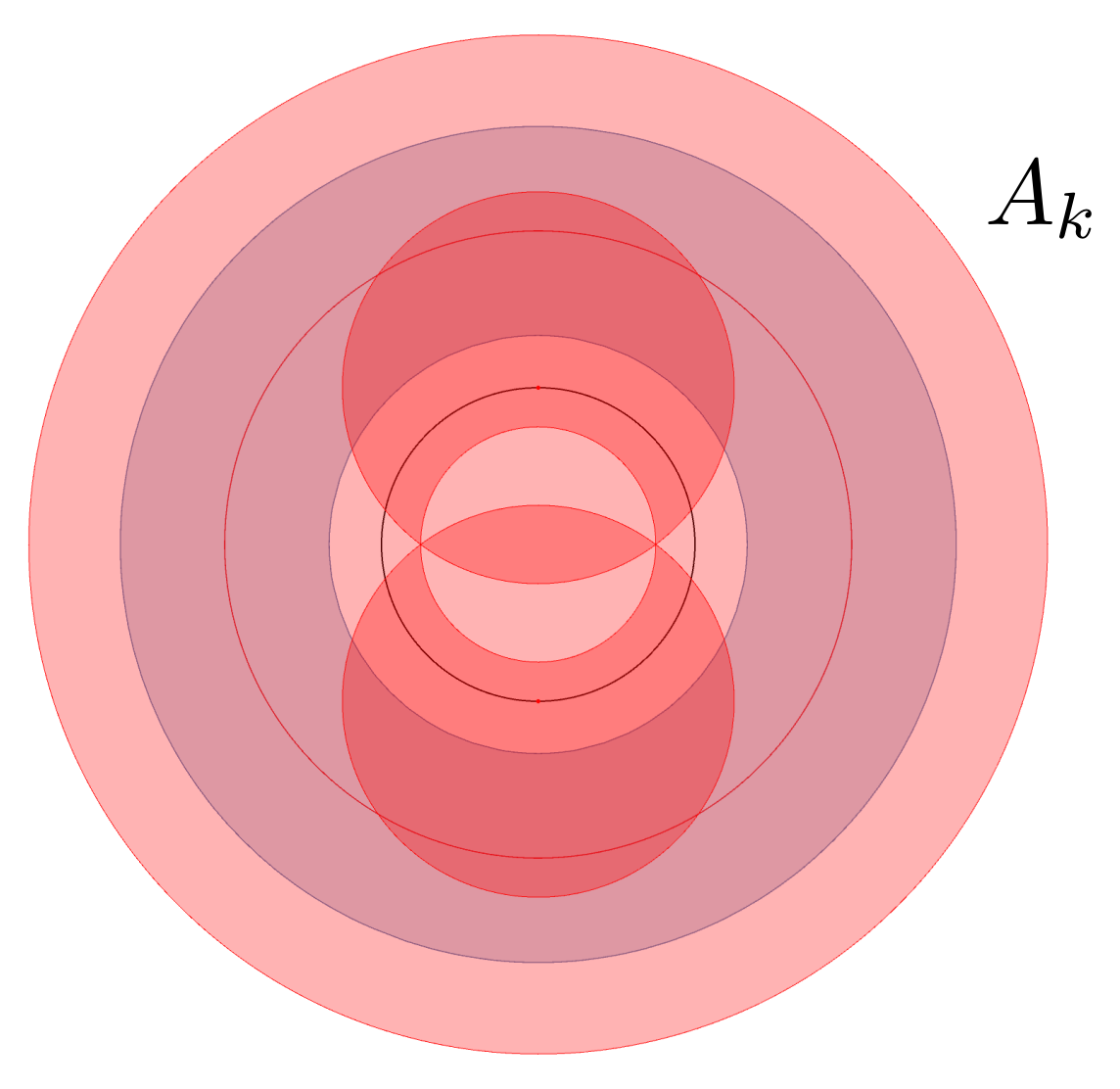}
		\subcaption{The two `gaps' of the set $(C_k \cup \{p_k,\tilde{p}_k\})^{\boxplus r}$
		disappear simultaneously.}%
		\label{fig:A_k_2}
	\end{minipage}
	\caption{ The changing homology of the set $P^{\boxplus r}$ 
	in the annuli $A_0, A_1, A_2$, and $A_k$. The set $\Su = A_0\cup \dots A_k$ is coloured light blue, the union of balls $P^{\boxplus r}$ 
	in pink. In black we depict the circles of radius $1-\delta$.}
	\label{fig:tight_set_nbhds}
\end{figure}

\begin{proof}[Proof of Proposition \ref{prop:counterexample_set}]
We show that for any $r\geq 0$, the union of balls $P^{\boxplus r}$ 
has different homotopy --- and even different homology --- than the set $\Su$. We first describe the development of the homotopy of the sets $(C_i \cup \{p_i,\tilde{p}_i\})^{\boxplus r}$ 
as $r$ increases:
\begin{itemize} 
\item For $r \in [0 ,r_0 )$, each set $(C_i \cup \{p_i,\tilde{p}_i\})^{\boxplus r}$ 
has three connected components, as illustrated in Figure \ref{fig:tight_set_tiny_r}. 
The three components merge into one at $r=r_0$, as the two balls $\{p_i\}^{\boxplus r}$ 
and $\{\tilde{p}_i\}^{\boxplus r}$
intersect the set $C_i ^{\boxplus r}$. 
\item For $r \in [r_{i}, r_{i+1} )$, the set $(C_i \cup \{p_i,\tilde{p}_i\})^{\boxplus r}$ has the homotopy type of two circles that share a point (also known as a wedge of two circles or a bouquet), as illustrated in Figures~\ref{fig:tight_set_small_r} and~\ref{fig:tight_set_large_r}. 
The smaller `gap' creating the additional cycle appears when $r=r_i$. 
Since, due to Lemma \ref{lemma:acute}, the triangle $p_i \tilde{p}_i q_i$ is acute, the `gap' persists until $r=R_i=r_{i+1}$. 
All sets $(C_j \cup \{p_j,\tilde{p}_j\})^{\boxplus r}$ with $j\neq i$ have the homotopy type of a circle. 
\item {At $r = r_k = 1-\delta$, all sets $(C_i \cup \{p_i,\tilde{p}_i\})^{\boxplus r}$ have the homotopy type of a circle but the last one, $(C_k \cup \{p_k,\tilde{p}_k\})^{\boxplus r}$, which has the homotopy type of two circles that share a point (see Figure \ref{fig:A_k_1}). 
Unlike the other cases, however, the `gaps' in the set $(C_k \cup \{p_k,\tilde{p}_k\})^{\boxplus r}$ are identical, and disappear simultaneously at $r = R_k \left(=\tfrac{(1+\varepsilon)^2+(1-\delta)^2}{2(1+\varepsilon)}\right)$ (Figure \ref{fig:A_k_2}). For larger $r$, the set $(C_k \cup \{p_k,\tilde{p}_k\}) ^{\boxplus r}$ is contractible.}
	\end{itemize} 
{Each annulus $A_i\subseteq \Su$ has the homotopy
	type of a circle, and thus the dimensions of the homologies of the set $\Su$ equal
	\[\dim\left(H_0(\Su)\right) = k+1, \qquad\dim\left(H_1(\Su)\right) = k+1. \] The dimensions of the homologies of the set $P ^{\boxplus r}$ are recorded in the table below.
	
	\begin{center}
      \begin{edtable}{tabular}{|c || c | c |c| }
			\hline
			~ & $r\in[0,r_0)$ & $r\in\left[r_0, R_k\right)$ & $r\geq R_k$\\ \hline
			$\dim\left(H_0\left(P ^{\boxplus r}\right)\right)$ & $3(k+1)$ & $k+1$ & $\leq k+1$\\ \hline
			$\dim\left(H_1\left(P ^{\boxplus r}\right)\right)$ & $k+1$ & $k+2$ & $\leq k$\\
			\hline
        \end{edtable}
	\end{center}
	
	One sees that the set $\Su$ never has the same homology as the union of balls $P ^{\boxplus r}$, and thus the two never have the same homotopy.
}
\end{proof}

\subsubsection{Manifolds}
\label{sec:optimality_mflds}
The construction of the manifold proving Proposition \ref{prop:counterexample_mfld} goes as follows:

\begin{example}\label{example:mfld}
  We define $\M$ to be a union of tori of revolution $T_i$ in $\R^3$. Each of these tori is the $1$-offset of a circle of radius $2$ in $\R^3$.
Put differently, each $T_i$ is --- up to Euclidean transformations --- the surface of revolution of a circle of radius $1$ in the $xz$-plane, centred at $(2,0,0)$, around the $z$-axis.
The set $T_i$ is illustrated in blue in Figures~\ref{fig:TightManifold1N} and~\ref{fig:self-centred}. 

We number the tori from $i=0$, and lay them out in a row at a distance at least $2$ apart from one another. Due to this assumption, the reach of $\M$ equals 1. Later we will see that the number of tori that we need for the construction is finite. 

The sample $P$ consists of sets $C_i$ which are tori with a part cut out, and pairs of points $\{p_i,\tilde{p}_i\}$ lying inside the hole of each torus $T_i$.
To construct each set $C_i$ we take the $\delta$-offset of
  $T_i$, keep the part that lies inside the solid
  torus bounded by $T_i$, and remove an  $\varepsilon$-neighbourhood of the circle obtained by revolving the point $(1,0,0)$ around the $z$-axis; see the red set in Figures~\ref{fig:TightManifold1N} and~\ref{fig:self-centred}. In other words, each $C_i$ is the set difference between the torus obtained by rotating the circle of radius $1-\delta$ centred in the $xz$-plane at $(2,0,0)$, and the open solid
  torus obtained by rotating the open disc of radius $\varepsilon$ centred in the
  $xz$-plane at $(1,0,0)$.

\begin{figure}[!h]
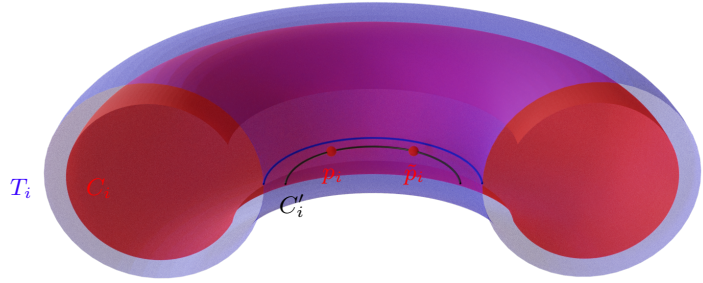

  	\centering
	\caption{
		The (half of the) torus $T_i$ depicted in blue; the sample --- the set $C_i$ and the points $p_i$ and $\tilde{p}_i$ --- in red.
		In black we indicate the circle $C'_{i}$. The closest point projection of this circle onto $\M$ is indicated in blue. 
	}
	\label{fig:TightManifold1N}
\end{figure}

 Let $C'_{i}$ be the circle found by revolving the point $(1-\delta,0,0)$ around the $z$-axis. Each pair of points, $p_i$ and $\tilde{p}_i$, lies on $C'_{i}$ at a distance $2 r_i$ from each other. Let {$q_i$ and $\tilde{q}_i$} be the two points in the intersection of the bisector of {$p_i$ and $\tilde{p}_i$} and the set $C_i$ that lie closest to  $p_i$ and $\tilde{p}_i$. Note that $q_i$ and $\tilde{q}_i$ lie on the boundary\footnote{Here we think of $C_i$ as a manifold with boundary. }  of $C_i$,
 and $\{ q_i,\tilde{q}_i \} = \pi_{C_i}\left(\frac{p_i + \tilde{p}_i}{2} \right) $. Denote the circumradius of the simplex $p_i \tilde{p}_i q_i \tilde{q}_i$ by $R_i$. 

 \begin{figure}[!h]
   \def\svgwidth{1.05\linewidth}
	\centering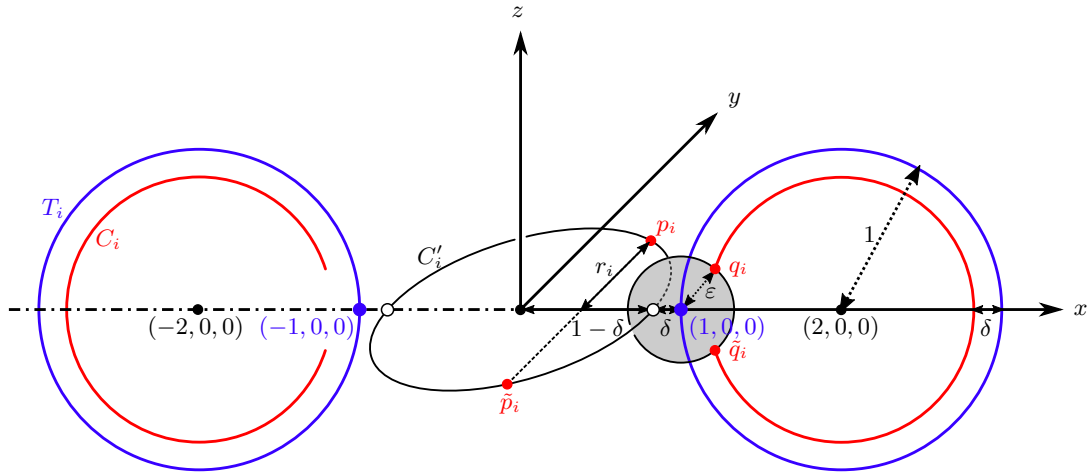
	\caption{The sets $T_i$, $C_i$ and $C'_i$ are obtained by rotating around the $z$-axis, respectively, the blue circles, the red arcs and the white point. 
}
	\label{fig:self-centred}
\end{figure}

As in Example \ref{example:set}, we define the distance $2r_i$ between each pair of points $p_i$ and $\tilde{p}_i$ inductively. We set the distance $r_0$ such that the balls $B(p_0, r)$ and $B(\tilde{p}_0,r)$ start to intersect each other at the same value of $r$ as the balls $B(q_0, r)$ and $B(\tilde{q}_0,r)$ start to intersect:
\[r_0 = \tfrac{1}{2}d\left(q_0, \tilde{q}_0\right). 
\]
We then define
	\[r_{i+1}=\begin{cases}
	R_i, &\text{if }R_i < 1-\delta,\\
	1-\delta, &\text{otherwise}.
	\end{cases}\]

We stop the sequence at the first value of $i$ such that $r_i = 1-\delta$.

Assume that $\varepsilon$ and $\delta$ fail to satisfy bound~\eqref{equation:BoundOnEpsilon2}. By Lemma
\ref{lemma:selfcentred}, $r_{i+1}$ is lower bounded by a positive constant
that only depends on $\delta$ and $\varepsilon$,
\[
r_{i+1}^2= R_i^2\geq r_i^2+ c^2\geq r_0^2 + i\cdot c^2.
\]
Hence, the sequence
of $r_i$ reaches the value $1-\delta$ in a finite number of steps.

Let $k$ be the index at which $r_k =
1-\delta$. Our constructed manifold $\M$ consists of the finitely many tori
$T_0 \cup T_1 \cup \ldots \cup T_k$, and our sample $P$ is defined as
$\bigcup_{0 \leq i \leq k} \left(C_i \cup \{p_i,\tilde{p}_i \}\right)$.
\end{example}

In the proof of Proposition \ref{prop:counterexample_set}, acuteness of triangles plays an essential role. In Lemma \ref{lemma:acute} we argue that if $\varepsilon$ and $\delta$ fail to satisfy Bound~\eqref{equation:BoundOnR0}, then any triangle $p_i\tilde{p}_i q_i$ is acute. Furthermore, a triangle is acute if and only if it contains its circumcentre. We generalize acuteness to simplices as follows:
\begin{definition}[Self-centred simplices, \cite{choudhary2020coxeter}] \label{Def:Self_centred}
	{A simplex is called (strictly) self-centred if it contains its circumcentre (in its interior).  }
\end{definition}
\begin{lemma} \label{lemma:selfcentred} 
If $\varepsilon$ and $\delta$ fail to satisfy bound~\eqref{equation:BoundOnEpsilon2}, {that is, $(1 - \delta)^2 - \varepsilon^2   <  4\sqrt{2}  - 5$, and $r_i$ satisfies 
	\begin{equation}\label{eq:lower_bound_distance_p_ptilde}
	2r_i\geq d(q_i,\tilde{q}_i),
\end{equation}}
then
\begin{itemize}
	\item the simplex $p_i \tilde{p}_i q_i \tilde{q}_i$ is strictly self-centred;
	\item there exists a constant $c > 0$, depending only on $\delta$ and $\varepsilon$, such that $R_i^2 \geq r_i^2+c^2$.
\end{itemize}
\end{lemma}

\begin{proof}
  A key observation is that the simplex $p_i \tilde{p}_i q_i
  \tilde{q}_i$ is {(strictly)} self-centred if and only if {the triangles $p_i
  \tilde{p}_i q_i$ and $q_i
  \tilde{q}_i p_i$ are (strictly)} acute.

To see this, assume without loss of
  generality that the torus $T_i$ is centred at the origin and that the points $q_i$ and $\tilde{q}_i$ lie in the
  $xz$-plane and have positive $x$-coordinates{, as in Figure~\ref{fig:self-centred}}. The circumcentre of a
  simplex is the intersection of the bisectors of pairs of its
  vertices. The circumcentre of the simplex $p_i \tilde{p}_i q_i
  \tilde{q}_i$ thus lies on the $x$-axis; indeed, the $x$-axis is the
  intersection of the bisector of $p_i$ and $\tilde{p}_i$, and the
  bisector of $q_i$ and $\tilde{q}_i$.
  
  Hence, $p_i \tilde{p}_i
  q_i \tilde{q}_i$ is strictly self-centred if and only if its circumcentre
  lies on {the intersection of the interior of $p_i \tilde{p}_i
  	q_i \tilde{q}_i$ with the $x$-axis ---} the open line segment connecting the midpoint $m_i =
  \frac{p_i+\tilde{p}_i }{2}$ 
  of $p_i$ and $\tilde{p}_i$ with the
  midpoint $n_i = \frac{q_i+\tilde{q}_i }{2}$ of $q_i$ and
  $\tilde{q}_i$.
  This happens precisely when the circumcentre of
  triangle $p_i \tilde{p}_i q_i$ (resp. $q_i \tilde{q}_i p_i$) lies on the open segment connecting
  $m_i$ to $q_i$ (resp. $n_i$ to $p_i$), in other words, when both triangles $p_i \tilde{p}_i q_i$ and $q_i \tilde{q}_i p_i$ are
  strictly acute. We illustrate the two extreme cases in Figure \ref{fig:self-centred-extreme-cases}.
  
  {We prove the fact that the two triangles are indeed strictly acute in Claim~\ref{claim:triangle1_acute} below.}

    \begin{figure} 
      \centering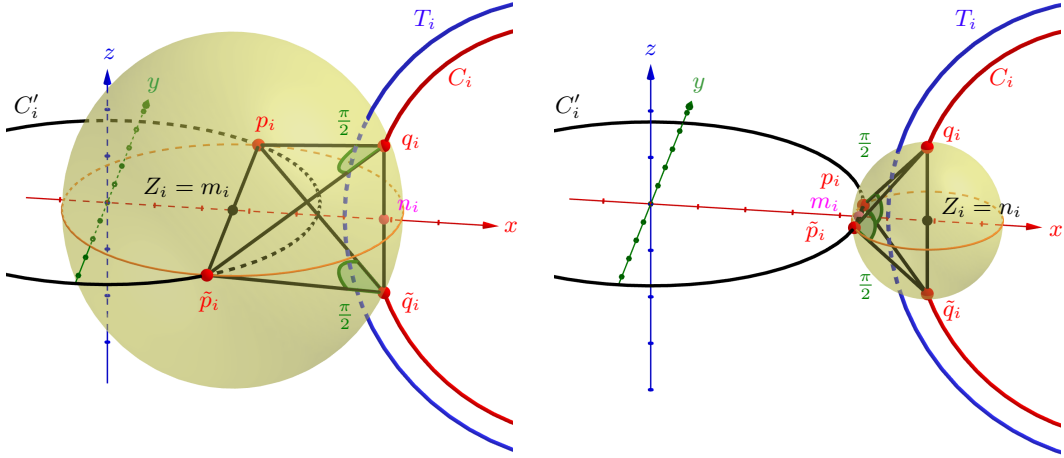
      \caption{ {When both triangles $p_i\tilde{p}_i q_i$ and $q_i
        \tilde{q}_i p_i$ are strictly acute, the circumcentre $Z_i$ of
        tetrahedron $p_i \tilde{p}_i q_i \tilde{q}_i$ lies on the open
        segment connecting $m_i$ to $n_i$. 
        When the triangle $p_i
        \tilde{p}_i q_i$ becomes
        right-angled, $Z_i$ reaches $m_i$ (on the left). When the triangle $q_i
        \tilde{q}_i p_i$ becomes
        right-angled, $Z_i$ reaches $n_i$ (on the right).}
        \label{fig:self-centred-extreme-cases}
      }
    \end{figure}
  
    Recall that both the circumcentre of the simplex $p_i \tilde{p}_i q_i \tilde{q}_i$ and the point $m_i$ lie on the $x$-axis. Let $u$ be the $x$-coordinate of the circumcentre. We have shown that, for all distances $r_i\in\left[\tfrac{1}{2}d(q_i,\tilde{q}_i), 1-\delta\right]$ defining the position of the points $p_i$ and $\tilde{p}_i$, the circumcentre lies further away from the origin than the midpoint $m_i$. 
    That is, $u-\norm{m_i}>0$. 
    Since $\left[\tfrac{1}{2}d(q_i,\tilde{q}_i), 1-\delta\right]$ is compact, there exists a constant $c$ such that
    \[u-\norm{m_i}\geq c.\]
    The triangle with vertices $p_i, m_i$, and the circumcentre is right-angled, with edge lengths $r_i, u-\norm{m_i}$, and the hypotenuse $R_i$. Thus,
    \[
    R_i^2 = r_i^2+{(u-\norm{m_i})}^2\geq r_i^2+c^2.
    \]
\end{proof}

\begin{claim}\label{claim:triangle1_acute}
	The triangles $p_i \tilde{p}_i q_i$ and $q_i \tilde{q}_i p_i$ are
	strictly acute, under the assumptions of Lemma \ref{lemma:selfcentred}. 
\end{claim}

\begin{proof}

\begin{figure}[!h]
	\centering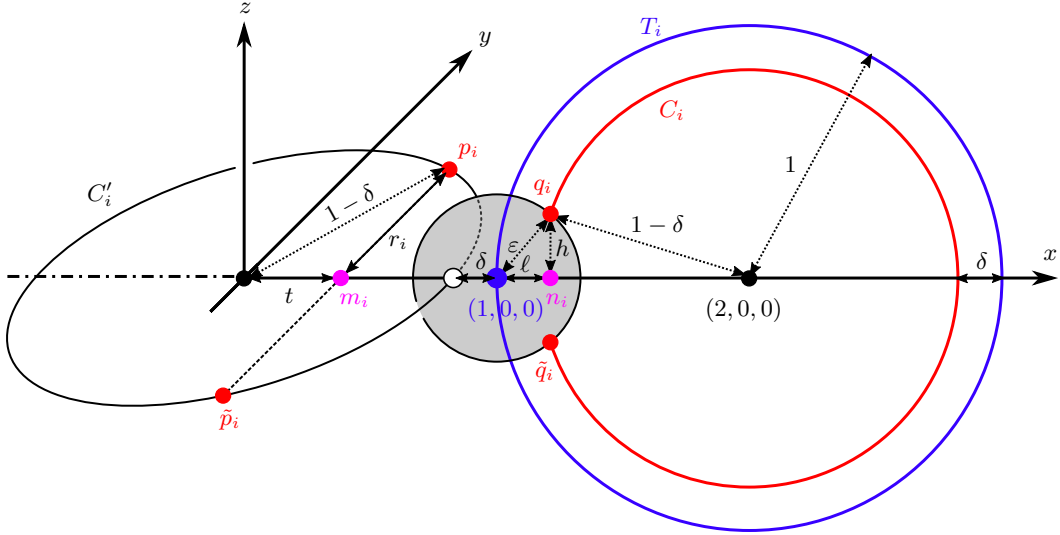
	\caption{Zoom-in of Figure \ref{fig:self-centred} with the notation used in the proof of Claim \ref{claim:triangle1_acute}.
	} 
	\label{fig:self-centred-zoom-in}
\end{figure}

Let $t\geq 0$ be the $x$-coordinate of $p_i$, and $\ell$ and $h$ define the $x$- and {$z$}-coordinates of $q_i$,
\[p_i = (t,r_i,0), \qquad {q_i = (1+\ell, 0, h)}.\]
Then $m_i = (t,0,0)$ and $n_i = (1+\ell,0,0)$. 	{We refer the reader to Figure \ref{fig:self-centred-zoom-in} for an overview of the notation.}

Due to the Pythagorean theorem,
\[\varepsilon^2 - \ell^2 = h^2 = (1-\delta)^2 - (1-\ell)^2,\]
and thus
\begin{equation*}
	\ell = \frac{\varepsilon^2 - \delta^2 + 2\delta}{2} \quad \text{and} \quad h = \sqrt{\varepsilon^2-\ell^2}.
\end{equation*}
Furthermore, due to Equation \eqref{eq:lower_bound_distance_p_ptilde}, 
\[r_i\geq \tfrac{1}{2}d(q_i,\tilde{q}_i) = h.\]
Note that the positions of both points $p_i$ and
$\tilde{p}_i$ on the circle $C_i'$ are completely determined
by $r_i$ (the $y$-coordinate of $p_i$). For the
purpose of the proof, we use
the $x$-coordinate $t= \sqrt{(1-\delta)^2 - r_i^2}$ of $p_i$ to parametrize the positions of
$p_i$ and $\tilde{p}_i$. Hence, showing that triangles $p_i
\tilde{p}_i q_i$ and $q_i \tilde{q}_i p_i$ are acute for all
$r_i \in [h,1-\delta]$ translates into showing that they
are acute for all $0\leq t \leq \sqrt{(1-\delta)^2 - h^2}=1-\ell$.

The triangle $p_i \tilde{p}_i q_i$ is isosceles. It is thus strictly acute if and only if its height, $\norm{m_i - q_i}$, is larger than half the length of its base, $\norm{m_i-p_i} = r_i$.
We obtain:
\begin{align}\label{eq:triangle1_inequality}
	&r_i^2 < \norm{m_i - q_i}^2\nonumber\\
	\iff&   (1-\delta)^2 - t^2 < (1+\ell-t)^2 + h^2\nonumber\\
	\iff & 0 < 2t^2 - 2t(1+\ell)+4\ell.
\end{align}
Let $Q(t) = 2t^2 - 2t(1+\ell)+4\ell$ be the quadratic form from the inequality \eqref{eq:triangle1_inequality}, and $\Delta$ be its reduced discriminant,
\[
\Delta = (1+\ell)^2 - 8 \ell = (\ell - 3 - 2\sqrt{2})(\ell - 3 + 2\sqrt{2}).
\]
The inequality \eqref{eq:triangle1_inequality} holds for all $t\in \left[0, 1-\ell\right]$ if and only if
\begin{itemize}
	\item either $\Delta<0$, and thus $Q(t)>0$ for all $t$, or
	\item $\Delta\geq 0$ and the interval $[t_1,t_2]\ni t$ for which $Q(t)\leq 0$, is disjoint from the interval $\left[0, 1-\ell\right]$.
\end{itemize}
$\Delta<0$ if and only if $3-2\sqrt{2}<\ell<3+2\sqrt{2}$. Substituting $2\ell = \varepsilon^2 - (1-\delta)^2+1$ translates into 
\[5-4\sqrt{2}<\varepsilon^2-(1-\delta)^2< 5+4\sqrt{2}.\]
The first inequality holds by assumption. The second follows from the fact that $0\leq \delta\leq\varepsilon< 1$.

In summary, our assumptions imply that $\Delta<0$, implying that $Q>0$, and thus the triangle $p_i \tilde{p}_i q_i$ is strictly acute.

Similarly, the triangle $q_i \tilde{q}_i p_i$ is isosceles, and is thus strictly acute if and only if its height, $\norm{n_i - p_i}$, is larger than half the length of its base, $\norm{n_i-q_i} = h$. This indeed holds, since
\[\norm{n_i - p_i}^2 = (1+\ell-t)^2 + r_i^2 \geq (1+\ell-t)^2 + h^2 > h^2.\]
\end{proof}

\begin{proof}[Proof of Proposition \ref{prop:counterexample_mfld}]	
We show that for any $r\geq 0$, the union of balls $P^{\boxplus r}$ has different homotopy than the manifold $\M$. To achieve this, it suffices to show that their homologies differ.
	
	The manifold $\M$ consists of $k+1$ tori, and thus the dimensions of the homologies of $\M$ equal
	\[\dim\left(H_0(\M)\right) = k+1, \qquad\dim\left(H_1(\M)\right) = 2(k+1), \qquad\dim\left(H_2(\M)\right) = k+1. \] 
	
	We first have a look at the second homology of the set $(C_i \cup
	\{p_i,\tilde{p}_i\})^{\boxplus r}$. For $r<r_0$, as well as $r\geq 1-\delta$, the second homology of the set $(C_i \cup	\{p_i,\tilde{p}_i\})^{\boxplus r}$ is trivial for each $i$. In the former case (see Figure \ref{fig:counterex_radius_too_small}), the set $(C_i )^{\boxplus r}$ has not yet `closed up' to form a (thickened) torus. In the latter case
, the inside of the torus  $(C_i)^{\boxplus r}$ gets filled in. 


The filling in of the torus kills both a $2$-cycle and a $1$-cycle at the same time. This action possibly also creates spurious $2$-cycles (see Remark \ref{remark:ExtraSpurious} below). Nevertheless, there is never more than one spurious $2$-cycle per torus, which kills a $1$-cycle that is present in that torus (in the underlying space $\M$). Hence the first and second Betti numbers of the sample and the underlying space do not match up.
 We stress that these events can only occur if $r\geq 1-\delta$ because the symmetry axis of the torus (the $z$ axis in Figure \ref{fig:Example2Cycle}) does not intersect $P^{\boxplus r}$ when $r< 1-\delta$.

Thus, $\M$ and $P^{\boxplus r}$ have different homology for $r\in [0, r_0)\cup [1-\delta, \infty)$.

\begin{figure}[!h]
\begin{center}
		  \includegraphics[width=0.75\textwidth]{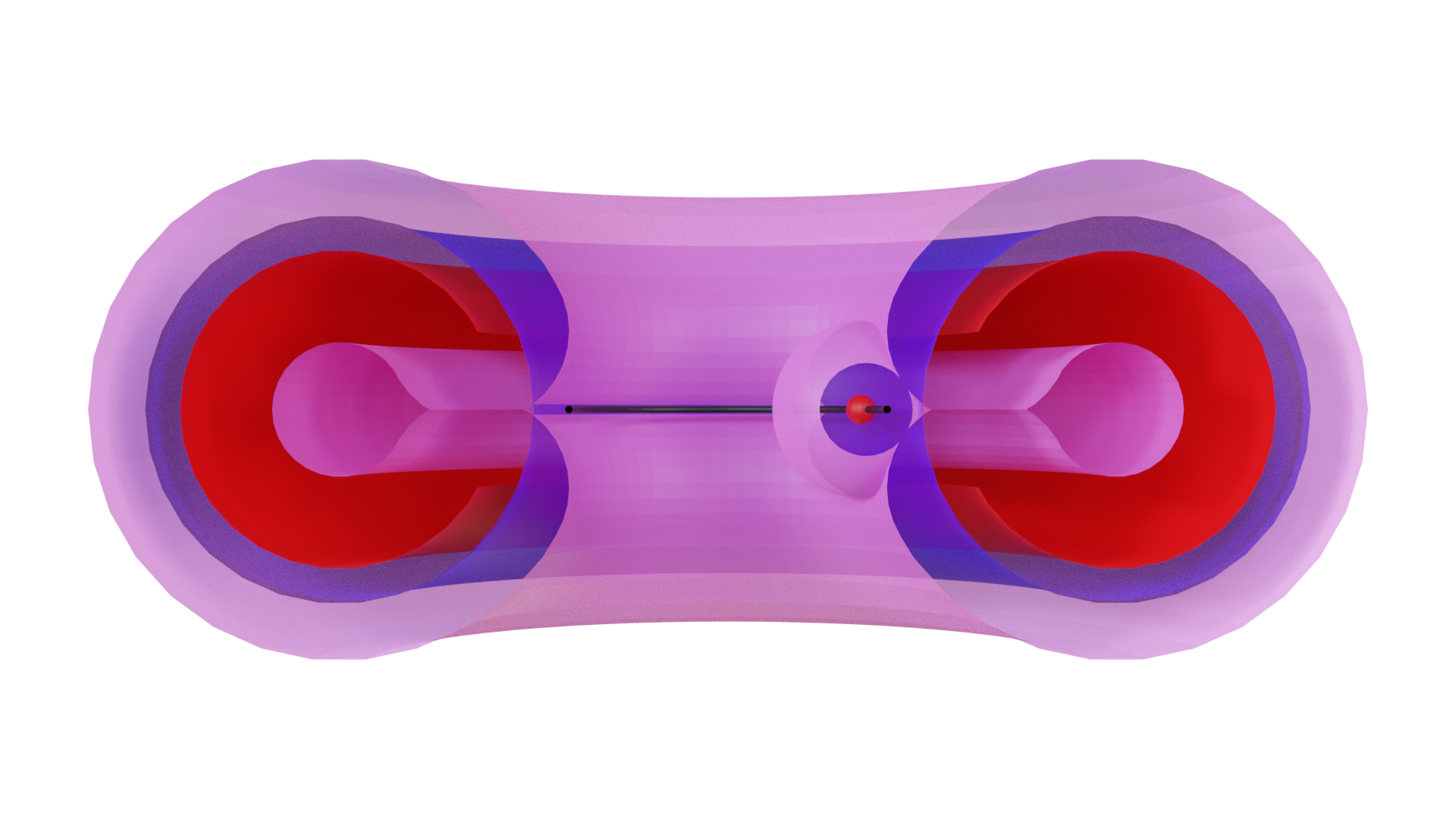} \\
      \includegraphics[width=0.5\textwidth]{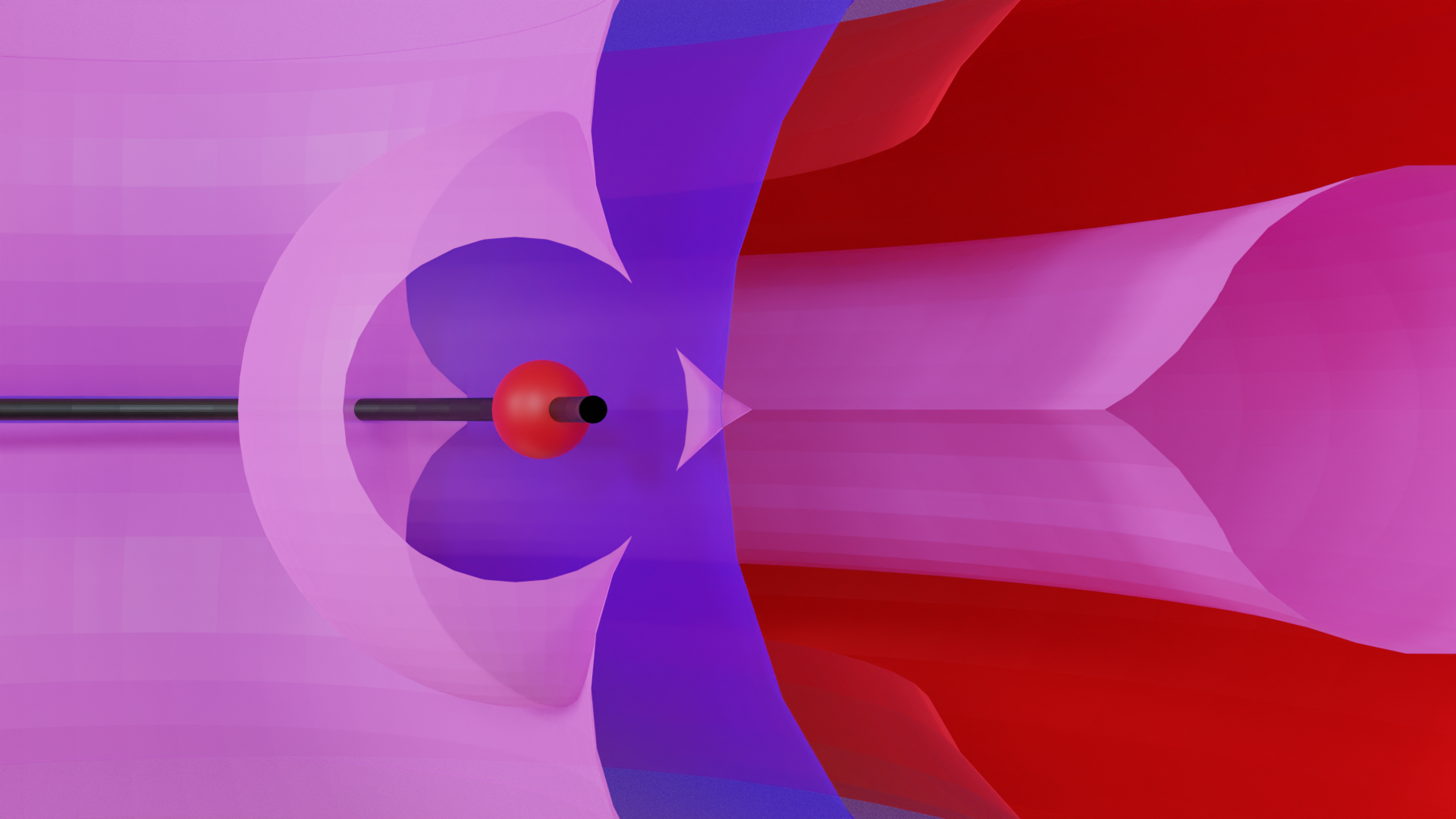}
\end{center}
\caption{The spurious 2-cycle that prevents the sample from having the same homology as the 
{manifold. The manifold $\M$ is depicted in blue, the sample $P$ in red, and the boundary of the thickening in pink. The cycle is clearly present in the zoomed-in image (bottom).}
}
\label{fig:Example2Cycle}
\end{figure}

	\begin{figure}[!h]
	  \begin{center}
\begingroup%
  \makeatletter%
  \providecommand\color[2][]{%
    \errmessage{(Inkscape) Color is used for the text in Inkscape, but the package 'color.sty' is not loaded}%
    \renewcommand\color[2][]{}%
  }%
  \providecommand\transparent[1]{%
    \errmessage{(Inkscape) Transparency is used (non-zero) for the text in Inkscape, but the package 'transparent.sty' is not loaded}%
    \renewcommand\transparent[1]{}%
  }%
  \providecommand\rotatebox[2]{#2}%
  \newcommand*\fsize{\dimexpr\f@size pt\relax}%
  \newcommand*\lineheight[1]{\fontsize{\fsize}{#1\fsize}\selectfont}%
  \ifx\svgwidth\undefined%
    \setlength{\unitlength}{359.75542971bp}%
    \ifx\svgscale\undefined%
      \relax%
    \else%
      \setlength{\unitlength}{\unitlength * \real{\svgscale}}%
    \fi%
  \else%
    \setlength{\unitlength}{\svgwidth}%
  \fi%
  \global\let\svgwidth\undefined%
  \global\let\svgscale\undefined%
  \makeatother%
  \begin{picture}(1,0.79275662)%
    \lineheight{1}%
    \setlength\tabcolsep{0pt}%
    \put(0,0){\includegraphics[width=\unitlength,page=1]{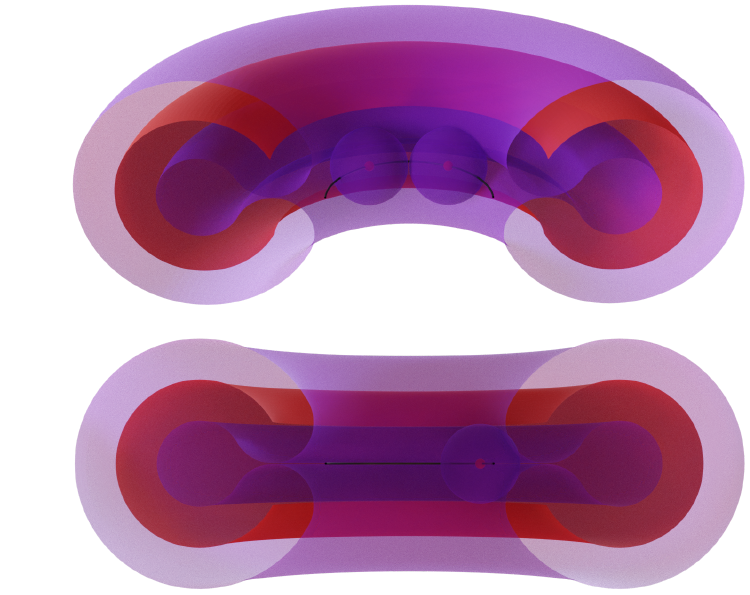}}%
    \put(0.13594174,0.61190122){\color[rgb]{1,0,0}\makebox(0,0)[lt]{\lineheight{1.25}\smash{\begin{tabular}[t]{l}$C_0$\end{tabular}}}}%
    \put(0.4765297,0.52948998){\color[rgb]{1,0,0}\makebox(0,0)[lt]{\lineheight{1.25}\smash{\begin{tabular}[t]{l}$p_0$\end{tabular}}}}%
    \put(0.58586952,0.53119844){\color[rgb]{1,0,0}\makebox(0,0)[lt]{\lineheight{1.25}\smash{\begin{tabular}[t]{l}$\tilde{p}_0$\end{tabular}}}}%
    \put(0.62762461,0.13388579){\color[rgb]{1,0,0}\makebox(0,0)[lt]{\lineheight{1.25}\smash{\begin{tabular}[t]{l}$p_0$\end{tabular}}}}%
    \put(-0.0025245,0.54356381){\color[rgb]{0.35294118,0.01568627,0.57254902}\makebox(0,0)[lt]{\lineheight{1.25}\smash{\begin{tabular}[t]{l}$(C_0)^{\boxplus r_0}$\end{tabular}}}}%
    \put(-0.0025245,0.16087457){\color[rgb]{0.35294118,0.01568627,0.57254902}\makebox(0,0)[lt]{\lineheight{1.25}\smash{\begin{tabular}[t]{l}$(C_0)^{\boxplus r_0}$\end{tabular}}}}%
    \put(0.13539229,0.24458781){\color[rgb]{1,0,0}\makebox(0,0)[lt]{\lineheight{1.25}\smash{\begin{tabular}[t]{l}$C_0$\end{tabular}}}}%
  \end{picture}%
\endgroup%

		\end{center}
		\caption{
			Two views of the situation at $r_0$. The sample $P$ in red and its thickening $P^{\boxplus r_0 } $ in purple. The balls $B(p_0, r_0)$ and $B(\tilde{p}_0,r_0)$ touch precisely and the thickened torus $(C_0)^{\boxplus r_0 } 
$ `closes up' and generates 2-homology.
		}
		\label{fig:counterex_radius_too_small}
	\end{figure}
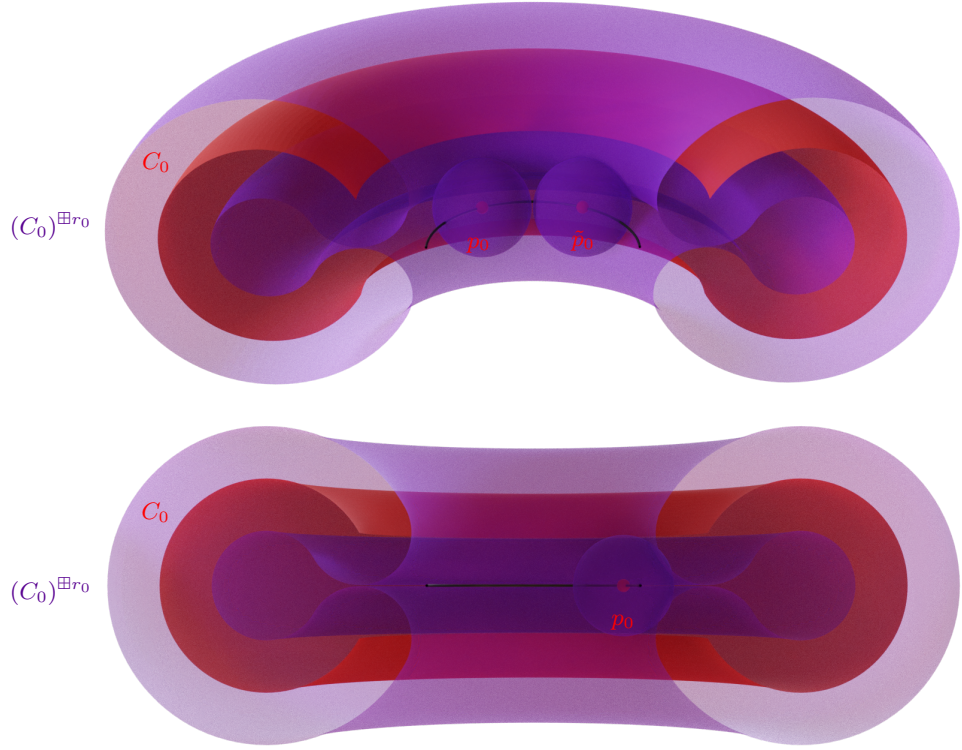

		For $r \in [r_{i-1} ,r_i )$, the set $(C_i \cup \{p_i, \tilde{p}_i\}) ^{\boxplus r}$ has the homotopy type of a torus with either at least a circle or 
	 a single 2-sphere attached,
	 depending on whether the radius $r$ is smaller or larger than the circumradius of the triangle $p_i \tilde{p}_iq_i$. 
		The smaller `gap' creating the additional 1-, and later 2-cycle appears when $r=r_i$. Since, due to Lemma \ref{lemma:selfcentred}, the simplex $p_i \tilde{p}_i q_i \tilde{q}_i$ is self-centred, the gap persists until $r=R_i=r_{i+1}$.
		
		All sets $(C_j \cup \{p_j,\tilde{p}_j\}) ^{\boxplus r}$ with $j\neq i$ have the homotopy type of a torus. Thus, for $r\in [r_0,1-\delta)$,	
\begin{align} 
\dim\left(H_1\left(P ^{\boxplus r}\right)\right) + \dim\left(H_2\left(P ^{\boxplus r} \right)\right) =  3(k+1)+1.
\label{ineq:HomologyThickSample}
\end{align}

In contrast, $\dim\left(H_1(\M)\right) +\dim\left(H_2(\M)\right) = 3(k+1),$ and thus the manifold $\M$ and the union of balls $P^{\boxplus r}$ have different homology also for $r\in [r_0,1-\delta)$.
\end{proof}

\begin{remark}\label{remark:ExtraSpurious} 
In the description of the spurious cycles we focused on the ones that lie near the circumcentre of the simplex $p_i \tilde{p}_i q_i \tilde{q}_i$. However, there are more spurious $2$-cycles to consider. 

Let us denote the mirror images of the points $q_i$ and $\tilde{q}_i$ in the $yz$ plane of Figure \ref{fig:self-centred} by $q_i'$ and $\tilde{q}_i'$.
Then these spurious $2$-cycles are located near the circumcentre of the simplex $p_i \tilde{p}_i q_i' \tilde{q}_i'$.  
The creation of each such spurious $2$-cycle kills a $1$-cycle --- exactly in the same way that the ``mirrored'' spurious $2$-cycle close to the simplex $p_i \tilde{p}_i q_i \tilde{q}_i$ does. The $1$-cycle that is killed matches the $1$-cycle in the torus that would persist after the torus is filled in.  
\end{remark}

\begin{figure}[ht]
  \begin{subfigure}[t]{0.99\textwidth}
    \def\svgwidth{1.1\linewidth}
    \centering\input{pictures/Annuli_row_0_annotated.pdf_tex}
\subcaption{At first, the balls around the points $p_i$ and $\tilde{p}_i$ do not intersect the thickening of the set $C_i$, and thus the number of connected components of the thickening (in pink) of $P$ is different from the number of components of the manifold. }
\end{subfigure}
\newline
\begin{subfigure}[t]{0.99\textwidth}
  \def\svgwidth{1.1\linewidth}
    \centering\input{pictures/Annuli_row_1_annotated.pdf_tex}
    \subcaption{Then we create a (or possibly multiple) spurious cycle(s) for the first torus in the sequence (on the left).  }
\end{subfigure}
\newline
\begin{subfigure}[t]{0.99\textwidth}
  \def\svgwidth{1.1\linewidth}
  \centering\input{pictures/Annuli_row_2_annotated.pdf_tex}
  \subcaption{{By the time the spurious cycles at the first torus have disappeared, others have been created at the second torus.} This process is then repeated for all tori in the sequence as $r$ increases.} 
\end{subfigure}
\begin{subfigure}[t]{0.99\textwidth}
  \def\svgwidth{1.1\linewidth}
  \centering
\begingroup%
  \makeatletter%
  \providecommand\color[2][]{%
    \errmessage{(Inkscape) Color is used for the text in Inkscape, but the package 'color.sty' is not loaded}%
    \renewcommand\color[2][]{}%
  }%
  \providecommand\transparent[1]{%
    \errmessage{(Inkscape) Transparency is used (non-zero) for the text in Inkscape, but the package 'transparent.sty' is not loaded}%
    \renewcommand\transparent[1]{}%
  }%
  \providecommand\rotatebox[2]{#2}%
  \newcommand*\fsize{\dimexpr\f@size pt\relax}%
  \newcommand*\lineheight[1]{\fontsize{\fsize}{#1\fsize}\selectfont}%
  \ifx\svgwidth\undefined%
    \setlength{\unitlength}{467.8673903bp}%
    \ifx\svgscale\undefined%
      \relax%
    \else%
      \setlength{\unitlength}{\unitlength * \real{\svgscale}}%
    \fi%
  \else%
    \setlength{\unitlength}{\svgwidth}%
  \fi%
  \global\let\svgwidth\undefined%
  \global\let\svgscale\undefined%
  \makeatother%
  \begin{picture}(1,0.2042638)%
    \lineheight{1}%
    \setlength\tabcolsep{0pt}%
    \put(0,0){\includegraphics[width=\unitlength,page=1]{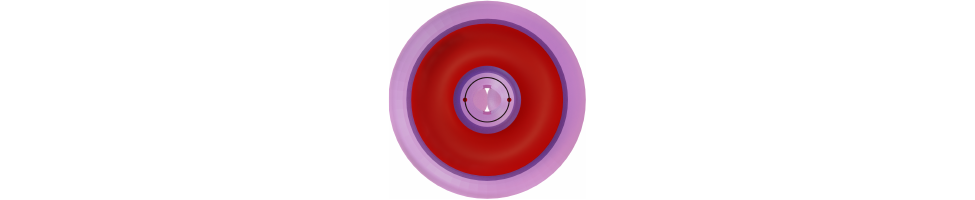}}%
    \put(0.480964,0.09426225){\makebox(0,0)[lt]{\lineheight{1.25}\smash{\begin{tabular}[t]{l}\pairf  k\end{tabular}}}}%
    \put(0.50655718,0.09426225){\makebox(0,0)[lt]{\lineheight{1.25}\smash{\begin{tabular}[t]{l}\pairs k\end{tabular}}}}%
    \put(0.37766471,0.16019473){\makebox(0,0)[lt]{\lineheight{1.25}\smash{\begin{tabular}[t]{l}$T_k$\end{tabular}}}}%
    \put(0.34143732,0.04855102){\makebox(0,0)[lt]{\lineheight{1.25}\smash{\begin{tabular}[t]{l}$C_k^{\boxplus r}$\end{tabular}}}}%
    \put(0.34143732,0.10614304){\makebox(0,0)[lt]{\lineheight{1.25}\smash{\begin{tabular}[t]{l}$C_k$\end{tabular}}}}%
    \put(0,0){\includegraphics[width=\unitlength,page=2]{Annuli_row_3_annotated.pdf}}%
  \end{picture}%
\endgroup%
 
  \subcaption{The points $p_k$ and $\tilde{p}_k$ lie in the symmetry plane for the final torus $T_k$ in the row of tori. }
\end{subfigure}
\caption{The construction for manifolds imitates the construction for general sets of positive reach as much as possible. The manifold $\M$ is depicted in blue, the sample $P$ in red, and the thickening in pink. We only display the part of objects below a horizontal clipping plane.
}
\label{Fig:CycleManifold2full}
\end{figure}  

\section{Subsets of Riemannian manifolds}
\label{sec:Riemannian_setting}

In this section we extend our analysis from the Euclidean setting to Riemannian manifolds with bounded curvature. We assume that the author is familiar with the basics of Riemannian geometry. We will be using results from comparison theory, which, for the convenience of the reader, we recall in Appendix~\ref{sec:RecapToponogov}.

\subsection{Definitions and settings}
\label{sec:Riemannian_definitions}

Before we can state and prove our homotopy reconstruction result, we need to generalize Federer's notions and results (Section~\ref{sec:Euclidean_definitions}) to Riemannian manifolds. This includes an appropriate definition of the reach, as well as a generalization of Federer's Theorem 4.8(12) (Theorem~\ref{Fed4.8.12}) and an extension of the normal cone. 

Throughout this section we will be working with the distance function, the closest point projection, and the medial axis. To this end, let $\Su$ be a closed non-empty subset of an ambient manifold $\N$. We denote the \emph{distance function} to $\Su$ by
\begin{equation}\label{eq:distance_function}
	\rho_{\Su} : \N  \to \mathbb{R}, \qquad \rho_{\Su} (q) := \inf \{ d_\N(q,p)  \mid p \in \Su \}.
\end{equation}
We write $\pi_\Su (q)$ for the set of points $p\in \Su$ such that $d_\N(q,p)= \rho_{\Su} (q)$, and call the set $\pi_\Su$ \emph{the closest point projection} of $q$ onto $\Su$. The \emph{medial axis} of $\Su$ is the set of those points in $\Su$ whose closest point projection consists of more than one point:
\begin{equation}\label{eq:medial_axis_Riemannian}
	\textrm{ax}_{\N}(\Su) : = \{ q \in \N \mid \operatorname{Card} (\pi_\Su(q)) > 1 \},
\end{equation}
where $\operatorname{Card}(A)$ denotes the cardinality of the set $A$.

\subparagraph{The cut locus reach.}
Generalizations of the reach have been studied before; however, none of the existing definitions fit our purpose. We thus introduce a new variant of the reach that is optimal in our setting. In addition, we discuss the various definitions of the reach in Appendix~\ref{sec:reach_history}.

Our variant of the reach is based on the cut locus.
The cut locus (see for example~\cite{berger2003panoramic}) is commonly defined for a single point --- say $p$ --- in a Riemannian manifold, and consists of those points in the manifold, for which there is no unique geodesic to $p$.

\SetCutLocus* 

Observe that {the cut locus contains the medial axis, that is, }$\textrm{ax}_{\N} (\Su) \subseteq \operatorname{cl}_{\N}(\Su)$.

\ReachCutLocus*

\subparagraph{The key tool: the flow}
In \cite{albano2013singular}, the authors extend the result of \cite{LIEUTIERhomotopytype}, 
namely that  any open bounded subset of Euclidean space has the same homotopy type as its medial axis, to the more general situation of an open bounded subset $\Omega$ of a Riemannian manifold.  
By using some  tools of non-smooth analysis, namely the properties of semi-concave functions,
 as well as some Riemannian geometry, they  generalize the result of \cite{LIEUTIERhomotopytype} while providing a shorter proof. However, the underlying idea in both \cite{albano2013singular} and \cite{LIEUTIERhomotopytype} is the same, which is to use the flow: given an open bounded subset $\Omega$ of a Riemannian manifold, the flow $\Phi :\Omega \times [0, \infty) \rightarrow \Omega$ is induced by a generalized gradient of the distance function on its boundary $\partial \Omega$. It is continuous, and realizes a homotopy equivalence --- more precisely, a weak deformation retraction --- between the set $\Omega$
and the cut locus of its boundary $\partial \Omega$. We refer the reader to Appendix \ref{sec:summary_albano} for more details.

We consider {the flow defined on the} complement of a closed subset $\Su\subseteq \N$:
	\begin{equation}\label{eq:the_flow}
		\Phi_{\Su}: \left( \N \setminus \Su \right) \times [0, \infty) \rightarrow \N \setminus \Su, \qquad (p,t) \mapsto  \Phi_{\Su}(p,t).
	\end{equation} 
Roughly speaking, the flow follows the steepest ascent of the distance function. The precise definition of the flow $\Phi_{\Su}$ is extensive and described in detail in Appendix~\ref{sec:singular_set_chracterization}. We refer the reader to Equation~\eqref{eq:RightDerivativeIsProjOnSuperGradient} for an explicit formula, and note that, thanks to Lemmas 3.4 and 3.5 as well as proof of Theorem 5.3 in \cite{albano2013singular}, $\Phi_{\Su}$ is locally Lipchitz in $p$
and $1$-Lipschitz in $t$.
 
In the next two lemmas we leverage the properties of the flow $\Phi_{\Su}$ to trace minimizing geodesics and define a deformation retract onto the set $\Su$. 
The next two lemmas show that near a set of positive cut locus reach the flow  $\Phi_{\Su}$ goes along geodesics, which in turn yields a generalization of part of Federer's Theorem 4.8(12) (Theorem~\ref{Fed4.8.12}), and that the flow induces a deformation retract on a set of positive reach. 
The proofs of the lemmas are given in Appendix~\ref{sec:singular_set_chracterization}.
\begin{restatable}{lemma}{TrajectoryAreMinimisingGeodesics}\label{lemma:TrajectoryAreMinimisingGeodesics}
Let $0<\rho <  \rchcl_{\N}(\Su) $. Then for any point $p\in\Su^{\boxplus \rho} \setminus \Su$, and any parameter $t\in\rho - \rho_{\Su}(p)$, there is a unique minimizing geodesic from the point $\Phi_{\Su}(p, t)$ to $\Su$. Moreover,
\begin{equation*}
	\pi_{\Su} (\Phi_{\Su}(p, t) ) =  \pi_{\Su} (p) ,
\end{equation*}
and the minimizing geodesic from $\Phi_{\Su}(p, t)$ to $\Su$ is the concatenation of the minimizing geodesic from $p$ to $\Su$
with the trajectory $\Phi_{\Su}(p, [0,t])$.	
\end{restatable}

\begin{corollary}\label{corollary:NestedBalls} 
Let $\Su\subseteq\N$ be a closed set. Pick a point $p \in \N $ satisfying $0< \rho_{\Su}(p) < \rchcl_{\N}(\Su)$, and define
 \[
 \zeta := \rchcl_{\N}(\Su) - \rho_{\Su}(p). 
 \]
The domain of the flow $\Phi_{\Su}$ can be extended to negative values of $t$, namely $ t \in  [-\rho_{\Su}, \,  \zeta]$. We denote this extension by $\overline{\Phi_{\Su}}$, and define it via the geodesic segment extending the geodesic
from $p$ to $\pi_{\Su}(p)$.
For every point $y$ on this segment, we have $\pi_\Su (y)=\pi_\Su(p)$. 
Because the balls centered at $y$ with radius $d(y , \pi_\Su(p) )$ are nested, we have in particular that 
\[
B\left(p, \, \rho_{\Su}(p)\right)^\circ \subseteq  B\left( \Phi_{\Su}(p, \zeta) , \, \rchcl_{\N}(\Su) \right)^\circ \subseteq \N \setminus \Su.
\]
\end{corollary}

The complement of the open offset of $\Su$ is defined as
\begin{equation}\label{eq:ComplementOffset}
\CplOffset^\rho(\Su) := \left\{ p \in \N  \mid \rho_\Su(p) \geq \rho \right\}.
\end{equation}

\begin{restatable}{lemma}{reversedFlow}\label{label:reversedFlow}
Let $0<\rho < \rho'< \rchcl_{\N}(\Su)$. Then
\begin{equation}\label{eq:ReachReversed}
\rchcl_{\N}(\CplOffset^{\rho'}(\Su) ) \geq \rho',
\end{equation}
and the homotopy $\mathcal{H} : \Su^{\boxplus \rho} \times [0, \rho] \rightarrow \Su^{\boxplus \rho}$, defined by 
\begin{equation}\label{eq:HomotopyReversedFlow}
\mathcal{H}(p, t) := 
\begin{cases}
p &  \text{if}\quad p \in \Su, \\
\Phi_{\CplOffset^{\rho'}(\Su) } \left(p, \max(t, \rho_{\Su}(p) ) \right)& \text{if}\quad p \notin \Su,
\end{cases}
\end{equation}
realizes a deformation retract from the thickening $\Su^{\boxplus \rho}$ to the set $\Su$ where the trajectory of each point is a minimizing geodesic to $\Su$.
\end{restatable}

\subparagraph{The normal cone.}
Finally, we extend the normal cone and the tangent cone to the
  Riemannian setting, keeping for both the same notation as in the
  Euclidean case, that is, omitting the reference to the Riemannian
  manifold $\N$.

\begin{definition}[Normal cone]\label{definition:NormalCone}
For a point $q\in \Su$, we define the \emph{normal cone} to $\Su$ at $q$ by
\begin{equation*}
\Nor(q, \Su) := \left\{ \lambda v \in \Tan(q,\N) \mid \lambda \geq 0 \text{ and } \exists\, p \in \N \setminus \Su 
\text{ with }  p =  \exp_q (v) \text{ and }  |v| = \rho_{\Su} (p)  \right\}.
\end{equation*}
\end{definition}
\begin{remark}\label{remark:NormalConeIncludedInNormalCone}
Note that in Definition~\ref{definition:NormalCone} one has $q = \pi_{\Su} (p)$. Moreover, we can easily check that $\Nor(q, \Su)$ is a subset of the dual cone of the tangent cone, which can be derived from the Definition \ref{def:4.3and4.4Fed}  for Euclidean space as
\[
\Tan(q, \Su) := \Tan (0 , \exp_q^{-1} ( \Su \cap B(q,\rho)) 
\]
where $\rho>0$ is smaller than the injectivity radius of $\N$.

The reverse inclusion holds as well but we do not need it there.
 \end{remark}

\begin{lemma}\label{lemma:SetCutLocus2}
Given a closed subset $\Su \subseteq \N$ such that $\rchcl_{\N}(\Su)>0 $, a point 
 $q \in \Su$ and a vector $v\in \Nor(q, \Su)$ with $ |v|=1$, then for any $\lambda < \rchcl_{\N}(\Su)$, the curve
 $\exp_q ([0, \lambda v])$ is the unique minimizing geodesic connecting the point $\exp_q (\lambda  v)$ with the set $\Su$. 
\end{lemma}
\begin{proof}
With Definition \ref{definition:NormalCone} of the normal cone $\Nor(q, \Su)$, 
the proof  follows from Lemma \ref{lemma:TrajectoryAreMinimisingGeodesics}.
\end{proof}

\subparagraph{Settings.}
In the Riemannian setting, we let $B(p,r) = \{ x \in \N \mid
  d_\N(p,x) \leq r \}$ designate a geodesic ball of $\N$ and
  $X^{\boxplus r} = \bigcup_{x \in X } B(x, r)$ designate a union of
  geodesic balls. With all the necessary notions in place, we recall
the setting we assume for the remainder of
Section~\ref{sec:Riemannian_setting}:

	\begin{tcolorbox} 
\assumptionRiemannianSetting*
\end{tcolorbox} 

\subsection{The geometric argument}\label{sec:Riemannian_geometric_argument}
In this section we show that if the union of {(geodesic)} balls $P^{\boxplus r}  = \bigcup_{p \in P } B(p, r)$ 
covers a sufficiently large neighbourhood of $\Su$ and the parameter $r$ is not too big, $P^{\boxplus r} $  
deformation-retracts to $\Su$. 

\begin{theorem}\label{theorem:geometric_argument_riemannian}
	Assume that a parameter $\alpha>0$ is small enough, so that the $\alpha$-neighbourhood $\Su^{\boxplus \alpha }$ 
of the set $\Su$ is contained in the union of balls $P ^{\boxplus r }$.
In other words,
	\begin{align}  
	\Su^{\boxplus \alpha } 
	\subseteq P^{\boxplus r} 
	.
	\label{eq:TubeRiem} 
	\end{align} 
Define 
\begin{equation}\label{equation:Def_f_k}
	f_\curvlowbnd(\reach, \delta, \alpha) \defunder{=} 
	\begin{cases}
				\frac{1}{\sqrt{\curvlowbnd}}  \arccos  \frac{ \cos \sqrt{\curvlowbnd} \: (\reach - \delta)}{ \cos \sqrt{\curvlowbnd} \: (\reach - \alpha) }     & \text{if} \quad \curvlowbnd> 0,
		 \\
		\sqrt{ (\reach - \delta)^2  - (\reach - \alpha)^2}     & \text{if} \quad \curvlowbnd = 0, \\
		\frac{1}{\sqrt{\curvlowbndn}}  \arccosh  \frac{ \cosh \sqrt{\curvlowbndn} \: (\reach - \delta)}{ \cosh \sqrt{\curvlowbndn} \: (\reach - \alpha) }     & \text{if} \quad \curvlowbnd< 0. 
	\end{cases}
\end{equation}
Moreover, for any point $q\in\Su$ and any vector $v \in \Nor(q, \Su)$, let $\gamma_{q,v} (t)$ be the (arc length parametrized) geodesic  emanating from $q$ in the direction $v$, and write 
\[ L \defunder{=}  \{ \gamma_{q,v} (t) \mid t \in [0, \reach)\} .\]  

If
\begin{equation}\label{equation:R2TooSmallToCNormalLineS_riemannian}
r <  f_\curvlowbnd(\reach, \delta, \alpha),
\end{equation}
then 
the intersection $L \cap  (P^{\boxplus r} )$ 
is a connected geodesic segment. 

Furthermore, $P^{\boxplus r} $ 
deformation-retracts onto $\Su$ along the closest point projection.
\end{theorem}

\begin{proof}[Proof of Theorem \ref{theorem:geometric_argument_riemannian}]

\begin{figure}[h!]
	\centering
	\begin{subfigure}[b]{0.45\textwidth}
		\centering
		\includegraphics[width=\textwidth]{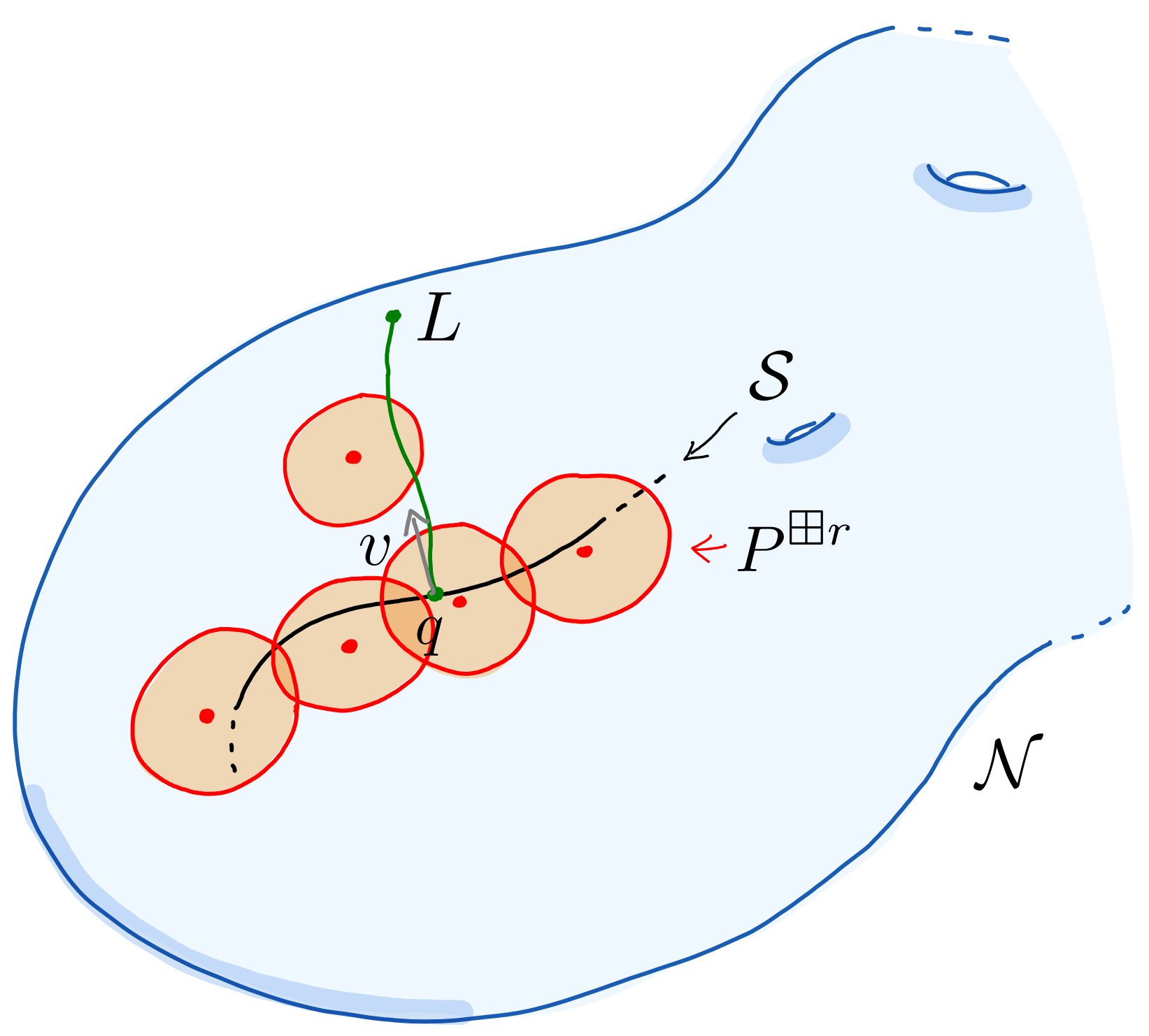}
		\subcaption{\small
			The intersection of the segment $L$ and the thickening $P^{\boxplus r}$ is not connected.
		}
		\label{fig:geometric_a}
	\end{subfigure}
	\hfill
	\begin{subfigure}[b]{0.45\textwidth}
		\centering 
		\includegraphics[width=\textwidth]{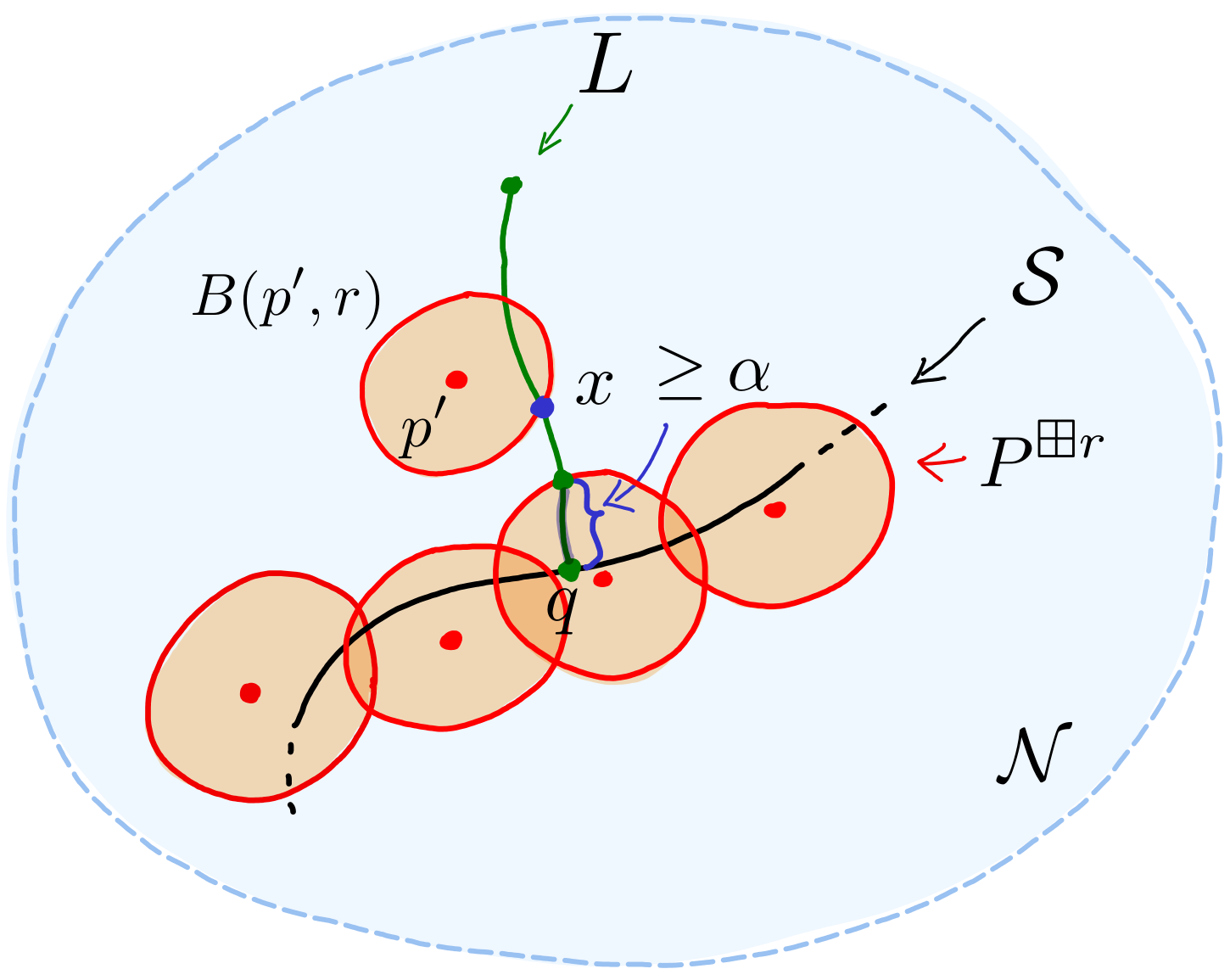}
		\subcaption{\small
			A close-up: the point $x$ lies on a different connected component of $L\cap P^{\boxplus r}$ than the point $q$, and thus the distance between $x$ and $q$ is at least $\alpha$.
		}
		\label{fig:geometric_b}
	\end{subfigure}
	\vskip\baselineskip
	\begin{subfigure}[b]{0.45\textwidth}    
		\centering 
		\includegraphics[width=\textwidth]{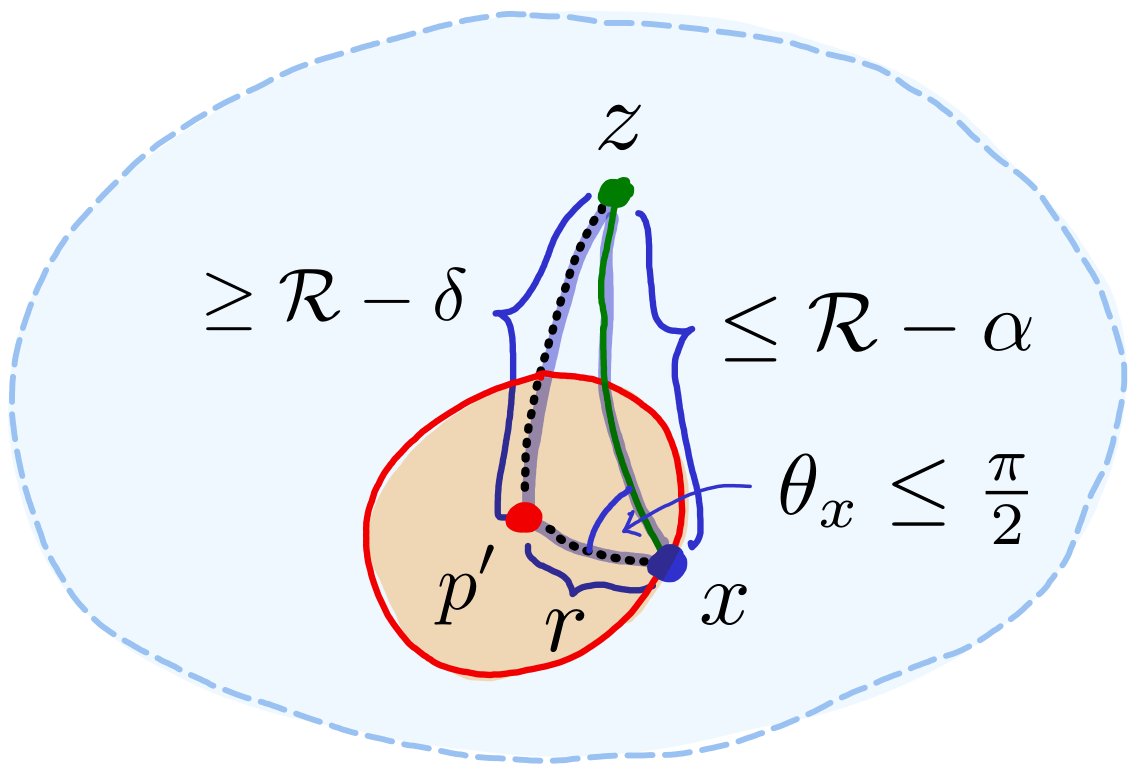}
		\subcaption{\small
			A close-up of the triangle $p'xz$.	
		}
		\label{fig:geometric_c}
	\end{subfigure}

	\caption{\small A pictorial overview of the proof. The blue shaded region represents a part of the manifold $\N$, the black line a part of the set $\Su$. The union of balls $P^{\boxplus r}$ 
		is coloured orange, and the segment $L$ green.
	} 
	\label{fig:OverviewProofRiemannian}
\end{figure}

We prove the claim by contradiction. 
	
Assume that there exists a point $q \in \Su$ and a vector $v \in \Nor(q, \Su)$, with $\|v\|=1$,
such that the intersection 
of $P^{\boxplus r}$ 
with the geodesic segment $L$ 
consists of several connected components (as illustrated in Figure \ref{fig:geometric_a}).
Thanks to Equation~\eqref{eq:TubeRiem}, the connected component that contains the point $q$ has length at least $\alpha$.
Let $x$ be first point along $L$, seen from $q$, lying inside a connected component of $ \left(P^{\boxplus r} 
\right) \cap L$ that does not contain $q$. 
Then $x$ lies at the intersection of the geodesic segment $L$ and the boundary of a ball $B(p', r)$, with $p'\in P$ (as illustrated in Figure~\ref{fig:geometric_b}).
	
Consider the `endpoint' $z=  \gamma_{q,v} (\reach)$ of the segment $L$. The distances between $z$ and the points $x$ and $p'$ satisfy (as in Figure \ref{fig:geometric_c}):
\begin{align}
d_\N(x, z) &\leq \reach - \alpha  \nonumber   
\\
d_\N(x, p') & =  r     \label{eq:BounfOnXPPrime},  \\
d_\N(z, p') &\geq \reach - \delta.  \nonumber 
\end{align}

Consider the geodesic triangle $p'xz$, and let $\theta_x$ denote its angle at $x$. 
By the Gauss Lemma~\cite[Lemma 3.5]{do1992riemannian}, the geodesic from $p'$ to $x$ is orthogonal, at the point $x$,  to the boundary of the geodesic ball $B(p',r)$. 
Due to the definition of $x$, 
the intersection of the ball $B(p', r)$ with  the segment of $L$ between $x$ and $q$
 is empty, and the angle $\theta_x$ satisfies
\begin{equation*}
\theta_x \leq \frac{\pi}{2}.
\end{equation*}

Thus, $\cos \theta_x \geq 0$. We now apply
Alexandrov-Toponogov distance comparison theorem (Theorem \ref{thm:toponogov_comparison_theorem_distance})
 and the law of cosines (Proposition \ref{proposition:CosineLaws}) to the triangle $p'xz$.

If $\curvlowbnd>0$, Theorem~\ref{theorem:BonnetSchoenbergMyers} bounds the diameter of the manifold $\N$ by $\reach \leq \operatorname{diam} ( \N ) \leq \frac{\pi}{\sqrt{\curvlowbnd} }$. We obtain:
\[
\cos \sqrt{\curvlowbnd} \: d_\N(z, p') \geq  \cos \sqrt{\curvlowbnd} \:d_\N(x, z) \:   \cos \sqrt{\curvlowbnd} \:d_\N(x, p').
\]

If $\curvlowbnd=0$,
\[
d_\N(z, p') ^2 \leq   d_\N(x, z)^2 + d_\N(x, p') ^2.
\]
If $\curvlowbnd<0$, we obtain
\[
\cosh \sqrt{\curvlowbndn} \: d_\N(z, p')  \leq  \cosh \sqrt{\curvlowbndn} \: d_\N(x, z) \:   \cosh \sqrt{\curvlowbndn} \: d_\N(x, p')  .
\]
Finally, inserting inequalities 
 \eqref{eq:BounfOnXPPrime} 
  yields:
\begin{equation}   \label{eq:ConstraintOnR}
\begin{rcases}
\cos \sqrt{\curvlowbnd} \: (\reach - \delta) \geq  \cos \sqrt{\curvlowbnd} \: (\reach - \alpha) \:   \cos \sqrt{\curvlowbnd} \:r,&\text{if}\quad \curvlowbnd>0, \quad 
\\
(\reach - \delta)^2 \leq   (\reach - \alpha)^2 + r^2,  &\text{if}\quad \curvlowbnd=0,   \quad   \\
\cosh \sqrt{\curvlowbndn} \: (\reach - \delta)  \leq  \cosh \sqrt{\curvlowbndn} \: (\reach - \alpha) \:   \cosh \sqrt{\curvlowbndn} \: r ,  &\text{if}\quad \curvlowbnd<0. \quad  
\end{rcases}
\end{equation}

Observe that the inequality \eqref{equation:R2TooSmallToCNormalLineS_riemannian} is precisely the negation of 
\eqref{eq:ConstraintOnR}, which gives the contradiction.

{We have proven the claim, namely that the intersection $L \cap  (P^{\boxplus r} )$
is a connected geodesic segment.

At last we turn our attention to the definition of $f_\curvlowbnd(\reach, \delta, \alpha)$ (Equation~\eqref{equation:Def_f_k}). Observe that, in each of the three cases, if $\reach- \alpha >0$, 
	then $ f_\curvlowbnd(\reach, \delta, \alpha)  < \reach - \delta$.
Equation~\eqref{equation:R2TooSmallToCNormalLineS_riemannian} {then}
implies that $\delta + r < \reach$, and thus
\begin{equation}\label{eq:POplusrInSOplusR}
P^{\boxplus r}  
\subseteq \left(\Su^{\boxplus \reach}\right) ^\circ . 
\end{equation}

Since $P^{\boxplus r}$ is a closed set, Equation~\eqref{eq:POplusrInSOplusR} implies that there exists a number $\rho <  \reach \leq \rchcl_{\N} (\Su)$ such that 
\begin{equation*}
P^{\boxplus r}  
\subseteq \Su^{\boxplus \rho}. 
\end{equation*}

We now apply Lemma~\ref{label:reversedFlow} with $\rho <  \rho' < \reach \leq \rchcl_{\N} (\Su)$. Consider the homotopy $\mathcal{H}$ from Equation~\eqref{eq:HomotopyReversedFlow}, and its restriction to the set $P^{\boxplus r} \times [0,  \rho]$.

Recall that $\rho_{\Su}$ denotes the distance function, defined by Equation~\eqref{eq:distance_function}. Given a point $p\in P^{\boxplus r}$, by Definition \ref{definition:NormalCone} of the normal cone, there is a point $q\in \Su$ and 
 a vector $v \in \Nor(q, \Su)$, with $\|v\|=1$, such that $p= \exp_q (\rho_{\Su} (p)\cdot v )$. 
 The image of $[0, \rho_{\Su} (p)  ]$
 under the exponential map $t \mapsto \exp_q t v$ is the unique minimizing geodesic from $p$ to $\Su$. Due to the claim, 
 this minimizing geodesic --- which is contained in the segment $L$ --- is also contained in the thickening $P^{\boxplus r}$.
%
 This implies that:
 \begin{equation*}
\forall p\in P^{\boxplus r} ,\, \forall t \in [0, \rho],\, \mathcal{H}(p,t) \in P^{\boxplus r}.
\end{equation*}
Thus, the restriction of the map $\mathcal{H}$ to the set $P^{\boxplus r} \times [0,  \rho]$ 
  is a deformation retract from the union of balls $P^{\boxplus r}$  
to the set~$\Su$.	
}

%
%
%

\end{proof}

\subsection{Bounds on the sampling parameters}
\label{sec:boundsRiemannian} 
In this section we extend Section \ref{sec:bounds} to subsets of Riemannian manifolds.

\subsubsection{Subsets of Riemannian manifolds with positive cut locus reach}
\label{sec:bounds_Riemannian_sets}
%

\HomotopyPositiveReachRiemannian*
%

\begin{proof}
{We begin by noting that the bound in the case where $\curvlowbnd = 0$ equals the bound in Proposition~\ref{theorem:HomotopyNoiselessPositiveReach}, and is deduced by the same analysis. In the following we thus only consider the cases $\curvlowbnd <0$ and $\curvlowbnd > 0$. }

We combine the bound from Lemma \ref{lem:Bounds_alpha_sets_pos_reach},
which, thanks to Remark \ref{remark:OffsetISAddiiveInGeodesicSpaces}, applies in the Riemannian context,
 with the conditions of Theorem~\ref{theorem:geometric_argument_riemannian}. 
More precisely,  
substituting $\alpha= r -\varepsilon $ in Equation \eqref{equation:R2TooSmallToCNormalLineS_riemannian} { yields 
\begin{align} 
\cos (\sqrt{\curvlowbnd} r) 
\geq
\frac{\cos (\sqrt{\curvlowbnd} (\reach- \delta) )}{\cos \sqrt{\curvlowbnd} (\reach+\varepsilon - r)} ,
& & \text{if } \curvlowbnd>0, \label{eq:IntervalForRSphericalCase}
\\
\cosh (\sqrt{\curvlowbndn} r) \leq \frac{\cosh (\sqrt{\curvlowbndn} (\reach- \delta) )}{\cosh \sqrt{\curvlowbndn} (\reach+\varepsilon - r)},
& & \text{if } \curvlowbnd<0. \label{eq:IntervalForRHyperbolicCase}
\end{align} 
These inequalities can be rearranged (using the product rule for the (hyperbolic) cosine) into
\begin{align} 
\frac{1}{2} \left( \cos (\sqrt{\curvlowbnd} (\reach +\varepsilon) ) +  \cos \sqrt{\curvlowbnd} (\reach+\varepsilon - 2r) \right ) & 
\geq 
\cos (\sqrt{\curvlowbnd} (\reach- \delta) ) ,
\tag{if $\curvlowbnd>0$}
\\
\frac{1}{2} \left( \cosh (\sqrt{\curvlowbndn}  (\reach +\varepsilon) ) + \cosh \sqrt{\curvlowbndn} (\reach+\varepsilon - 2r)  \right) &\leq  \cosh (\sqrt{\curvlowbndn} (\reach- \delta) ) ,
\tag{if $\curvlowbnd<0$}
\end{align} 
or 
\begin{align} 
\cos \left(\sqrt{\curvlowbnd} (\reach+\varepsilon - 2r)\right)  & 
\geq
2\cos (\sqrt{\curvlowbnd} (\reach- \delta) ) - \cos (\sqrt{\curvlowbnd} (\reach +\varepsilon) )  ,
\tag{if $\curvlowbnd>0$}
\\
 \cosh \left(\sqrt{\curvlowbndn} (\reach+\varepsilon - 2r)\right)  & \leq  2\cosh (\sqrt{\curvlowbndn} (\reach- \delta) )- \cosh (\sqrt{\curvlowbndn}  (\reach +\varepsilon) )  .
\tag{if $\curvlowbnd<0$}
\end{align} 
Because the left hand side of the inequality above is upper bounded by $1$ if $\curvlowbnd$ is positive, and is likewise lower bounded by $1$ if $\curvlowbnd$ is negative, we find \eqref{equation:BoundOnR0Riemannian}. The interval in which one may choose $r$ also follows immediately from the inequality above. It is clear that the interval in question is symmetric around $\frac{1}{2}(\reach+\varepsilon)$ (assuming \eqref{equation:BoundOnR0Riemannian} is satisfied).

} 
\end{proof} 

\subsubsection{Submanifolds of Riemannian manifolds with positive reach}
\label{sec:bounds_Riemannian_mflds}
In this section we extend Lemma \ref{lemma:BoundOnTubularCoverForManifolds} to the Riemannian setting. {In other words, 
we show that the bounds from Proposition \ref{theorem:HomotopyPositiveReachRiemannian} can be improved further if the set of positive reach {$\Su$ is a submanifold $\M$ of $\N$}.}

{Unlike the proof of Lemma \ref{lemma:BoundOnTubularCoverForManifolds}, which was rather algebraic,} the proof {of its extension }in the Riemannian setting --- Lemma~\ref{lemma:BoundOnTubularCoverForManifolds_Riemann} --- is purely geometrical and involves the Toponogov comparison theorem (see Appendix~\ref{sec:RecapToponogov} for an overview of results). 
In fact, a part of the proof of Lemma~\ref{lemma:BoundOnTubularCoverForManifolds_Riemann}
 can also be seen as an alternative proof of Lemma~\ref{lemma:BoundOnTubularCoverForManifolds}.  


\begin{lemma}\label{lemma:BoundOnTubularCoverForManifolds_Riemann}
Suppose that $\M \subseteq 
P^{\boxplus \varepsilon} 
\subseteq \N$ for some
$\varepsilon \geq 0$. Then, for any $r\geq\alpha\geq 0$ satisfying 
{
		\begin{equation}\label{equation:ConditionTubularCoverForManifold_Riemann}
		r \geq r_m,
		\end{equation}
} with $r_m$ defined via 
\begin{align} 
\cos \left (\sqrt{\curvlowbnd} r_m  \right) 
=& \frac{1}{\sin \left(  \sqrt{\curvlowbnd} \reach  \right)}  \left[ \sin \left( \sqrt{\curvlowbnd} \left(\reach+\alpha\right)   \right) \cos \left( \sqrt{\curvlowbnd} \varepsilon \right) - \right.\nonumber\\
& \qquad\qquad\qquad\qquad - \left.\sin \left( \sqrt{\curvlowbnd} \alpha   \right) \cos \left( \sqrt{\curvlowbnd} (\reach -\delta) \right) \right],
\tag{if $\curvlowbnd >0$}
\\
r_m^2 =& \;\alpha^2 + \varepsilon^2 + \frac{\alpha}{\reach} \left ( \reach^2+ \varepsilon^2- (\reach-\delta)^2 
\right),  
\tag{if $\curvlowbnd =0$}
\\
\cosh \left (\sqrt{\curvlowbndn} r_m  \right) 
=& \frac{1}{\sinh \left(  \sqrt{\curvlowbndn} \reach  \right)}  \left[ \sinh \left( \sqrt{\curvlowbndn} \left(\reach+\alpha\right)   \right) \cosh \left( \sqrt{\curvlowbndn} \varepsilon \right)\right.
\nonumber
\\
&\left.- \sinh \left( \sqrt{\curvlowbndn} \alpha   \right) \cosh \left( \sqrt{\curvlowbndn} (\reach -\delta) \right) \right]. 
\tag{if $\curvlowbnd <0$}
\\
\tag{\ref{eqdef:rm}}
\end{align} 
		the $\alpha$-neighbourhood $\bigcup_{q\in \M}  B(q,\alpha)={\M^{\boxplus \alpha}}$ of $\M$ is contained in the union of balls\\ $\bigcup_{p\in P} B(p,r)={P^{\boxplus r}}$. That is,
		\[
		\M^{\boxplus \alpha}\subseteq P^{\boxplus r}.
		\]
\end{lemma} 

\begin{proof} 
Given a point $q \in \M \subseteq \N$, the tangent {space} 
$T_q\M$ and the normal {space}
 $N_q\M$ are orthogonal vector spaces satisfying $T_q\M \times N_q\M = T_q \N$.
The normal cone $\Nor(q, \M)$, as defined in Definition~\ref{definition:NormalCone},
 is a subset of $N_q\M$ (see also Remark~\ref{remark:NormalConeIncludedInNormalCone}).

Since $\M \subseteq 
P^{\boxplus \varepsilon}$, the intersection $P\cap B(q,\varepsilon)$ is non-empty. Let $p \in P\cap B(q,\varepsilon)$.


\begin{figure}[h!]
	\centering
	\begin{subfigure}[b]{0.45\textwidth}
		\centering
		\includegraphics[width=\textwidth]{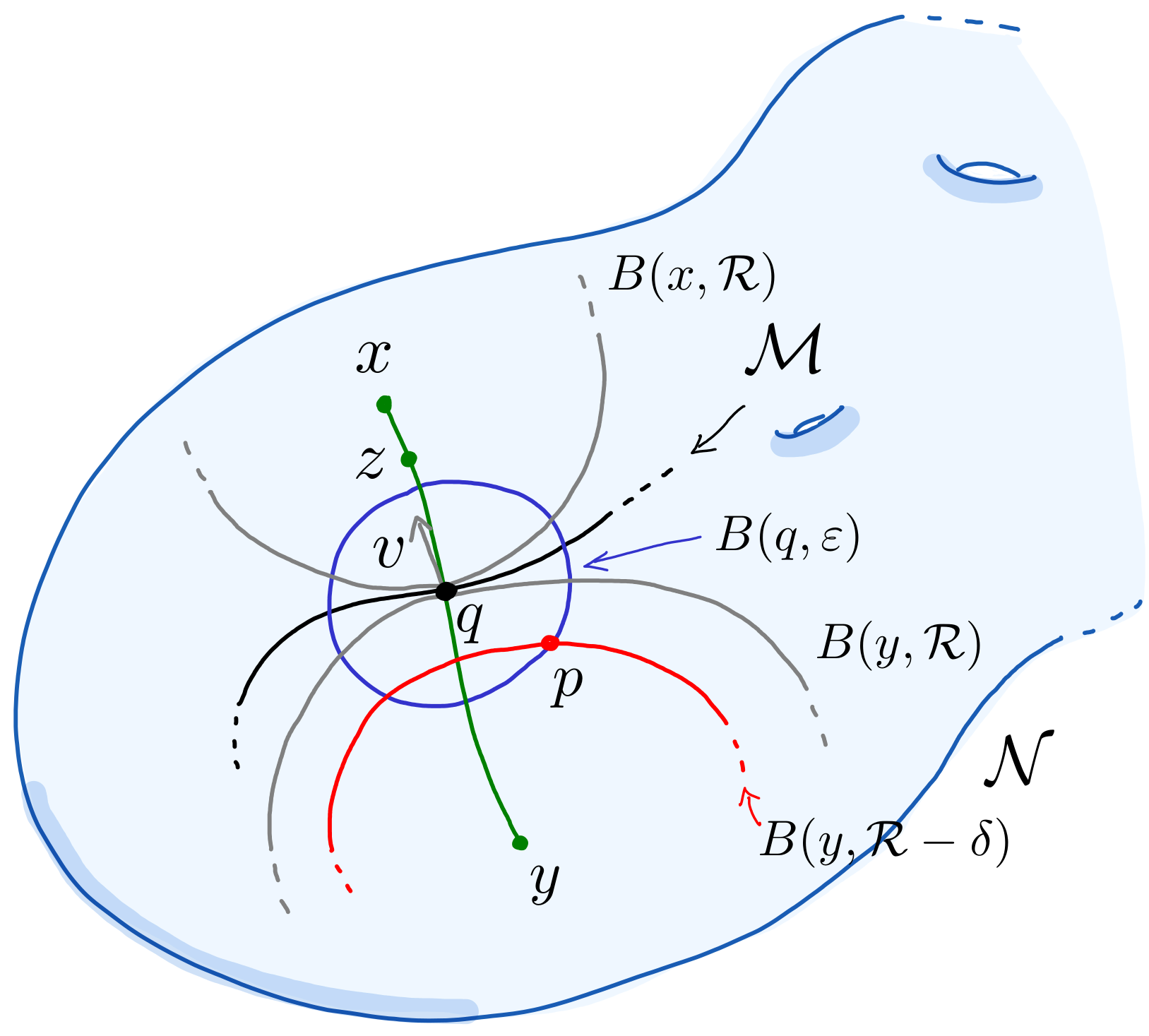}
		\subcaption{\small
		}
		\label{fig:two_comparison_triangles_1}
	\end{subfigure}
	\hfill
	\begin{subfigure}[b]{0.45\textwidth}
		\centering 
		\includegraphics[width=\textwidth]{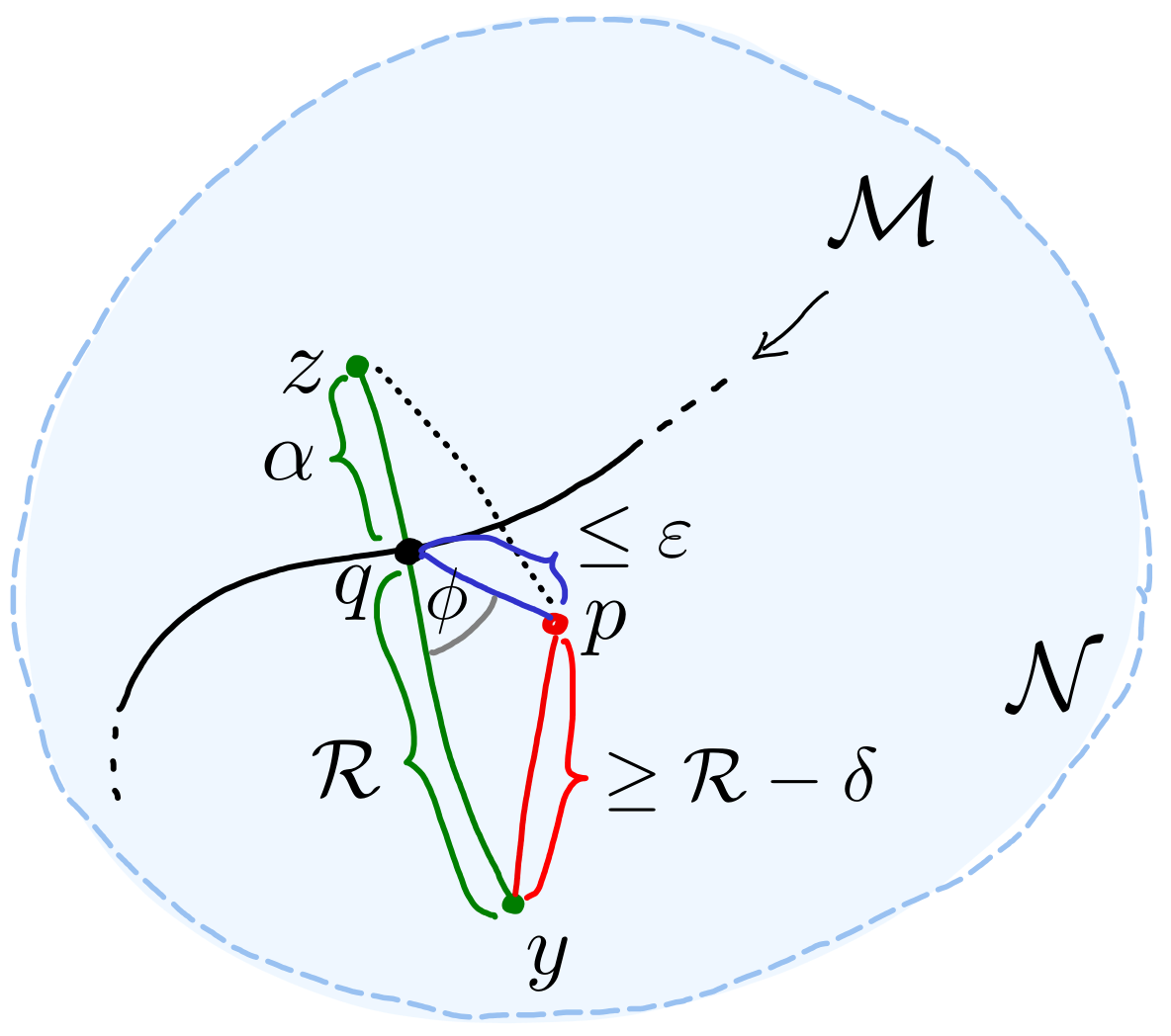}
		\subcaption{\small
		}
		\label{fig:two_comparison_triangles_2}
	\end{subfigure}
	
	\caption{Overview of the notation used in the proof of Lemma~\ref{lemma:BoundOnTubularCoverForManifolds_Riemann}. On the right we see the two comparison triangles $ypq$ and $zpq$.
	} 
	\label{fig:two_comparison_triangles}\label{fig:ProofManifoldCase1_Riemann}
\end{figure}

Further, let {
 $v \in \Nor(q, \M) \subseteq N_q\M\subseteq T_q \N$ }. {Denote the geodesic emanating from the point $q$ in direction $v$ by $\gamma_{q,v}$ and write $\gamma_{q,v} (\reach)=x$
  and $\gamma_{q,-v} (\reach)=y$ }. Thanks to Lemma~\ref{lemma:SetCutLocus2}, 
the manifold $\M$ and the ball $B(
x, \reach)^\circ$ {(resp. $B(
	y, \reach)^\circ$)} do not intersect, and thus
\begin{align}
p \notin B(x
, \reach- \delta)^\circ, \quad \text{as well as} \quad p \notin B(y
, \reach- \delta)^\circ.
\label{eq:p_not_in_Breach_Riemann}
\end{align}
{Finally, write $\gamma_{q,v} (\alpha)=z$. Our goal is to upper bound the distance between the point $p$ and the point $z$. To this end, we}  consider two geodesic triangles: 
$y
p q$ and 
$z
p q$. We sketch the situation in Figure~\ref{fig:ProofManifoldCase1_Riemann}.
The lengths of their edges satisfy: 
\[d_\N(y, q) =\reach, \qquad d_\N(y,p)\geq \reach -\delta, \qquad d_\N(p,q)\leq\varepsilon, \qquad d_\N(z, q) = \alpha.\]

{In the remainder of the proof, we use the terminology and results from comparison theory, which we summarize in Appendix~\ref{sec:RecapToponogov}.
We first determine a lower bound $\phi_\ell$ on the angle $\phi = \angle y
q p$, by applying 
Alexandrov-Toponogov angle comparison theorem (Theorem~\ref{thm:toponogov_comparison_theorem_angle})
to the triangle $y
p q $. 
}

Having established a bound on the angle $\phi$, 
{we 
apply Alexandrov-Toponogov distance comparison Theorem (Theorem~\ref{thm:toponogov_comparison_theorem_distance}) to the triangle $qpz$.
Since $\angle zqp =  \pi - \phi$,
Theorem~\ref{thm:toponogov_comparison_theorem_distance} gives us an upper bound on the length of the edge $pz$,
which we denote by $r_m$. }
We stress that $\pi- \phi_\ell$ is an upper bound on the angle $\angle z
q p$. 

  The length of the closing edge of a hinge
  in a space form is monotone in the lengths of the edges and the angle of the hinge,
  since, in the case when $\curvlowbnd>0$, 
 we can assume that $\reach \leq \frac{\pi}{2 \sqrt{\curvlowbnd}}$.

 Thus, $r_m$ is upper bounded by the closing edge of the hinge with edge lengths $\alpha$ and $\varepsilon$, and the angle $\pi- \phi_\ell$. 
Using the  law of cosines for 
{space forms} (Proposition \ref{proposition:CosineLaws}) we deduce that $r_m$ satisfies
\begin{align} 
\cos \left (\sqrt{\curvlowbnd} r_m  \right) 
=& \cos \left( \sqrt{\curvlowbnd} \alpha  \right) 
\cos \left( \sqrt{\curvlowbnd} \varepsilon \right)
\nonumber
\\
&- \frac{\sin \left( \sqrt{\curvlowbnd} \alpha   \right) }{ \sin \left(  \sqrt{\curvlowbnd} \reach  \right)} 
\left(
\cos \left( \sqrt{\curvlowbnd} (\reach -\delta) \right) - \cos \left(  \sqrt{\curvlowbnd} \reach  \right) \cos \left( \sqrt{\curvlowbnd} \varepsilon \right)
 \right), \nonumber \\
 =&\frac{1}{\sin \left(  \sqrt{\curvlowbnd} \reach  \right)}  \left[ \sin \left( \sqrt{\curvlowbnd} \left(\reach+\alpha\right)   \right) \cos \left( \sqrt{\curvlowbnd} \varepsilon \right) \right.\nonumber\\
 & \qquad\qquad\qquad\qquad-\left. \sin \left( \sqrt{\curvlowbnd} \alpha   \right) \cos \left( \sqrt{\curvlowbnd} (\reach -\delta) \right) \right],
\tag{if $\curvlowbnd >0$}
\\
r_m^2 =&\; \alpha^2 + \varepsilon^2 + \frac{\alpha}{\reach} \left ( \reach^2+ \varepsilon^2- (\reach-\delta)^2 
\right),
\tag{if $\curvlowbnd =0$}
\\
\cosh \left (\sqrt{\curvlowbndn} r_m  \right) 
=& \cosh \left( \sqrt{\curvlowbndn} \alpha   \right) 
\cosh \left( \sqrt{\curvlowbndn} \varepsilon \right)
\nonumber
\\
&- \frac{\sinh \left(\sqrt{\curvlowbndn} \alpha \right) }{ \sinh \left(\sqrt{\curvlowbndn} \reach  \right)} 
\left[
\cosh \left( \sqrt{\curvlowbndn} (\reach -\delta)  \right) \right.\nonumber\\
& \qquad\qquad\qquad - \left. \cosh \left( \sqrt{\curvlowbndn} \reach \right) \cosh \left( \sqrt{\curvlowbndn} \varepsilon \right)
 \right]
 \nonumber
 \\
 =&\frac{1}{\sinh \left(  \sqrt{\curvlowbndn} \reach  \right)}  \left[ \sinh \left( \sqrt{\curvlowbndn} \left(\reach+\alpha\right)   \right) \cosh \left( \sqrt{\curvlowbndn} \varepsilon \right)\right.
 \nonumber
 \\
  &\left.- \sinh \left( \sqrt{\curvlowbndn} \alpha   \right) \cosh \left( \sqrt{\curvlowbndn} (\reach -\delta) \right) \right] . 
\tag{if $\curvlowbnd <0$}
\\
\label{eqdef:rm} 
\end{align} 
\end{proof}
	
We are now ready to generalize Proposition \ref{theorem:DeformRetractsTheoremForManifolds}. The philosophy of the proof is the same as in Proposition~\ref{theorem:HomotopyPositiveReachRiemannian} --- we combine the conditions of Theorem \ref{theorem:geometric_argument_riemannian} and the bounds of Lemma \ref{lemma:BoundOnTubularCoverForManifolds_Riemann}. However, the involvement of trigonometric functions makes the analysis significantly more complicated.  	

\DeformRetractsTheoremForManifoldsRiemann*

\begin{proof}	
	%
	%
	We combine Theorem~\ref{theorem:geometric_argument_riemannian} and Lemma~\ref{lemma:BoundOnTubularCoverForManifolds_Riemann}. To improve the readability of the formulas, we write $\tilde x = \sqrt{\curvlowbndn} x$.
		
	At first, we assume that $\curvlowbnd>0$.
	Combining Equations \eqref{equation:R2TooSmallToCNormalLineS_riemannian} 
	and \eqref{equation:ConditionTubularCoverForManifold_Riemann} yields:
		\begin{align}
			& \frac{1}{\sin \tilde{\reach}}  \left[ \sin  \left(\tilde{\reach}+\tilde{\alpha}\right) \cos \tilde{\varepsilon} - \sin \tilde{\alpha} \cos  (\tilde{\reach}-\tilde{\delta}) \right]
			\geq \cos  \tilde{r}_m   
			\geq \frac{ \cos \left(\tilde{\reach}-\tilde{\delta}\right)}{ \cos \left(\tilde{\reach}-\tilde{\alpha}\right)}.
			\label{equation:BoundOnrForManifolds_RiemannN}
		\end{align}
		By multiplying both sides of the inequality by $\sin \tilde{\reach}\cos \left(\tilde{\reach}-\tilde{\alpha}\right) $ and subtracting $\sin \tilde{\reach}\cos \left(\tilde{\reach}-\tilde{\delta}\right)$, we obtain
		\begin{equation}
			0\leq \cos \tilde{\varepsilon} \left[\cos \left(\tilde{\reach}-\tilde{\alpha}\right) \sin \left(\tilde{\reach}+\tilde{\alpha}\right)\right] \label{eq:expression_at_varepsion}
			- \cos \left(\tilde{\reach}-\tilde{\delta}\right)\left[\cos \left(\tilde{\reach}-\tilde{\alpha}\right)\sin \tilde{\alpha} + \sin \tilde{\reach}\right].
		\end{equation}
		Observe that the terms are neatly divided: We have one term with $\tilde{\varepsilon}$, followed by an expression involving $\tilde{\reach}$ and $\tilde{\alpha}$ in square brackets, and one term with $\tilde{\reach} - \tilde{\delta}$, followed again by an expression involving $\tilde{\reach}$ and $\tilde{\alpha}$ in square brackets.
		
		Using the standard sum and double angle formulas for trigonometric functions, we transform the terms in the square brackets into
		
		\begin{equation*}
			\cos \left(\tilde{\reach}-\tilde{\alpha}\right) \sin \left(\tilde{\reach}+\tilde{\alpha}\right) = \frac{1}{2} \left[ \sin 2\tilde{\alpha} + \sin 2\tilde{\reach} \right]
		\end{equation*}
		and
		\begin{equation*}
			\cos \left(\tilde{\reach}-\tilde{\alpha}\right)\sin \tilde{\alpha} + \sin \tilde{\reach}=\frac{1}{2}\left[ \sin 2\tilde{\alpha}\cos \tilde{\reach} + \sin \tilde{\reach}\left(3-\cos 2\tilde{\alpha}\right)\right].
		\end{equation*}
		
		
		With this reformulation we extracted all expressions with $\tilde{\alpha}$.
		Denoting $x := 2\tilde{\alpha}$ we finally rearrange the inequality \eqref{eq:expression_at_varepsion} in terms of $\cos x$ and $\sin x$:
		
		\begin{equation*}
			0\leq \sin x \left[\cos \tilde{\varepsilon}  - \cos \left(\tilde{\reach}-\tilde{\delta}\right) \cos \tilde{\reach}\right] + \cos x  \left[ \sin \tilde{\reach}\cos \left(\tilde{\reach}-\tilde{\delta}\right) \right]+\cos \tilde{\varepsilon}\sin 2\tilde{\reach}-3\sin \tilde{\reach}\cos \left(\tilde{\reach}-\tilde{\delta} \right).
		\end{equation*}
		
		Recall that we want to determine conditions on~$\varepsilon$ and~$\delta$ (in terms of~$\reach$), under which there exists a value~$x\in [0,\pi]$ that satisfies the inequality above. To this end, let  
		\begin{align*}
			&A:=\cos \tilde{\varepsilon}  - \cos \left(\tilde{\reach}-\tilde{\delta}\right) \cos \tilde{\reach},\\
			&B:=\sin \tilde{\reach}\cos \left(\tilde{\reach}-\tilde{\delta}\right) , \\
			&C:=\cos \tilde{\varepsilon}\sin 2\tilde{\reach}-3\sin \tilde{\reach}\cos \left(\tilde{\reach}-\tilde{\delta} \right),
		\end{align*}
		denote the three terms of the right hand side of the inequality. Observe that
		\begin{align} 
		A&>0, & B&>0, &\text{and}&& C&<0.
		\label{eq:SignsOfABC}
		\end{align} 
		
		 Indeed, the inequalities 
		 \[0\leq \tilde{\reach}-\tilde{\delta} \leq\tilde{\reach}\leq \pi/2 \quad\text{and}\quad \varepsilon<\tilde{\reach}\]
		 imply
		 \[0\leq \cos \tilde{\reach} \leq \cos \left(\tilde{\reach}-\tilde{\delta} \right)\quad\text{and}\quad \cos \tilde{\reach} <\cos\varepsilon,\]
		 and thus
		 \begin{align*}
		 	A&=\cos \tilde{\varepsilon}  - \cos \left(\tilde{\reach}-\tilde{\delta}\right) \cos \tilde{\reach} > \cos \tilde{\reach} \left(1-\cos \left(\tilde{\reach}-\tilde{\delta}\right)\right)\geq 0,\\
		 	\tfrac{C}{\sin \tilde{\reach}} &= 2\cos \tilde{\varepsilon}\cos \tilde{\reach}-3\cos \left(\tilde{\reach}-\tilde{\delta} \right)<\cos \tilde{\reach}\left(2\cos \tilde{\varepsilon}-3\right)<0.
		 \end{align*}
		Define
		\[f:[0,\pi]\to \R, \qquad f(x) = A\sin x + B\cos x + C.\]


{Our goal is to determine if there exists a
  point~$x\in[0,\pi]$ with $f(x)\geq 0$. Consider $x_0 \in [0,\pi]$
  such that $\cos x_0 = \frac{B}{\sqrt{A^2 + B^2}}$ and $\sin x_0 = \frac{A}{\sqrt{A^2 + B^2}}$. With this definition we can rewrite $f(x) \geq 0$ as
  \begin{align}
  f(x) = \sqrt{A^2 + B^2} \left( \cos(x-x_0) + \frac{C}{\sqrt{A^2 + B^2}} \right) \geq 0 . 
	\label{eq:OtherFprmForF} 
  \end{align} 
We see that $f$ has only one global maximum value in the interval $[0,\pi]$, namely at $x_0$. Hence, there exists a point $x \in [0,\pi]$
  with $f(x) \geq 0$ if and only if the global maximum $x_0$ of $f$ satisfies $f(x_0) \geq 0$. This, in turns, translates into
  \[
  1 + \frac{C}{\sqrt{A^2 + B^2}} \geq 0,
  \]
  that is, $-C \leq \sqrt{A^2 + B^2}$. {Thanks to \eqref{eq:SignsOfABC}, this can be rewritten as $ C^2\leq A^2+B^2$. }
}

Plugging in the values of $A$, $B$, and $C$ yields	
\begin{align}		
\left(2\cos \tilde{\varepsilon}\cos \tilde{\reach}-3\cos \left(\tilde{\reach}-\tilde{\delta} \right)\right)^2\leq  \left( \frac{\cos \tilde{\varepsilon}  - \cos \left(\tilde{\reach}-\tilde{\delta}\right) \cos \tilde{\reach}}{\sin \tilde{\reach}} \right)^2+ \cos^2 \left(\tilde{\reach}-\tilde{\delta}\right).
\tag{\ref{equation:BoundOnEpsilon2_Riemann}}
	\end{align}
	
	When $\curvlowbnd=0$, the proof reduces to calculations identical to those in the proof of Proposition~\ref{theorem:DeformRetractsTheoremForManifolds}. We refrain from repeating them here, and refer the reader back to the proof.

	Finally, we assume that $\curvlowbnd<0$. Our procedure is identical to the treatment of the case where $\curvlowbnd>0$.
	Combining Equations \eqref{equation:R2TooSmallToCNormalLineS_riemannian} 
	and \eqref{equation:ConditionTubularCoverForManifold_Riemann} yields:
		\begin{align}
			\frac{1}{\sinh \tilde{\reach}}  \left[ \sinh  \left(\tilde{\reach}+\tilde{\alpha}\right) \cosh \tilde{\varepsilon} - \sinh \tilde{\alpha} \cosh  (\tilde{\reach}-\tilde{\delta}) \right]
			\leq \cosh  \tilde{r}_m   
			\leq \frac{ \cosh \left(\tilde{\reach}-\tilde{\delta}\right)}{ \cosh \left(\tilde{\reach}-\tilde{\alpha}\right)}.
			\label{equation:BoundOnrForManifolds_RiemannHyperbolicN}
		\end{align}
		By multiplying both sides of the inequality by $\sinh \tilde{\reach}\cosh \left(\tilde{\reach}-\tilde{\alpha}\right) $ and subtracting\\ $\sinh \tilde{\reach}\cosh \left(\tilde{\reach}-\tilde{\delta}\right)$, we obtain
		\begin{align}
			0\geq \cosh \tilde{\varepsilon} \left[\cosh \left(\tilde{\reach}-\tilde{\alpha}\right) \sinh \left(\tilde{\reach}+\tilde{\alpha}\right)\right] \label{eq:expression_at_varepsion2}
			- \cosh \left(\tilde{\reach}-\tilde{\delta}\right)\left[\cosh \left(\tilde{\reach}-\tilde{\alpha}\right)\sinh \tilde{\alpha} + \sinh \tilde{\reach}\right].
		\end{align}
		We transform the terms in the square brackets into
		
		\begin{align*}
			\cosh \left(\tilde{\reach}-\tilde{\alpha}\right) \sinh \left(\tilde{\reach}+\tilde{\alpha}\right) = \frac{1}{2} \left[ \sinh 2\tilde{\alpha} + \sinh 2\tilde{\reach} \right]
		\end{align*}
		and
		\begin{align*}
			\cosh \left(\tilde{\reach}-\tilde{\alpha}\right)\sinh \tilde{\alpha} + \sinh \tilde{\reach}=\frac{1}{2}\left[ \sinh 2\tilde{\alpha}\cosh \tilde{\reach} + \sinh \tilde{\reach}\left(3-\cosh 2\tilde{\alpha}\right)\right].
		\end{align*}
		
		Denoting $x := 2\tilde{\alpha}$ we finally rearrange the inequality \eqref{eq:expression_at_varepsion2} into a polynomial in $\cosh x$ and $\sinh x$:
		
		\begin{align*}
			0\geq& \sinh x \left[\cosh \tilde{\varepsilon}  - \cosh \left(\tilde{\reach}-\tilde{\delta}\right) \cosh \tilde{\reach}\right] + \cosh x  \left[ \cosh \left(\tilde{\reach}-\tilde{\delta}\right) \sinh \tilde{\reach}\right]\\
			&\quad +\cosh \tilde{\varepsilon}\sinh 2\tilde{\reach}-3\sinh \tilde{\reach}\cosh \left(\tilde{\reach}-\tilde{\delta} \right).
		\end{align*}
		
		Recall that we want to determine conditions on~$\varepsilon$ and~$\delta$ (in terms of~$\reach$), under which there exists a value~$x\geq 0$ that satisfies the inequality above. To this end, let  
		\begin{align*}
			&A_h:=\cosh \tilde{\varepsilon}  - \cosh \left(\tilde{\reach}-\tilde{\delta}\right) \cosh \tilde{\reach},\\
			&B_h:=\sinh \tilde{\reach}\cosh \left(\tilde{\reach}-\tilde{\delta}\right) , \\
			&C_h:=\cosh \tilde{\varepsilon}\sinh 2\tilde{\reach}-3\sinh \tilde{\reach}\cosh \left(\tilde{\reach}-\tilde{\delta} \right),
		\end{align*}
		denote the three terms of the right hand side of the inequality. Observe that
		\[A_h<0, \qquad B_h>0, \qquad\text{and}\qquad A_h+B_h>0.\]
		
		Indeed, the inequality $ 0\leq \tilde{\varepsilon}< \tilde{\reach}$
		implies $\cosh\tilde{\varepsilon}<\cosh \tilde{\reach}$,
		and thus
\begin{align*}
A_h&=\cosh \tilde{\varepsilon}  - \cosh \left(\tilde{\reach}-\tilde{\delta}\right) \cosh \tilde{\reach} < \cosh \tilde{\reach} \left(1-\cosh \left(\tilde{\reach}-\tilde{\delta}\right)\right)\leq 0,
\\
A_h+B_h&=\cosh\tilde{\varepsilon}  - e^{-\tilde{\reach}}\cosh \left(\tilde{\reach}-\tilde{\delta} \right)\geq \cosh\tilde{\varepsilon}  - e^{-\tilde{\reach}}\cosh\tilde{\reach} = \cosh\tilde{\varepsilon}  - \tfrac{1}{2}\left(1+e^{-2\tilde{\reach}}\right)>0. 
\end{align*}
		
		Define
		\[g:[0,\infty)\to \R, \qquad g(x) = A_h\sinh x + B_h\cosh x + C_h.\]
		
		%
		%
		%
		%
		  %
		

Because $A_h+B_h > 0$, we have that $B_h^2 - A_h^2 >0$. 
We now define $\rho$ to be the positive solution of the equations
\[ 
\rho^2 = B_h^2 - A_h^2.  
\] 
Because $\cosh^2(x) - \sinh^2 (x)=1$ and  $0\leq  -A_h< B_h$, there exists an $x_0$ such that $-A_h = \rho \sinh x_0$ and $B_h = \rho \cosh x_0$. Indeed, $x_0$ is given by $x_0 = \arctan (-A_h/B_h)$.    
Using the sum formula for $\cosh (a -b)$ we get the condition
\begin{equation} \label{eq:OtherFprmForG}
\frac{1}{\rho} g(x) = \cosh (x - x_0) + \frac{C_h}{\rho} \leq 0. 
\end{equation} 
Since the minimum of $t \mapsto \cosh t$ is $1$, this condition reduces to $ \frac{C_h}{\rho} \leq -1$, or equivalently 
\begin{align} 
C_h \leq  -\sqrt{B_h^2 - A_h^2 }. 
\label{eq:BoundCH}
\end{align} 
It is not difficult to recover the interval where $g(x)\leq 0$ from \eqref{eq:OtherFprmForG}.
It is convenient to reformulate \eqref{eq:BoundCH} as
	\begin{align*}		
	C_h\leq 0 \qquad\text{and}\qquad  B_h^2\leq A_h^2+C_h^2.
	\end{align*}
	In terms of $\tilde{\varepsilon}, \tilde{\delta}$, and $\tilde{\reach}$, these inequalities are equivalent to
	\begin{align}
		2 \cosh \tilde{\varepsilon}\cosh \tilde{\reach}&\leq 3\cosh \left(\tilde{\reach}-\tilde{\delta}\right) \qquad\text{and}\qquad
		\nonumber
		\\
		 \cosh^2 \left(\tilde{\reach}-\tilde{\delta}\right)&\leq\left(\frac{\cosh \tilde{\varepsilon}  - \cosh \left(\tilde{\reach}-\tilde{\delta}\right) \cosh \tilde{\reach}}{\sinh \tilde{\reach}}\right)^2+ \left(2\cosh \tilde{\varepsilon}\cosh \tilde{\reach}-3\cosh \left(\tilde{\reach}-\tilde{\delta} \right)\right)^2.
		\tag{\ref{equation:BoundOnEpsilon2_RiemannHyper}}
	\end{align}
\end{proof}

\subsection{Tightness of the bounds on the sampling parameters}
\label{sec:Riemannian_tightness}

We now prove that the bounds in the Riemannian setting are also tight in the following sense:

\begin{proposition}\label{prop:counterexample_set_Riemann}
Let $\curvlowbnd\in \R$. 
Assume that the one-sided Hausdorff distances $\varepsilon$ and $\delta$ fail to satisfy bound~\eqref{equation:BoundOnR0Riemannian}. Then there exists a manifold $\N$ (namely a space form) of dimension $d\geq 2$ whose sectional curvatures satisfy $K\geq \curvlowbnd$ (in fact $K = \curvlowbnd$), a subset $\Su\subseteq\N$ of positive (cut locus) reach $\reach$, and a sample $P$ that satisfy Universal Assumption~\ref{assumption:RiemannianSetting}, while the homology of the union of balls $P^{\boxplus r}$ 
does not equal the homology of $\Su$ for any $r$.
\end{proposition}

\begin{proposition}\label{prop:counterexample_mfld_Riemann}
Let $\curvlowbnd\in \R$.
Assume moreover that the one-sided Hausdorff distances $\varepsilon$ and $\delta$ fail to satisfy bound~\eqref{equation:BoundOnEpsilon2_Riemann}, and $\delta \leq \varepsilon$. Then there exists a manifold $\N$ (namely a space form\footnote{In the case of positive curvature we need a space with multiple connected components, that is, a number of spheres. See Remark \ref{Remark:MultipleConnectedComponents} for a more extensive discussion.}) of dimension $d\geq 3$ whose sectional curvatures satisfy $K\geq \curvlowbnd$ (in fact $K = \curvlowbnd$), a submanifold $\M \subseteq \N$ of positive (cut locus) reach, and a sample $P$ that satisfy Universal Assumption~\ref{assumption:RiemannianSetting}, while the homology of the union of balls $P^{\boxplus r}$
does not equal the homology of $\M$ for any~$r$.
\end{proposition} 

As in Section~\ref{sec:Euclidean_tightness}, we prove Propositions \ref{prop:counterexample_set_Riemann} and \ref{prop:counterexample_mfld_Riemann} by an explicit construction.
We construct the set $\Su$, the manifold $\M$, and the corresponding samples in Examples~\ref{example:set_Riemann} and~\ref{example:mfld_Riemann}, respectively. 

\begin{remark}\label{Remark:MultipleConnectedComponents}
As in Section~\ref{sec:Euclidean_tightness}, our construction involves a large (but finite) number of annuli or tori in a space form. If the curvature of the space form is positive then its volume is finite, such as in Figure \ref{fig:SampleWithLargeReachInSphere}. To be able to accommodate all the annuli, resp. tori, in our space form, we have to assume that it consists of multiple connected components. 
\end{remark} 

Instead of resorting to multiple connected components one could also weaken the statement as follows: 
\begin{proposition}
Let $\curvlowbnd\in \R$. 
Assume that the one-sided Hausdorff distances $\varepsilon$ and $\delta$ fail to satisfy bound~\eqref{equation:BoundOnR0Riemannian} (\eqref{equation:BoundOnEpsilon2_Riemann} respectively). Then there exists \emph{no} $r$ such that
\begin{itemize} 
\item for any manifold $\N$ of dimension $d\geq 2$ ($d \geq 3$ respectively), whose sectional curvatures satisfy $K\geq \curvlowbnd$, 
\item for any a subset $\Su\subseteq\N$ of positive cut locus reach $\reach$ (for every manifold $\M \subseteq \N$ of positive cut locus reach $\reach$, respectively), and 
\item for any sample $P$ that satisfies Universal Assumption~\ref{assumption:RiemannianSetting}
\end{itemize} 
the homology of the union of balls $P^{\boxplus r}$ equals the homology of $\Su$ ($\M$ respectively).
\end{proposition} 

\begin{figure}[!h]
		\begin{center}
		  \includegraphics[width=0.4\textwidth]{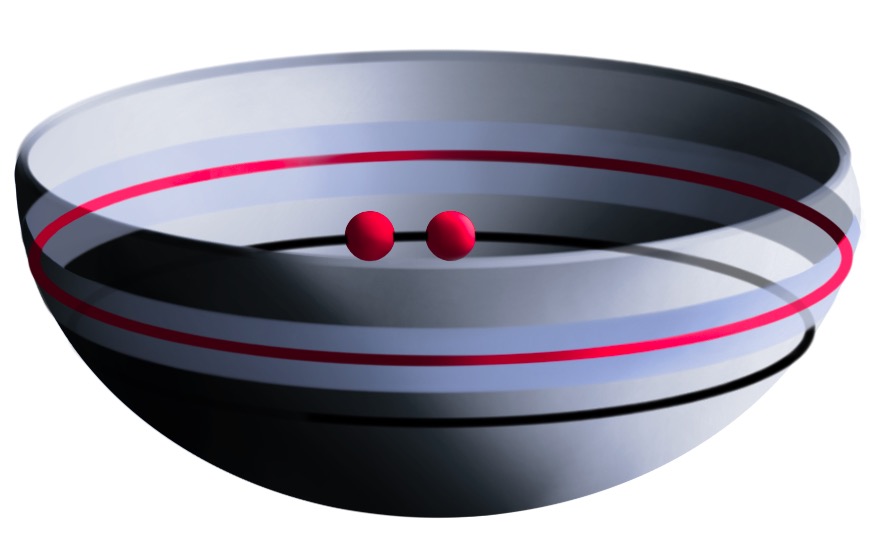}
		\end{center}
		\caption{
		An illustration of an annulus (in blue) on the sphere (in gray, we depict only half the sphere for the visualization) as well as the sample $P$ (in red). It is clear from the figure that the annulus takes up so much space that placing another one in the same sphere is impossible. 
		}
		\label{fig:SampleWithLargeReachInSphere}
	\end{figure}

\subsubsection{Subsets of Riemannian manifolds with positive reach} 
\label{sec:optimality_mflds_Riemann}
The construction of example in a space form 
of non-zero curvature, with which we
 prove Proposition \ref{prop:counterexample_set_Riemann}, generalizes the construction in Euclidean space (Example~\ref{example:set}) quite directly, see Figure~\ref{fig:sequenceOfAnnuliOnTheSphere}. 

\begin{example}\label{example:set_Riemann}
We choose $\N$ to be a two-dimensional space form of curvature $\curvlowbnd$, which we denote by $\mathbb{H}^2 (\curvlowbnd)$. The set $\Su$ is a union of annuli, where by an annulus $A_i$ we mean a set $A_i= B(z_i,\reach+{2}\varepsilon) \setminus B(z_i,\reach)^{\circ}  \subseteq \mathbb{H}^2 (\curvlowbnd)$. 
We call the point $z_i\in \mathbb{H}^2 (\curvlowbnd)$ the centre of the annulus. We assume the annuli lie at a distance at least $2\reach$ away from each other. 

The sample consists of a geodesic circle $C_i= \partial B(z_i,\reach+\varepsilon)$ and two points $\{p_i, \tilde{p}_i \} \subseteq \partial B(z_i,\reach-\delta)$ that are separated by a distance $2r_i$ for each annulus $A_i$. We provide an explicit definition for the parameter $r_i$ shortly.

We recall that bisectors in $\mathbb{H}^2 (\curvlowbnd)$ are geodesics. Thus, the bisector of the points $p_i$ and $\tilde{p}_i$ intersects the circle $C_i$ in two points. We let
$q_i$ be the intersection point that is the closest to $p_i$ (and thus $\tilde{p}_i$).

Consider the triangle $p_i \tilde{p}_i q_i$. We denote its circumradius by $R_i$ and note that $R_i \geq r_i$. 
Finally, we define the distance $2r_i$ between each pair of points $p_i$ and $\tilde{p}_i$: We set the distance $r_0$ to be
\[r_0 = \tfrac{1}{2}d\left(q_0, \tilde{q}_0\right), 
\]
and define
\[r_{i+1}=
\begin{cases}
R_i, &\text{if }R_i < 1-\delta,\\
	1-\delta, &\text{otherwise}.
	\end{cases}
\]
\end{example}

\begin{figure}[!h]
		\begin{center}
		  \includegraphics[width=0.4\textwidth]{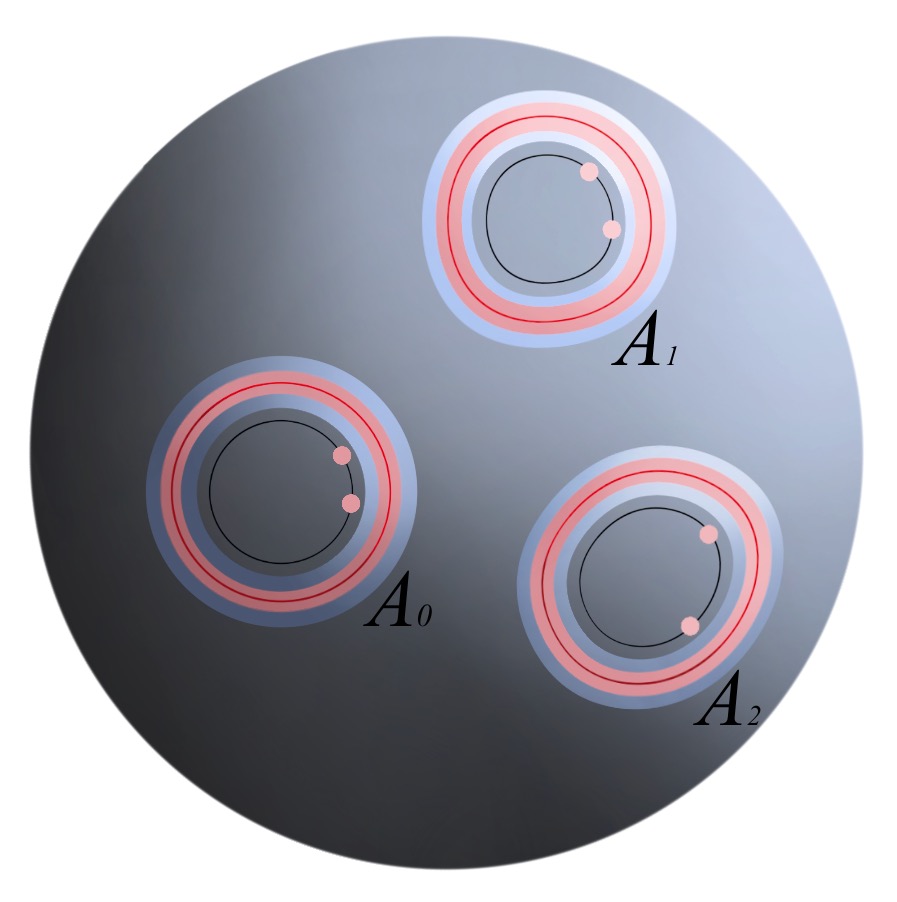}
		\end{center}
		\caption{
		A sequence of annuli on a space form. 
		}
		\label{fig:sequenceOfAnnuliOnTheSphere}
	\end{figure}

\begin{figure}[!h]
		\begin{center}
		  \includegraphics[width=0.5\textwidth]{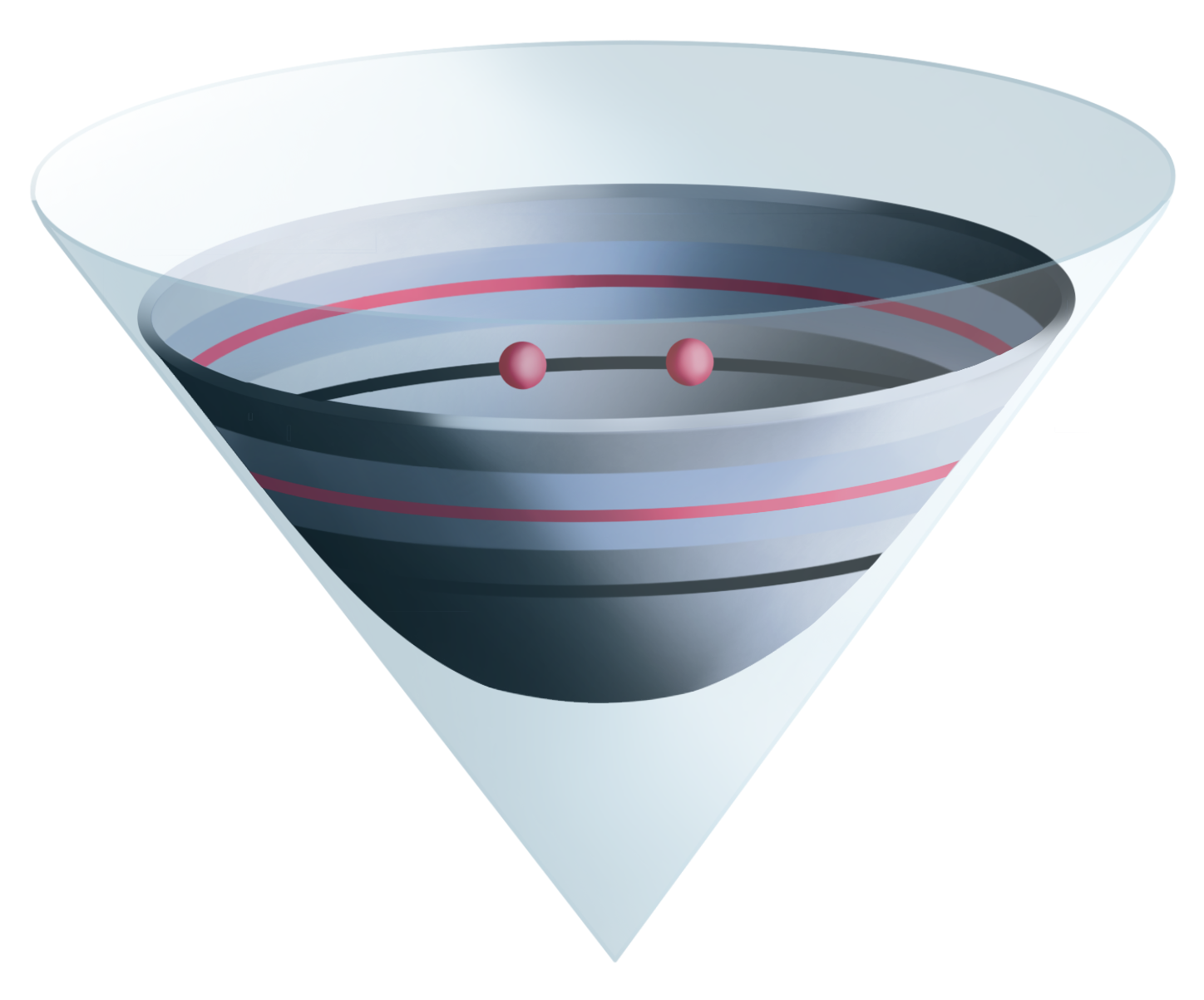} \\
		\end{center}
		\caption{
		A single annulus with sample in hyperbolic space, visualized using the Minkowski or hyperboloid model \cite{thurston2014three}. The limiting cone of the hyperboloid is included in transparent light blue. 
		}
		\label{fig:sequenceOfAnnuliOnTheHyperboloid}
	\end{figure}	

Next, we prove the generalization of Lemma~\ref{lemma:acute}.
\begin{lemma}
  \label{lemma:acute_Riemannian}
  If $\varepsilon$ and $\delta$ fail to satisfy the bound~\eqref{equation:BoundOnR0Riemannian}, then, for any $r_i \in [0,1-\delta]$,
  \begin{itemize}
  \item the triangle $p_i \tilde{p}_i q_i$ is strictly self-centred;
  \item there exists a constant $c > 0$, depending only on $\delta$, $\varepsilon$, and $\curvlowbnd$, such that $R_i - r_i \geq c \, r_i$.
  \end{itemize}
\end{lemma}
\begin{proof}
We observe that the triangle $p_i \tilde{p}_i q_i$ is self-centred for a sufficiently small value of $r_i$ (see Figure~\ref{fig:acute2}). 

Recall that the point $z_i$ is the centre of $C_i$, and let $C'_i$ be the geodesic circle centred at $z_i$ with radius $\reach-\delta$. That is, 
\[C'_i = \partial B(z_i,\reach- \delta).\] 
By construction, the circle $C_i'$ contains the points $p_i$ and $\tilde{p}_i$, while the circle $C_i$ contains the point $q_i$. Because the circumcentre of a triangle in a two-dimensional space lies on the bisector of any two of its vertices, the circumcentre of the triangle $p_i \tilde{p}_i q_i$ lies on the geodesic that contains the points $q_i$ and $z_i$ --- the bisector of $p_i$ and $\tilde{p}_i$. By definition, the midpoint $\mu_i$ of the segment connecting $p_i$ and $\tilde{p}_i$ lies on the bisector of $p_i$ and $\tilde{p}_i$.

We observe that the triangle $p_i \tilde{p}_i q_i$ is (strictly) self-centred if and only if the distance $d(q_i, \mu_i)$ is (strictly) longer than the circumradius of $p_i \tilde{p}_i q_i$. 
In other words, the transition between self-centredness and non-self-centredness happens when $d(q_i, \mu_i)= d(p_i, \mu_i)= d(\tilde{p}_i, \mu_i)= r_i$.

We now consider the triangle $z_i q_i p_i$, which is right-angled. At the moment when the triangle $p_i \tilde{p}_i q_i$ transitions from self-centred to non-self-centred, the edge lengths of the triangle $z_i q_i p_i$ are $\reach -\delta$, $\reach+\varepsilon -r_i$, and $r_i$ (see Figure~\ref{fig:acute2}). Applying the law of cosines yields
\begin{align}
\cos (\sqrt{\curvlowbnd} (\reach - \delta) ) &= \cos ( \sqrt{\curvlowbnd} (\reach +\varepsilon - r_i) ) \cos (\sqrt{\curvlowbnd} r_i) 
\nonumber
\\
&= \frac{1}{2} \left(  \cos ( \sqrt{\curvlowbnd} (\reach +\varepsilon ) ) + \cos ( \sqrt{\curvlowbnd} (\reach +\varepsilon - 2 r_i) ) \right) ,
\tag{if $\curvlowbnd>0$} 
\\
\cosh (\sqrt{\curvlowbndn} (\reach - \delta) ) &= \cosh ( \sqrt{\curvlowbndn} (\reach +\varepsilon - r_i) ) \cosh ( \sqrt{\curvlowbndn} r_i) 
\nonumber
\\
&= \frac{1}{2} \left(  \cosh ( \sqrt{\curvlowbndn} (\reach +\varepsilon ) ) + \cosh ( \sqrt{\curvlowbndn} (\reach +\varepsilon - 2 r_i) ) \right) ,
\tag{if $\curvlowbnd<0$} 
\\
\label{eq:conditionsTransitionPosReachRiemannMan} 
\end{align} 
Thus, the triangle $p_i \tilde{p}_i q_i$ transitions between self-centred and non-self-centred if the above equations have a real solution. Finally, since for a sufficiently small value of $r_i$ the triangle $p_i \tilde{p}_i q_i$ is self-centred if there is no solution to the equation \eqref{eq:conditionsTransitionPosReachRiemannMan}, the triangle $p_i \tilde{p}_i q_i$ is self-centred for all $r_i \in [ 0, 1-\delta]$.
 
The solution to \eqref{eq:conditionsTransitionPosReachRiemannMan} being vacuous is equivalent to $\varepsilon$ and $\delta$ failing to satisfy the bound~\eqref{equation:BoundOnR0Riemannian}. This completes the proof of the first statement. 
We note that because the triangle $p_i \tilde{p}_i q_i$ is strictly self-centred, $R_i > r_i$ and the second statement of the lemma also follows. 
\end{proof}

\begin{figure}[!h]
  \centering\includegraphics[width=0.65\textwidth]{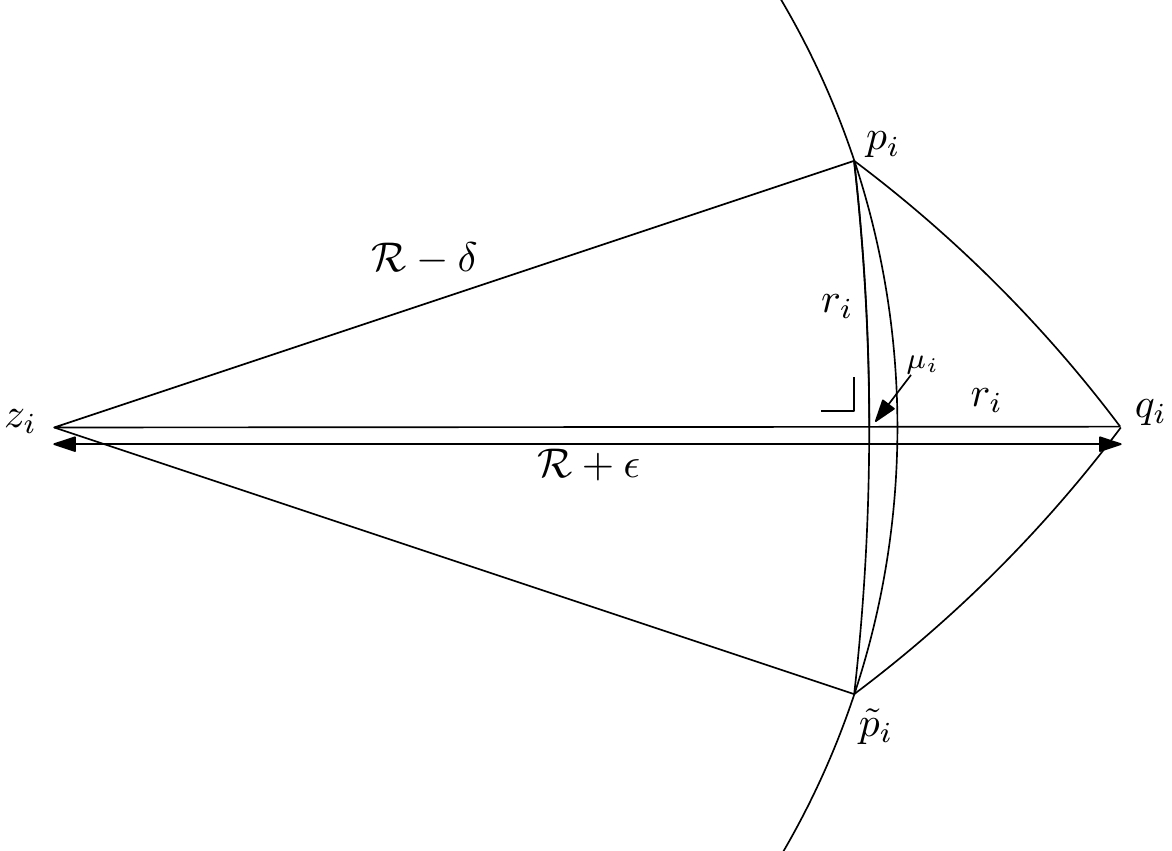}
  \caption{{ The figure illustrates the positive curvature case. }
    \label{fig:acute2}}
\end{figure}

\begin{proof}[Proof of Proposition \ref{prop:counterexample_set_Riemann}]  
The example has been set up in such a way that the transitions of the homology are precisely the same as in the Euclidean setting, and as we have described in the proof of Proposition \ref{prop:counterexample_set}. Hence, the set $\Su$ never has the same homology as the union of balls $\bigcup_{p \in P} B(p,r)$, and thus the two never have the same homotopy. 
\end{proof} 

\subsubsection{Submanifolds of Riemannian manifolds with positive reach}

The construction of the submanifold of $\mathbb{H}^d(\curvlowbnd)$ proving Proposition \ref{prop:counterexample_mfld_Riemann} generalizes the construction of Example~\ref{example:mfld}. 

\begin{example}\label{example:mfld_Riemann}
We choose $\N$ to be a three-dimensional space form of curvature $\curvlowbnd$, which we denote by $\mathbb{H}^3 (\curvlowbnd)$, and define $\M \subseteq \mathbb{H}^3 (\curvlowbnd)$ to be a union of tori $T_i$.
A geodesic circle $S^1(z,r)$ with centre $z$ and radius $r$ in a subspace $\mathbb{H}^2 (\curvlowbnd) \subseteq \mathbb{H}^3 (\curvlowbnd)$ is the boundary of a geodesic 2-ball (disk) $B(z,r) \subseteq \mathbb{H}^2 (\curvlowbnd)$. We write $H_S$ for the subspace $\mathbb{H}^2 (\curvlowbnd) \subseteq \mathbb{H}^3 (\curvlowbnd)$ that contains the circle $S^1(z,r)$. Whenever we want to express a circle in terms of the subspace $H_S$ it is lying in, we write $S^1(z,r, H_S)$. We refer to $H_S$ as the symmetry plane\footnote{We note that $H_S$ is indeed not a plane, but a totally geodesic subspace.}.   
Finally, each of the tori $T_i$ is a $\reach$-offset of the circle $S^1_i(z_i,2 \reach,H_S)$ --- a circle of radius $2 \reach$ in $\mathbb{H}^3 (\curvlowbnd)$. We refer to $z_i$ as the centre of the torus.

We number the tori from $i=0$, and 
we assume that their centres lie on a geodesic at a distance at least $2 \reach$ apart from one another, in such a way that they all share one symmetry plane $H_S$. Due to this assumption, the {cut locus} reach of $\M = \bigcup_i T_i$ equals $\reach$. 

The sample $P$ consists of sets $C_i$ which are tori with a part cut out, and pairs of points $\{p_i,\tilde{p}_i\}$ lying inside the hole of each torus $T_i$.
To construct each set $C_i$ we take the $\delta$-offset of the torus $T_i$, keep the part that lies inside the solid
  torus bounded by $T_i$, and remove an  $\varepsilon$-neighbourhood of the circle $S^1(z,\reach, H_S)$.

Each pair of points, $p_i$ and $\tilde{p}_i$, lies on the circle $S^1(z,\reach -\delta , H_S)$ at a distance $2 r_i$ from each other. 
Let {$q_i$ and $\tilde{q}_i$} be the two points in the intersection of the bisector of {$p_i$ and $\tilde{p}_i$} and the set $C_i$ that lie closest to  $p_i$ and $\tilde{p}_i$. Note that $q_i$ and $\tilde{q}_i$ lie on the boundary\footnote{Here we think of $C_i$ as a manifold with boundary. }  of $C_i$. We denote the circumradius of the simplex $p_i \tilde{p}_i q_i \tilde{q}_i$ by $R_i$. 

As in the Euclidean setting, we define the distance $2r_i$ between each pair of points $p_i$ and $\tilde{p}_i$ inductively. We set the distance $r_0$ to be:
\[r_0 = \tfrac{1}{2}d\left(q_0, \tilde{q}_0\right). 
\]
Moreover, we define
\[r_{i+1}=
\begin{cases}
R_i, &\text{if }R_i < 1-\delta,\\
	1-\delta, &\text{otherwise}.
	\end{cases}
\]
\end{example}

As before, we need a result on self-centredness of simplices. 

\begin{lemma} \label{lemma:selfcentredMan_Riemann} 
If $\varepsilon$ and $\delta$ fail to satisfy bound~\eqref{equation:BoundOnEpsilon2_Riemann}, and $r_i$ satisfies 
	\begin{align}\tag{\ref{eq:lower_bound_distance_p_ptilde}}
r_i &\leq 1-\delta,	&2r_i&\geq d(q_i,\tilde{q}_i),
\end{align}
then
\begin{itemize}
	\item the simplex $p_i \tilde{p}_i q_i \tilde{q}_i$ is strictly self-centred;
	\item there exists a constant $c > 0$, depending only on $\delta$, $\varepsilon$ and $\curvlowbnd$, such that $R_i  \geq r_i +c$.
\end{itemize}
\end{lemma}

\begin{figure}[!h]
  \centering\includegraphics[width=0.99\textwidth]{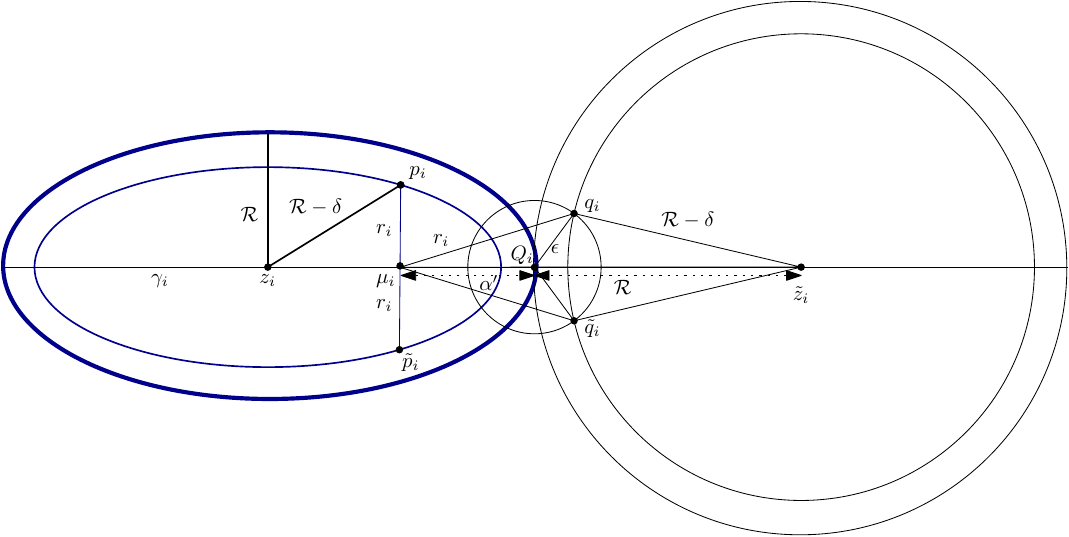}
  \caption{{The figure illustrates the transition between self-centred and non-self-centred simplices in the Euclidean case. The blue circles lie in $H_S$. }
    \label{fig:OptManRiem}}
\end{figure}

\begin{proof} 
By definition, the circumcentre of the simplex $p_i \tilde{p}_i q_i \tilde{q}_i$ lies on the bisector of $p_i$ and $\tilde{p}_i$ and the bisector of $q_i$ and $\tilde{q}_i$. Hence, the circumcentre of $p_i \tilde{p}_i q_i \tilde{q}_i$ lies on the geodesic $\gamma_i$ that contains the midpoint $\mu_i$ of $p_i, \tilde{p}_i$, and $z_i$.  For convenience we assume that $\gamma_i$ is arc length parametrized, that $\gamma_i(0) = z_i$, and that, for some parameter $t>0$, $\gamma_i ( [ 0 ,t])$ is the minimizing geodesic connecting the points $z_i$ and $\mu_i$. We write $Q= \gamma_i(\reach)$ and $\tilde{z}_i =\gamma_i( 2\reach) $. Finally, we denote the midpoint of $q_i$ and $\tilde{q}_i$ by $\tilde{\mu_i}$, and note that $\tilde{\mu}_i \in \gamma_i$. 

We start by noting that if $2r_i =d(q_i,\tilde{q}_i)$, the circumcentre of the simplex $p_i \tilde{p}_i q_i \tilde{q}_i$ is the midpoint of $\mu_i$ and $\tilde{\mu}_i$. Therefore, the simplex $p_i \tilde{p}_i q_i \tilde{q}_i$ is self-centred. We denote this midpoint by $\gamma_i(\tau)$. As we increase the value of the parameter $r_i$, the circumcentre of $p_i \tilde{p}_i q_i \tilde{q}_i$ moves along $\gamma_i$ in such a way that if we parametrize the movement by $\gamma_i(\tau')$, the parameter $\tau'$ decreases.

The transition between self-centredness and non-self-centredness of $p_i \tilde{p}_i q_i \tilde{q}_i$ takes place when the midpoint $\mu_i$ is the circumcentre of $p_i \tilde{p}_i q_i \tilde{q}_i$. This critical point is depicted in Figure~\ref{fig:OptManRiem}. 
In this case, the distance between the points $q_i$ and $\mu_i$ equals $d(q_i , \mu_i) = r_i$, and, by symmetry, $d(q_i , \mu_i) = \tilde{r}_i$).

Let $\alpha'= d(Q_i,\mu_i)$, and consider the two triangles $\mu_i Q_i q_i$ and $Q_iq_i \tilde{z}_i$ (or the symmetric triangles $\mu_i Q_i q_i$ and $Q_iq_i \tilde{z}_i$). Applying the law of cosines and using the fact that $\angle \mu_i Q_i q_i = \pi - \angle q_i Q_i \tilde{z}_i$ yields
\begin{align} 
\cos \left (\sqrt{\curvlowbnd} r_i  \right) 
=& \cos \left( \sqrt{\curvlowbnd} \alpha ' \right) 
\cos \left( \sqrt{\curvlowbnd} \varepsilon \right),
\nonumber
\\
&+ \frac{\sin \left( \sqrt{\curvlowbnd} \alpha '   \right) }{ \sin \left(  \sqrt{\curvlowbnd} \reach  \right)} 
\left(
\cos \left( \sqrt{\curvlowbnd} (\reach -\delta) \right) - \cos \left(  \sqrt{\curvlowbnd} \reach  \right) \cos \left( \sqrt{\curvlowbnd} \varepsilon \right)
 \right) 
\tag{if $\curvlowbnd >0$}
\\
\cosh \left (\sqrt{\curvlowbndn} r_i  \right) 
=& \cosh \left( \sqrt{\curvlowbndn} \alpha '  \right) 
\cosh \left( \sqrt{\curvlowbndn} \varepsilon \right)
\nonumber
\\
&+ \frac{\sinh \left(\sqrt{\curvlowbndn} \alpha '\right) }{ \sinh \left(\sqrt{\curvlowbndn} \reach  \right)} 
\left[
\cosh \left( \sqrt{\curvlowbndn} (\reach -\delta)  \right) \right.\nonumber\\
& \qquad\qquad\qquad\qquad - \left. \cosh \left( \sqrt{\curvlowbndn} \reach \right) \cosh \left( \sqrt{\curvlowbndn} \varepsilon \right)
 \right] ,
\tag{if $\curvlowbnd <0$}
\\
\label{eq:CritManRiem1} 
\end{align} 
Non-coincidentally, this expression is, up to relabeling of certain variables, the same as Equation~\eqref{eqdef:rm}.
 
On the other hand, applying the law of cosines to the triangle $p_i \mu_i z_i$ (or symmetrically to $\tilde{p}_i \mu_i z_i$) yields 
\begin{align} 
\cos \sqrt{\curvlowbnd} ( \reach - \delta) & = \cos \sqrt{\curvlowbnd} r_i \cos \sqrt{\curvlowbnd } (\reach - \alpha'), 
\tag{if $\curvlowbnd >0$} 
\\
\cosh \sqrt{\curvlowbndn} ( \reach - \delta) & = \cosh \sqrt{\curvlowbndn} r_i \cosh \sqrt{\curvlowbndn } (\reach - \alpha') .
\tag{if $\curvlowbnd <0$}
\\
\label{eq:ConditionThirdTriangle}
\end{align}

Combining Equations~\eqref{eq:CritManRiem1} and \eqref{eq:ConditionThirdTriangle} yields inequalities \eqref{equation:BoundOnrForManifolds_RiemannN} and \eqref{equation:BoundOnrForManifolds_RiemannHyperbolicN}, with the inequality replaced by an equality, and $\alpha$ replaced by $\alpha'$. A transition between self-centred and non-self-centred simplices can thus only take place if this equation has a real solution. The existence of this solution has been analyzed in the proof of Proposition~\ref{theorem:DeformRetractsTheoremForManifolds_Riemann}, leading to the inequalities \eqref{equation:BoundOnEpsilon2_Riemann} and \eqref{equation:BoundOnEpsilon2_RiemannHyper}.

In summary, the simplex $p_i \tilde{p}_i q_i \tilde{q}_i$ is self-centred if the conditions \eqref{eq:lower_bound_distance_p_ptilde} are satisfied. 


Since, for any $\varepsilon$ and $\delta$ failing inequalities \eqref{equation:BoundOnEpsilon2_Riemann} and \eqref{equation:BoundOnEpsilon2_RiemannHyper}, the simplex $p_i \tilde{p}_i q_i \tilde{q}_i$ is strictly self-centred for all $r \in \left[ \frac{d(q_i, \tilde{q}) }{2} ,1-\delta\right]$,  there is a lower bound on the difference between the length $2r_i$ of the edge $p_i \tilde{p}_i$ and the circumradius of the simplex $p_i \tilde{p}_i q_i \tilde{q}_i$. From this we deduce the second claim of the lemma.
\end{proof}

\begin{proof}[Proof of Proposition \ref{prop:counterexample_mfld_Riemann}]  
The example has been set up in such a way that the transitions of the homology are precisely the same as in the Euclidean setting, and as we have described in the proof of Proposition~\ref{prop:counterexample_mfld}. Once again, the manifold $\M$ never has the same homology as the union of balls $\bigcup_{p \in P} B(p,r)$, and thus the two never have the same homotopy. 
\end{proof} 

\clearpage

\section*{Appendix II: Additional material}
\section{The Toponogov comparison theorem and spaces of constant curvature}\label{sec:RecapToponogov}
We  rely on the Toponogov comparison theorems and the geometry of spaces of constant curvature. 
In this appendix we recall the results we use, however for the reader that is completely unfamiliar with the topic 
it may help to also take a look at the pedagogical overview in \cite{berger2003panoramic}. 

We use the notation $\mathbb{H}(\Lambda)$
for the complete, simply connected space of dimension $2$ with
constant sectional curvature $\Lambda$. A complete simply connected space with
constant sectional curvature is also called a \emph{space form}. Unless we state differently we assume a space of constant curvature to mean a space form.

The $2$-dimensional space of constant  curvature $\Lambda$  is, explicitly \cite[Theorem 39, pp. 228]{berger2003panoramic}:
\begin{equation}
\mathbb{H}(\Lambda) =
\begin{cases}
\frac{1}{\sqrt{-\Lambda}} \mathbb{HYP}^2 & \text{if} \quad \Lambda <  0 \\
\mathbb{E}^2 & \text{if}  \quad \Lambda = 0  \\
\frac{1}{\sqrt{\Lambda}} \mathbb{S}^2 & \text{if} \quad \Lambda >  0 .
\end{cases}
\end{equation}
where $ \mathbb{HYP}^2$, $\mathbb{E}^2$ and $ \mathbb{S}^2$ denote, respectively, the $2$-dimensional 
hyperbolic space, Euclidean space  and sphere.

We are now ready to make the following definitions.
\begin{definition}[Geodesic triangle] \label{def:GeodesicTriangle}
 A \emph{geodesic triangle} $ABC$ in a
  Riemannian manifold $\N$ consists of three minimizing geodesics
  connecting the three points $A,B,C$, sometimes also referred to as vertices. (We
  stress that a geodesic triangle does not include an interior.) 
\end{definition}
%

Complete Riemannian manifolds with positive lower bound on  sectional curvature have bounded diameter \cite[Theorem 62, pp. 266]{berger2003panoramic}:
\begin{theorem}[Bonnet-Schoenberg-Myers theorem]\label{theorem:BonnetSchoenbergMyers}
If a complete Riemannian manifold $\N$  has sectional curvature $K$ bounded below by a positive constant $\curvlowbnd$:
 \[
 0 < \curvlowbnd \leq K,
 \]
then it satisfies:
\begin{equation}\label{eq:CurvatureBoundsDiameter}
\operatorname{diam} ( \N ) \leq \frac{\pi}{\sqrt{\curvlowbnd} } .
\end{equation}
\end{theorem}
%

The next two theorems are adapted from \cite[Theorems IX.5.1 and  IX.5.2]{chavel2006riemannian}.
Since, unlike in  \cite{chavel2006riemannian}, our definition of geodesic triangles requires each edge to be a minimizing geodesic, 
and thanks to \eqref{eq:CurvatureBoundsDiameter},
 the statements in \cite{chavel2006riemannian} can be simplified.
\begin{theorem}[Alexandrov-Toponogov distance comparison theorem]\label{thm:toponogov_comparison_theorem_distance}
Let $\N$ be a  complete Riemannian manifold  with sectional curvatures  bounded below by  $ \curvlowbnd $.

 Let $ABC$ be a geodesic triangle in $\N$. Let us denote by $a$,$b$,and $c$ the respective lengths of sides $BC$,$CA$, and $AB$, and by $\alpha$ the angle at vertex $A$
  (see Figure \ref{fig:labelled.triangle}).
Then there exists a geodesic triangle  $A'B'C'$  in $\mathbb{H}(\curvlowbnd)$ such that sides $A'B'$ and $A'C'$ have respective lengths $c$ and $b$ and whose angle at $A'$
is $\alpha$. If $a'$ is the length of edge $B'C'$, then:
\[
    a \leq a'
\]
\end{theorem}

\begin{theorem}[Alexandrov-Toponogov angle comparison theorem]\label{thm:toponogov_comparison_theorem_angle}
Let $\N$ be a  complete Riemannian manifold  with sectional curvatures bounded below  by  $ \curvlowbnd $.

 Let $ABC$ be a geodesic triangle in $\N$. Let us denote by $a$,$b$,and $c$ the respective lengths of sides $BC$,$CA$, and $AB$, and by $\alpha$, $\beta$ and $\gamma$ 
 the respective angles  at vertex $A$, $B$, and $C$
  (see Figure \ref{fig:labelled.triangle}).

Then there exists a geodesic triangle  $A'B'C'$  in $\mathbb{H}(\curvlowbnd)$ such that sides $B'C'$,$C'A'$, and $A'B'$ have respective lengths $a$,$b$, and $c$,
and, if  $\alpha'$, $\beta'$ and $\gamma'$ are
 the respective angles  at vertex $A'$, $B'$, and $C'$, then:
\begin{eqnarray*}
    \alpha &\geq \alpha' \\
    \beta &\geq \beta' \\
    \gamma &\geq \gamma'.
\end{eqnarray*}
Unless $\curvlowbnd >0$,  and one of side lengths $a$,$b$, or $c$, is $\frac{\pi}{\sqrt{\curvlowbnd}}$, 
the triangle $A'B'C'$ is uniquely determined,
 up to isometries.
\end{theorem}

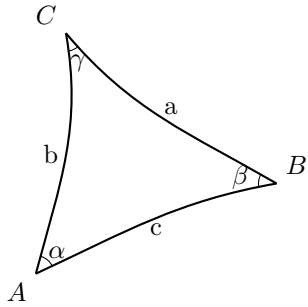
\begin{figure}[!htb]
\centerline{
\begin{tikzpicture}[scale=0.8]
  \coordinate [label={below left:$A$}] (A) at (0, 0);
  \coordinate [label={above left:$C$}] (C) at (0.5, 4);
  \coordinate [label={above right:$B$}] (B) at (4, 1.5);
	\draw[thick] (A) to [out=75,in=-80] (C);
 	\draw[thick] (A) to [out=25,in=190] (B);
	\draw[thick] (C) to [out=-50,in=150] (B);
	\draw[thin] (A) ++ (0.271892, 0.126785) arc (25:75:0.3);
	\draw[thin] (C) ++ (0.0520945,-0.295442) arc (-80:-50:0.3);
	\draw[thin] (B) ++ (-0.259808, 0.15) arc (150:190:0.3);
	\node at (0.25,2) {b};
	\node at (2,0.73) {c};
	\node at (2.25,2.75) {a};
  \node at (0.35, 0.35) {$\alpha$};
	\node at (3.4, 1.6) {$\beta$};
	\node at (0.7,3.5) {$\gamma$};
\end{tikzpicture}}
\caption{ Triangle with the standard symbols for angles and lengths.   }
\label{fig:labelled.triangle}
\end{figure}

\begin{remark}
The case, in Theorem \ref{thm:toponogov_comparison_theorem_angle}, 
where $\curvlowbnd >0$,  and one of side lengths $a$,$b$, or $c$, is $\frac{\pi}{\sqrt{\curvlowbnd}}$
can be ignored, in light  of \eqref{eq:CurvatureBoundsDiameter},
 if the sectional curvature is assumed  bounded below by some $\curvlowbnd' > \curvlowbnd$,
 where $\curvlowbnd'$ can be chosen arbitrarily close to $\curvlowbnd$.
\end{remark}

\begin{remark}
Propositions similar to Theorems  \ref{thm:toponogov_comparison_theorem_distance}
 and \ref{thm:toponogov_comparison_theorem_angle}
hold for manifolds with upper bounded sectional curvature, that imply reversed inequalities, but they require additional conditions,
in particular for the edge lengths to not exceed the injectivity radius.
\end{remark}

Theorem \ref{thm:toponogov_comparison_theorem_distance} 
will be combined with the law of cosines 
 for spaces of constant curvature.
\begin{proposition}[Law of cosines]\label{proposition:CosineLaws}
We consider a geodesic triangle  $ABC$  in $\mathbb{H}(\curvlowbnd)$.
We denote by $a =\operatorname{length} (BC)$, $b=\operatorname{length} (CA)$ 
and $c=\operatorname{length} (AB)$ the side lengths 
and by $\alpha$ the angle at vertex $A$,
as pictured on Figure \ref{fig:labelled.triangle}.

In the hyperbolic case, that is when $\Lambda<0$, then:
\[
\cosh \sqrt{| \Lambda |}  a = \cosh \sqrt{| \Lambda |} c\:   \cosh \sqrt{|\Lambda |} b
\,-\,  \sinh \sqrt{| \Lambda |}  c\:   \sinh \sqrt{| \Lambda |} b\: \cos \alpha
\]
In the Euclidean case,  that is when $\Lambda=0$, then:
\[
a^2 =  c^2 + b^2 - 2 \, c\, b\cos \alpha
\]

In the spherical case,  that is when $\Lambda>0$,  then:
\[
\cos \sqrt{\Lambda}  a = \cos \sqrt{\Lambda}  c\:   \cos \sqrt{\Lambda}  b
\,+\,  \sin \sqrt{\Lambda}  c\:   \sin \sqrt{\Lambda}  b\: \cos \alpha
\]

\end{proposition}

\section{Bounds on the reach of submanifolds of Riemannian manifolds with positive lower bound on the curvature}\label{Sec:BoundsReachManCurv}

In Theorem \ref{theorem:geometric_argument_riemannian} we have used the Bonnet-Schoenberg-Myers theorem (Theorem~\ref{theorem:BonnetSchoenbergMyers}) to prove that the reach of a set in a Riemannian manifold with positive curvature $\curvlowbnd$ is upper bounded by $\frac{\pi}{\sqrt{\curvlowbnd}}$. In this section we improve this bound by a factor of two in the case where the set in question is a manifold. The proof adjusts the argument for the  Bonnet-Schoenberg-Myers theorem as given in \cite[Theorem 62]{berger2003panoramic}. For this, we need to recall notation and a result from~\cite[Section 6.2]{berger2003panoramic}. 


Let $[a,b]\subseteq\R$ be an interval in $\mathbb{R}$, and $t\in [a,b]$. Following Berger~\cite{berger2003panoramic}, we write $c_\alpha (t)= c(\alpha,t)$ for a family of curves neighbouring a geodesic $\gamma(t) = c_0(t)$. The infinitesimal displacement is denoted by 
\[ 
Y(t) = \frac{\partial c}{\partial \alpha} \bigg|_{\alpha=0}.  
\] 
We assume that the displacement is orthogonal to the geodesic. We denote the sectional curvature for the directions $v,w$ by $K (v,w)$, and write $\nabla_t$ for the covariant derivative. 
If the endpoints of $c_\alpha$ are fixed then~\cite[Equation (6.7)]{berger2003panoramic}
\begin{align}
\frac{\partial^2 \text{length} \, c_\alpha}{\partial \alpha^2}\bigg|_{\alpha=0}=  \int_a^b \left( \|\nabla_t Y(t) \| ^2 - K( \gamma'(t) , Y(t) ) \| Y(t)\|^2 \right)  \ud t .
\label{eq:LengthVariation}
\end{align}
If the endpoints are not fixed, Equation~\eqref{eq:LengthVariation} gains an additional term~\cite[Theorem II.4.3]{chavel2006riemannian}:
\begin{align}
\frac{\partial^2 \text{length} \, c_\alpha}{\partial \alpha^2}\bigg|_{\alpha=0}= \langle \nabla_\alpha \partial_\alpha c(\alpha , t)  \mid _{\alpha =0},  \gamma '(t) \rangle \mid^{t=b}_{t=a} +  \int_a^b \left( \|\nabla_t Y(t) \| ^2 - K( \gamma'(t) , Y(t) ) \| Y(t)\|^2 \right)  \ud t .
\label{eq:LengthVariationWithVarEndPoints}
\end{align}
Naturally, if the endpoints of $c_\alpha$ are not fixed, the first variation is 
 non-zero (see \cite[Theorem II.4.1]{chavel2006riemannian}) and, using that $\gamma$ is a geodesic, we have
\begin{align}
\frac{\partial \text{length} \, c_\alpha}{\partial \alpha}\bigg|_{\alpha=0}= \langle Y (t),  \gamma '(t) \rangle \mid^{t=b}_{t=a} .
\label{eq:FirstOrderLengthVariationWithVarEndPoints}
\end{align}

\begin{lemma} \label{lemma:ManifoldHaveReachLessThanPiOverTwo}
Suppose that the sectional curvatures of a Riemannian manifold $\N$ are lower bounded by $\curvlowbnd>0$. 
Let $\M \subseteq \N$ be a { $C^2$} submanifold of $\N$ of dimension and codimension at least one, with cut locus reach $\operatorname{cl}_{\N} (\M)>0$.  
Then, \[\rchcl_\N (\M) \leq \frac{\pi}{2 \sqrt{\curvlowbnd} }.\]
\end{lemma} 

\begin{figure}[!h]
  \centering\includegraphics[width=0.6\textwidth]{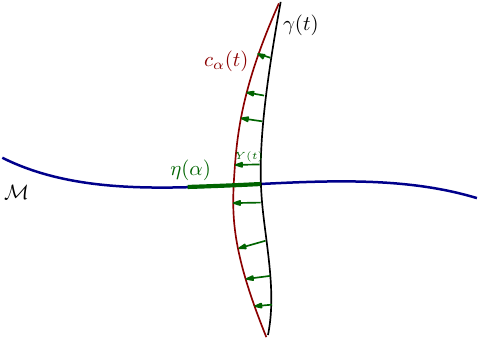}
  \caption{{The figure illustrates the notation used in Appendix \ref{Sec:BoundsReachManCurv}. The second order behaviour of  $c_\alpha(t)$ is determined by the vector $W$. However, $W$ is not indicated in the figure.}
    \label{fig:NotAppB}}
\end{figure}

\begin{proof} To derive a contradiction assume that $\rchcl_\N(\M) > \frac{\pi}{2 \sqrt{\curvlowbnd} }$.  For any point $p \in \N \setminus \M$ sufficiently close to $\M$, the minimizing geodesic from $p$ to $\M$ has a tangent vector that is normal to $\M$ at the endpoint of the geodesic. Let us call this endpoint $q\in  \M$, set $L = 2\rchcl_\N(\M)$, and parametrize the geodesic by a map
\[\gamma: \left[-\tfrac{L}{2},\tfrac{L}{2}\right]\to \N\]	
in such a way that $\gamma$ is arc length parametrized and $\gamma(0)=q$. We refer to Figure \ref{fig:NotAppB} for an overview of the notation used.  Furthermore, pick a tangent vector $Z \in T_q\M \subseteq T_q \N$. Due to the definition of $\gamma$, the vectors $Z$ and $\gamma'(0)$ are perpendicular.
 
%
%
%
As in the proof of \cite[Theorem 62]{berger2003panoramic}, we then consider the parallel transport of $Z$ along $\gamma$, which we denote by $Z(t)$. With $L = 2\rchcl_\N(\M)$, we define 
\[Y: \left[-\tfrac{L}{2},\tfrac{L}{2}\right]\to T\N, \qquad Y(t) = \cos  \left(\tfrac{\pi t}{L}\right) Z(t).\]

We choose the second order derivative of $c (\alpha,t)$ with respect to $\alpha$ as follows.
We write $\eta(\alpha)$ for the geodesic in $\M$ emanating from $q$ in the direction $Z$, i.e., $\eta(0)=q$ and $\eta'(0)=Z$. Next, we set 
\[W:= \nabla_{\alpha} \eta'(\alpha) \mid_{\alpha= 0}\]  
and, as with the vector $Z$, use parallel transport along $\gamma$ to extend the vector $W$ to a vector field $W(t)$ along the entire length of $\gamma$. Finally, we impose that 
\[\nabla_\alpha \partial_\alpha c(\alpha , t)  \mid _{\alpha =0} =  \cos   \left(\tfrac{\pi t}{L}\right) W(t).   \]

We stress that, since the vector $Z$ lies in the tangent space $T_q\M$ and due to the way the vector field $W(t)$ is defined, the members of the family of curves $c_\alpha$ arising from $Y(t)$ pass (up to second order) through $\M$. 

Finally, write $\psi_{\alpha} (t)$ for the restriction of  $c_{\alpha} (t)$ to the interval $\left[-\tfrac{L}{2},0\right]$, and $\tilde{\psi}_{\alpha} (t)$ for the restriction of  $c_{\alpha} (t)$ to the interval $\left[0,\tfrac{L}{2}\right]$.  
With this notation, and applying $L = 2\rchcl_\N(\M) > \frac{\pi}{\sqrt{\curvlowbnd} }$,  Equation~\eqref{eq:LengthVariationWithVarEndPoints} yields 
\begin{align}
\frac{\partial^2 \text{length} \, \psi_{\alpha} (t)}{\partial \alpha^2}\bigg|_{\alpha=0}&=   \langle \nabla_\alpha \partial_\alpha c (\alpha , t)  \mid _{\alpha =0},  \gamma '(t) \rangle \mid^{t=0}_{t=-L/2}
\nonumber
\\
& \phantom{=} + \int_{-L/2}^{0} \left( \frac{\pi^2 }{L^2} \sin^2 \left (\frac{\pi t}{L} \right) \| Z\|{^2}  - K( \gamma'(t) , Y(t) )\cos^2 \left (\frac{\pi t}{L} \right) \| Z\|{^2}  \right)  \ud t 
\nonumber
\\
&\leq \langle W,  \gamma '(0) \rangle
\nonumber
\\
& \phantom{=} +
\int_{-L/2}^{0} \left( \frac{\pi^2 }{L^2} \sin^2 \left (\frac{\pi t}{L} \right) \| Z\|{^2}  - \curvlowbnd \cos^2 \left (\frac{\pi t}{L} \right) \| Z\|{^2}  \right)  \ud t 
\nonumber 
\\ 
&=\langle W,  \gamma '(0) \rangle + \frac{L \| Z\|{^2} }{4}  \left( \frac{\pi^2 }{L^2}  - \curvlowbnd  \right ) 
\nonumber
\\
&<\langle W,  \gamma '(0) \rangle,
\nonumber
\end{align}
and 
\begin{align}
\frac{\partial^2 \text{length} \, \tilde{\psi}_{\alpha} (t)}{\partial \alpha^2}\bigg|_{\alpha=0}&=   \langle \nabla_\alpha \partial_\alpha c(\alpha , t)  \mid _{\alpha =0},  \gamma '(t) \rangle \mid^{t=L/2}_{t=0}
\nonumber
\\
& \phantom{=} + \int_{0}^{L/2} \left( \frac{\pi^2 }{L^2} \sin^2 \left (\frac{\pi t}{L} \right) \| Z\|{^2}  - K( \gamma'(t) , Y(t) )\cos^2 \left (\frac{\pi t}{L} \right) \| Z\|{^2}  \right)  \ud t 
\nonumber
\\
&\leq - \langle W,  \gamma '(0) \rangle
\nonumber
\\
& \phantom{=} +
\int_{0}^{L/2} \left( \frac{\pi^2 }{L^2} \sin^2 \left (\frac{\pi t}{L} \right) \| Z\|{^2}  - \curvlowbnd \cos^2 \left (\frac{\pi t}{L} \right) \| Z\|{^2}  \right)  \ud t 
\nonumber 
\\ 
&=- \langle W,  \gamma '(0) \rangle + \frac{L \| Z\|{^2} }{4}  \left( \frac{\pi^2 }{L^2}  - \curvlowbnd  \right ) 
\nonumber
\\
&<-\langle W,  \gamma '(0) \rangle.
\nonumber
\end{align}
Observe that, because $\langle Y (0),  \gamma '(0) \rangle =0$, the first order variation of the length is zero.
We conclude that the length of at least one of the curves $\psi_{\alpha} (t)$ and $\tilde{\psi}_{\alpha}(t)$ decreases in the second order as $\alpha$ increases.

At the same time, the paths $\psi_{\alpha} (t)$ and $\tilde{\psi}_{\alpha}(t)$ end and start, respectively, at the second order Taylor approximation of $\eta(\alpha)$.
Furthermore, the distance between $\eta (\alpha)$ and its second order Taylor approximation is zero up to second order. Hence, at least one of the paths $\{\gamma(t)| t\in \left[0,\pm \tfrac{L}{2}\right] \}$ is not the shortest geodesic to $\M$ --- contradicting our assumption.
\end{proof}

\begin{remark} 
The assumption that $\M$ is $C^2$ can be removed with some additional technical work. Indeed, it is known~\cite{lytchak2004geometry,lytchak2005almost} that submanifolds of positive reach are $C^{1,1}$, meaning that the tangent bundle is Lipschitz. And one can locally smoothen $C^{1,1}$ manifolds without (significantly) decreasing their reach.\footnote{Doing so globally is not as easy as one may expect and it will be reported on in a different publication. } We refer the reader to~\cite{hirsch2012differential} for an introduction to smoothing.
\end{remark}

\section{Alternative proofs} \label{app:Alternative_proof} 

\begin{Altproof}[of Theorem \ref{theorem:geometric_argument}]
    Let us prove that the set $(q + \Nor(q, \Su))\cap B(q,\reach) \cap
    P^{\boxplus r}$ is star-shaped with respect to $q$. For this,
    consider a point $x \in (q + \Nor(q, \Su))\cap B(q,\reach) \cap
    P^{\boxplus r}$ and let us prove that the segment $xq$ is also
    contained in $P^{\boxplus r}$. We consider two cases. First,
    suppose that $\|x-q\| \leq \alpha$. In that case, $xq \subseteq
    B(q,\alpha) \subseteq \Su^{\boxplus \alpha} \subseteq P^{\boxplus
      \alpha}$ and we are done. Second, suppose that $\|x-q\| >
    \alpha$ as illustrated on Figure \ref{fig:alternate-proof}.

  \begin{figure}[hbt]
  \def\svgwidth{0.90\linewidth}
  \centering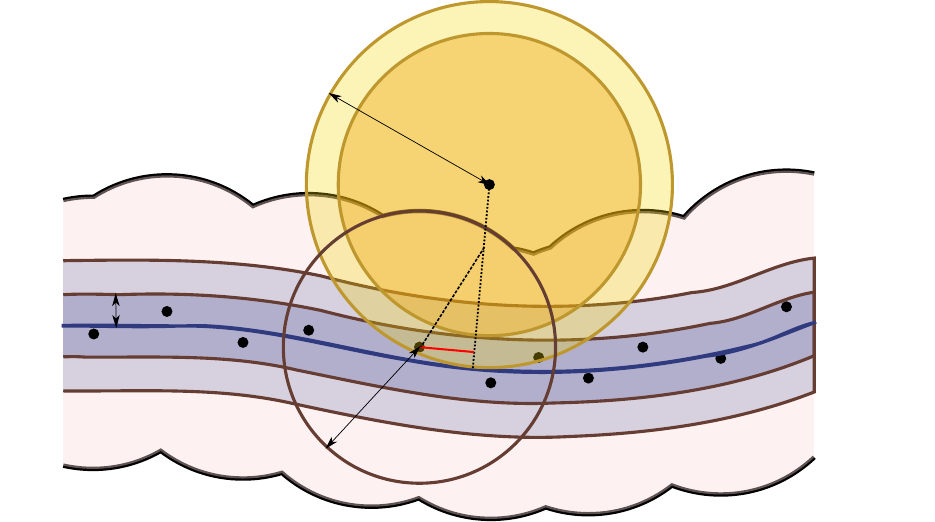
   \caption{For the alternate proof of Theorem \ref{theorem:geometric_argument}.
     \label{fig:alternate-proof}}
  \end{figure}

  In that case, $x \neq q$ and the half-line with origin at $q$ and
  passing through $x$ is well-defined.  Let $y$ be the point on this
  half-line whose distance to $q$ is $\alpha$.  Let $z$ be the point
  on this half-line whose distance to $q$ is $\reach$. Because $x \in
  B(q,\reach)$, we have that $x$ lies on the segment $qz$.  Let $p$ be
  any point of $P$ whose distance to $x$ is smaller than or equal to
  $r$. {It is this assumption that later gives $\|x-p\| \leq r$.} Let
  $p'$ be the projection of $p$ onto the straight-line passing through
  $q$ and $x$. We have that the five points $x$, $y$, $z$, $q$ and
  $p'$ are aligned and $y$ lies between $x$ and $q$. We claim that $y$
  also lies between $x$ and $p'$. The claim is clearly true if $q$
  lies between $x$ and $p'$. Let us assume that $q$ does not lie
  between $x$ and $p'$, in other words, let us assume that $p'$ is on
  the half-line with origin at $q$ and passing through $x$, as
    in Figure \ref{fig:alternate-proof}.  Let $\varphi$ be the
  internal angle of triangle $xpz$ at $x$.  The law of cosines gives:
  \begin{align} 
  \|z-p\|^2 = \|z-x\|^2 + \|x-p\|^2 - 2 \|z-x\| \|x-p\| \cos \varphi.
	\label{CosineRule} 
  \end{align} 
  By Theorem \ref{Fed4.8.12}, the interior of $B(z,\reach)$ does not
  intersect $\Su$ and because $P \subseteq \Su^{\boxplus \delta} 
	$, we
  have $\|z-p\| \geq \reach-\delta$. {By construction}, we have
  that $\|x-p\| \leq r$. Furthermore, $\|z-x\| = \|z-q\| - \|q-x\|
  \leq \reach -\alpha$. It follows that
  \begin{align}
    2 \|z-x\| \|x-p\| \cos \varphi &= \|z-x\|^2 + \|x-p\|^2 - \|z-p\|^2 
		\tag{reshuffling \eqref{CosineRule}}
		\\
    &\leq  (\reach-\alpha)^2 + r^2 - (\reach-\delta)^2
		\nonumber\\
    &\leq 0,
		\nonumber
  \end{align}
  showing that $\cos \varphi \leq 0$, or equivalently $\varphi \geq
  \frac{\pi}{2}$. Hence, $p'$ lies on the segment $qx$. Let us show
  that $\|q-p'\| \leq \alpha$. Because $p'$ belongs to the segment $qx \subseteq qz$, we have
  \begin{align*}
    \|q-p'\| &= \|q-z\| - \|p'-z\| \\
    &= \reach - \sqrt{\|z-p\|^2 - \|p'-p\|^2} \\
    &\leq  \reach - \sqrt{\|z-p\|^2 - \|x-p\|^2} \\
    &\leq \reach - \sqrt{(\reach - \delta)^2 - r^2} \\
    &\leq \alpha = \|q-y\|.
  \end{align*}
  Hence, $y$ lies between $x$ and $p'$. This shows that the distance
  to $p$ decreases as we move along the segment $xy$, starting from $x$
  and going toward $y$. It follows that $xy \subseteq B(p,r)
  \subseteq P^{\boxplus r}  
	$. Since $yq \subseteq B(q,\alpha) \subseteq
  \Su^{\boxplus \alpha}  
	\subseteq P^{\boxplus r} 
	$, we deduce that the
  whole segment $xq$ belongs to $P^{\boxplus r}$. 
	The proof is
    completed by using the same deformation retract argument as in the
    first version of the proof.
\end{Altproof}

As mentioned in the Remark \ref{remark:ExtendedInterval}, the interval \eqref{EQ:InvervalrSetPosReach} can be extended somewhat to 
  \begin{equation}
    \tag{\ref{equation:intervalForPositiveReach}}
    r \in  \left  [ \frac{1}{2} \left(\reach+\varepsilon  -  \sqrt{\Delta}\right),  
      \sqrt{\frac{1}{2} (\reach-\delta)^2 + \frac{1}{2} (\reach+\varepsilon) \sqrt{\Delta}}  \right  ].
  \end{equation}
This is a consequence of this slightly more complicated proof. 
\begin{Altproof}[
of Proposition \ref{theorem:DeformRetractsTheoremForManifolds}]
Combining the bound from Lemma
\ref{lem:Bounds_alpha_sets_pos_reach} with the conditions of Theorem
\ref{theorem:geometric_argument}, we obtain that the union of balls $P^{\boxplus r}$ 
deformation-retracts onto $\Su$ along the closest point
projection as soon as the following two inequalities are satisfied:
\begin{equation}
  \label{equation:ConditionTubularCoverForPositiveReach}
(\alpha + \varepsilon)^2 \leq r^2 \leq (\reach-\delta)^2 - (\reach-\alpha)^2.
\end{equation}
In particular, the inequality between leftmost and rightmost members,
which needs to be satisfied for a non-empty range of values for $r$ to exist, can be rearranged as:
\[
2\alpha^2 + 2\alpha(\varepsilon-\reach) - (\reach-\delta)^2 + \varepsilon^2 + \reach^2 \leq 0.
\]
Using the abc-formula for quadratic equations, the above inequality is
satisfied for all $\alpha \in [\alpha_{\min}, \alpha_{\max}]$, with
\begin{align} 
\alpha_{\min} &=  \frac{1}{2} \left(\reach-\varepsilon - \sqrt{ \Delta }\right),
\nonumber
\\
\alpha_{\max} &=  
\frac{1}{2} \left(\reach-\varepsilon + \sqrt{ \Delta }\right),
\nonumber
\end{align} 
where the discriminant is
\[
\Delta = 2(\reach-\delta)^2 - (\reach+\varepsilon)^2.
\]
The interval $[\alpha_{\min},\alpha_{\max}]$ is non-empty whenever the discriminant is positive, that is, whenever $\varepsilon
+ \sqrt{2} \, \delta \leq (\sqrt{2} - 1) \reach$, which we have assumed to be true. We thus deduce that for
all $r$ such that
\[
(\alpha_{\min} + \varepsilon)^2 \leq r^2 \leq (\reach-\delta)^2 - (\reach-\alpha_{\max})^2,
\]
or equivalently for all $r$ that satisfies (\ref{equation:intervalForPositiveReach}),
we can find $\alpha \in  [\alpha_{\min}, \alpha_{\max}]$ that satisfies the inequalities in (\ref{equation:ConditionTubularCoverForPositiveReach}). Hence, for that $\alpha$, the assumptions of Lemma \ref{lem:Bounds_alpha_sets_pos_reach} are satisfied, which in turn implies that Theorem \ref{theorem:geometric_argument} can be applied, allowing us to conclude the proof.
\end{Altproof}


\section{Previous work on the reach and medial axis in Riemannian manifolds}
\label{sec:reach_history}
The reach and medial axis in Riemannian manifolds have been studied intensely in the past by Kleinjohann \cite{kleinjohann1980convexity,kleinjohann1981nachste} and Bangert \cite{bangert1982sets}, see also~\cite{ReachSubmanifolds}.
We introduce Bangert's definition, which makes Kleinjohann's definition a little more precise. 
The \emph{unique projection point set} is the complement of the medial axis $\textrm{ax}_{\N} (\Su)$ (defined by Equation~\eqref{eq:medial_axis_Riemannian}). It is defined as
\[
\operatorname{Unp} (\Su) : = \{ q \in \N \mid \operatorname{Card} (\pi_\Su(q)) = 1 \},    
\]
where $\operatorname{Card}(A)$ denotes the cardinality of the set $A$. 
With this notation, Bangert defines the local feature size\footnote{Bangert follows Federer and writes $\rch(p,\Su)$ for the local feature size.} ($\lfs_\Su^B$) and the reach $\rch^{B}(\Su)$ as follows: 
\begin{definition}[{Bangert's reach \cite{bangert1982sets}}] 
	The \emph{local feature size} of a point $p \in \Su$ is defined as
	\[ 
	\lfs_\Su^B (p) := \sup \{ r \geq 0 \mid B(p,r) \subseteq \operatorname{Unp} (\Su) \}. 
	\]
	The \emph{reach} of the set $\Su$ is given by $\rch^{B}(\Su) = \inf_{ p \in \Su} \lfs_\Su^B (p)$. 
\end{definition} 

This definition is not sufficient for our purposes. 
Indeed, consider an example where $\N$ is a sphere, and $\Su$ is a point on it. In this setting, $\operatorname{Unp} (\Su)=\N$, and thus $\rch^B(\Su) =\infty$. As a result, the cut locus (Definition~\ref{definition:SetCutLocus}) is ignored, which is not possible in our setting. In particular, because the reach is infinite one would expect that the ball centred at the point itself has the homotopy type of a point for any radius. However, clearly, this ball has the homotopy type of the sphere itself once the radius of the ball is $\pi$ times the radius of sphere.

Recently, Boissonnat \emph{et al.}~\cite{ReachSubmanifolds} suggested to add the injectivity radius as a bound on the reach. However, this too is slightly suboptimal. This is illustrated by the example\footnote{This example was also used in \cite{ReachSubmanifolds}.} of the cylinder, that is the product of $\mathbb{R}$ and the circle $S$. With a slight abuse of notation we'll also refer to the symmetrically embedded circle in $S\times \mathbb{R}$ as $S$ as well. Let $P \subset S$ be a sample. For any sufficiently\footnote{It suffices for $S \subseteq P^{\boxplus r}$. } large radius $r$ the thickening $P^{\boxplus r}$ deformation retracts onto $S$. 
{Another (compact and slightly more sophisticated) example illustrating the suboptimality of including the injectivity radius in the definition of the reach is a subset of the flat torus $(\varepsilon \Sphere^1) \times \Sphere^1$, where $0< \varepsilon < 1$. The sectional curvature (or Gaussian curvature) of this flat torus is identically equal to zero and its injectivity radius is $\pi \varepsilon$.
Let $0 < \theta< \pi - \varepsilon$. We now consider the set 
\[ [-\theta , \theta]_{\Sphere^1}  := \{( \cos t , \sin t) \in \Sphere^1 \mid t \in [-\theta , \theta] \} .\] 
The cut locus reach of the set $\Su= (\varepsilon \Sphere^1) \times [-\theta , \theta]_{\Sphere^1} \subset (\varepsilon \Sphere^1) \times \Sphere^1$
  is $\pi - \theta$. The cut locus reach corresponds to the upper bound of the range $r$ for which the offset $P^{\boxplus r}$ of a sample $P\subset (\varepsilon \Sphere^1) \times [-\theta , \theta]_{\Sphere^1} $ deformation retracts onto $\Su$.
}

%
%

It therefore makes sense to use the cut locus reach in our context   
 and not the reach as defined in \cite{ReachSubmanifolds}. 
We should admit that in some contexts (in particular when triangulating submanifolds) it is convenient for the balls in the ambient manifold to be topological balls, 
which is not the case for large radii here. 
In~\cite{ReachSubmanifolds}, Bangert's reach is called the pre-reach.


\section{The cut locus is the singular set of the distance function}\label{section:CutLocusIsSingularSet}

{ This section in the appendix has been added for the convenience of the reader (and to maintain the anonymity of the authors), but the results are part of a separate {larger} project, the full details of which will be reported upon later. }

\subsection{generalized gradient of the distance function in  Riemannian manifolds}\label{sec:summary_albano}
In this section we recall some definitions and results from \cite[Section 4]{albano2013singular}.

We denote by $T^{\ast}_p\N$ the dual space\footnote{This duality refers to the duality of linear spaces not of cones. In particular this has nothing to do with the $\operatorname{Dual}$ as used by Federer. } 
of $T_p\N$. We use the natural map from the tangent space to its dual induced by the metric. In particular,  for $u\in T^{\ast}_p\N$,
$u^{\ast} \in T_p\N$ denotes its dual vector, defined by: $\forall v  \in T_p\N, \langle u^{\ast}, v \rangle = u(v)$.

We now recall a version of Hadamard's definition of differentiation \cite{duistermaat2004multidimensional} for manifolds.  
If $f:\N \rightarrow \R$ is a smooth function and $df \in T^{\ast}_p\N$ its differential at $p$, then:
\begin{equation}\label{eq:gradientRiemann}
\forall  w\in T_p\N,  \quad f(\exp_p(w)) - f(p) = \langle df(p)^{\ast} , w \rangle + o(|w|), \quad w \rightarrow 0,
\end{equation}
where $\langle\cdot,\cdot\rangle : T_p\N \times T_p\N \rightarrow \R$ denotes the Riemannian inner product 
at $p$ and $\exp$ the exponential map.

\begin{definition}[Superdifferential \cite{albano2013singular}]\label{definition:superDifferential}
Consider a (not necessarily smooth) function  $f:\N\rightarrow\R$.
We say that $v  \in T_p\N$ belongs to the {\em superdifferential} of $f$ at $p$, denoted $d^+f(p)$, if:
\begin{equation}\label{eq:GenetralizedGradientRiemann}
\forall  w\in T_p\N,  \quad f(\exp_p(w)) - f(p) \leq  \langle v, w \rangle + o(|w|), 
\end{equation}
as $w \rightarrow 0$.  
\end{definition}

%
\begin{remark}
By definition, $d^+f(p)$ is a convex closed subset of $T_p\N$.
Moreover, the superdifferential $d^+f(p)$ is uniformly bounded for every $p$ in some open set $U$ if and only if $f$ is uniformly Lipschitz in $U$.
\end{remark}
\begin{remark}
Note also that  $f$ is differentiable at $x$
if and only if $d^+f(p)$ is a single point, in which case $d^+f(p) = \{df(p)^{\ast}\}$.
\end{remark}

In this appendix we make the following global assumption and use the following abbreviated notation: 
 {
	\begin{tcolorbox} We assume the Riemannian manifold $\N$ is complete and at least $C^2$. If $\Su$ is a closed subset of $\N$ we write $\rho_{\Su} : \N \rightarrow \R$ for the distance to $\Su$:
\[
\rho_{\Su} (p) := d(p, \Su) 
\]
\end{tcolorbox}
}




In \cite{LIEUTIERhomotopytype} it was proven that any open bounded subset in Euclidean space has the same homotopy type as its medial axis. 
Albano \emph{et al.} \cite{albano2013singular} extended this result to any open bounded subset $\Omega$ of a Riemannian manifold.  
The proof (in \cite{albano2013singular}) made use of more sophisticated tools from non-smooth analysis \cite{Clarke1990}, namely the properties of semi-concave functions \cite{albano1999structural}, compared to the tools in \cite{LIEUTIERhomotopytype}. These techniques shortened the proof of \cite{LIEUTIERhomotopytype} as well as allowing the extension to the Riemannian setting.  

{However the formulation of \cite{albano2013singular} diverges quite significantly from the standard definition (in computational geometry and topology) of the set of interest and instead hacks back to Thom's work on singularity theory and in particular his results on the singularities of the cut locus \cite{thom1972cut}.   
}
In Albano \emph{et al.}'s main Theorem \cite[Theorem 5.3]{albano2013singular}, the medial axis  is replaced, in the Riemannian context, by the singular set of the distance function $\rho_{\partial \Omega}$ to the boundary $\partial \Omega$ of $\Omega$, set of points  where  $\rho_{\partial \Omega}$ is not differentiable. 
In the introduction \cite[Page 3]{albano2013singular} we read that
{\em ``the singular set of the distance function is closely related to the cut-locus of the boundary of $\Omega$''}, as part of the motivation for their work. 
However, no more formal assertion is given in \cite{albano2013singular}.
 
Our Theorem \ref{theorem:ClosestDirectionIsMinusGradient} below asserts that the singular set {\em is} the cut locus of  $\partial \Omega$. Although our initial motivation for this work was the homotopy learning result in the main body of this text we believe that our practical characterization of the singular set of the distance function will be of more general use in computational geometry and topology.

The main idea of the proof in \cite{albano2013singular} resembles the core of \cite{LIEUTIERhomotopytype} quite closely. More precisely, the authors build a continuous flow $\Phi :\Omega \times [0, \infty) \rightarrow \Omega$
 induced by a generalized gradient of $\rho_{\partial \Omega}$ as defined in Definition \ref{definition:superDifferential}.
 This flow is proven to realize a homotopy equivalence (more precisely a weak deformation retraction) between $\Omega$
 and the singular set of $\rho_{\partial \Omega}$. 
{The flow, in the setting of \cite{albano2013singular}, ``pushes'' points in the open set $\Omega$ inside $\Omega$ away from its boundary $\partial \Omega$.
 In our setting we consider a closed $\Su= \Omega^c$ and the same flow $\Phi_{\Su}: \Su^c \times [0, \infty) \rightarrow \Su^c$ ``pushes''
points in  $\Su^c$ away from $\Su$,  in the direction of steepest ascent of $\rho_{\Su}$. }
%
%

\subsection{Result: the characterization of the singular set}\label{sec:singular_set_chracterization}

Our main result is more general than just the characterization of the singular set. We give a geometric interpretation of the superdifferential. For this geometric interpretation we have to make the following definition. 

\begin{definition}[Directions of shortest paths to $\Su$]
Let $\Su$ be a closed subset of the complete Riemannian manifold $\N$. If $p$ is a point in $\N \setminus \Su$, then the {\em directions of shortest paths to $\Su$ at $p$},
denoted $\Gamma_{\Su}(p)$ is the subset of the unit sphere in $T_p\N$ of all directions of minimizing geodesics from $p$ to $\Su$, that is,
\begin{equation}\label{eq:DefinitionGamma}
\Gamma_{\Su}(p) :=  \frac{1}{\rho_{\Su}(p)} \exp_p^{-1}  \left(B(p, \rho_{\Su}(p) ) \cap \Su  \right) .
\end{equation}
\end{definition}

Having defined our geometric interpretation we are ready to state our main result on the superdifferential. 
\begin{theorem}\label{theorem:ClosestDirectionIsMinusGradient}
If $\Su$ is a closed subset of the complete Riemann manifold $\N$ and $p$ a point in $\N \setminus \Su$, then
\begin{equation}\label{eq:GenetralizedGradientIsMinusClosetsDirections}
d^+ \rho_{\Su} (p) =  - \Hull \left( \Gamma_{\Su}(p) \right),  
\end{equation}
where $\Hull(\cdot)$ denotes the convex hull.
\end{theorem}
Section \ref{section:proofOfClosestDirectionIsMinusGradient} is dedicated to the proof of Theorem \ref{theorem:ClosestDirectionIsMinusGradient}. 


{
The singular set of $\rho_{\partial \Omega}: \Omega \rightarrow \R$, where $\rho_{\partial \Omega}$ is as defined in $\cite{albano2013singular}$,
is the singular set of our $\rho_{\Su}: \Su^c \rightarrow \R$ in our setting. 
This singular set is by definition the set of points where $\rho_{\Su}$ is not differentiable, it corresponds to 
the set of points $p$  where $d^+ \rho_{\Su} (p)$ is a not a singleton, and thus, by Theorem 
\ref{theorem:ClosestDirectionIsMinusGradient}, is the set of points where $ \Gamma_{\Su}(p)$ is a not a singleton,
which in turn coincides with the set of point $p$ where there are more than one minimizing geodesics to $\Su$.
It follows that:
\begin{corollary}
The singular set of  both our distance function and the distance function as defined in \cite{albano2013singular}
  coincides with the cut locus of Definition \ref{definition:SetCutLocus}.
\end{corollary}
}

The rest of this section will discuss the consequences of this central result. 
In particular we are working towards characterizations of the medial axis and normal cones that imitate Federer's definitions in Euclidean space as closely as possible.

We now need to recall/introduce some notation. 
We write $\Phi_{\Su} :\N \setminus \Su \times [0, \infty) \rightarrow \N \setminus \Su$ for the outward directed  flow.  
As in \cite{albano2013singular}, the trajectory of a single point $p \in \N \setminus \Su$
is denoted by $\gamma: [0, \infty) \rightarrow \N \setminus \Su$, so that $\gamma(t) = \Phi_{\Su}(p,t)$ and in particular $\gamma(0) = p$. 


It is convenient for us to  give an explicit expression of the relation between the superdifferential and the (right) derivative of $t \mapsto \Phi_{\Su}(p, t)$ below.
 This relation  is only implicit in \cite[Theorem 4.4]{albano2013singular}.

\begin{lemma}\label{lemma:RightDerivativeOfFlowIsProjOnSupergradient}
For all $p  \in \N \setminus \Su$, $t\mapsto \Phi_{\Su}(p, t)$ is Lipschitz. In particular it is differentiable  for {\em almost all} $t\in [0, \infty)$  and
\begin{equation}\label{eq:DerivativeIsProjOnSuperGradientAE}
\frac{d}{dt}\Phi_{\Su}(p, t) = \pi_{d^+ \rho_{\Su} (\Phi_{\Su}(p, t) )} (0) \quad a.e.
\end{equation}
where $\pi_{d^+ \rho_{\Su} (p)} (0) :=  \argmin_{w \in d^+ \rho_{\Su} (p)} |w|$ is the orthogonal projection of $0$ on the convex set $d^+ \rho_{\Su} (p)$. 

Moreover,  for all $p  \in \N \setminus \Su$, the map $t\mapsto \Phi_{\Su}(p, t)$ is right differentiable for {\em all} $t\in [0, \infty)$ and:
\begin{equation}
\label{eq:RightDerivativeIsProjOnSuperGradient}
\frac{d}{dt^+}\Phi_{\Su}(p, t) = \pi_{d^+ \rho_{\Su} (\Phi_{\Su}(p, t) )} (0) 
\end{equation}
where $\frac{d}{dt^+}$ denotes the right derivative with respect to $t$.


\end{lemma}
\begin{proof}
 We follow \cite[Section 4]{albano2013singular} and denote the supergradient of $\rho_{\Su}$ at point $p$ by $C$, that is $C := d^+ \rho_{\Su} (p)$.
Theorem \cite[Theorem 4.4]{albano2013singular} tells 
us that $y\mapsto \gamma(t)$ is Lipschitz and a has a right derivative everywhere.
(Actually it is $1$-Lipschitz by Theorem \ref{theorem:ClosestDirectionIsMinusGradient}).

Note that Equations (4.5) and  (4.6) in  \cite[Theorem 4.4]{albano2013singular} give,
\begin{equation}\label{eq:GammaPrimeInC}
\gamma'(t)  \in C := d^+ \rho_{\Su} (p),
\end{equation}
{where  $\gamma' (t) := \frac{d}{dt^+} \gamma(t)=\frac{d}{dt^+}\Phi_{\Su}(p, t)$ is the right derivative of $t \mapsto \gamma(t)$,}
and
\begin{equation}\label{eq:DistanceDertivativeIsSquareNormGradient}
\frac{d}{dt^+} \rho_{\Su} (\gamma(t) )=   \left\langle \gamma'(t) , \gamma'(t)   \right\rangle 
\end{equation}
respectively in our setting.
Since $\rho_{\Su}$ is Lipschitz, one has:
\begin{equation}
\frac{d}{dt^+} \rho_{\Su} (\gamma(t) )= \frac{d}{du^+}_{\mid u=0}   \rho_{\Su} \left( \exp_{\gamma(t)}  u \gamma'(t) \right).
\label{eq:RightDerivativeRhoS}
\end{equation}
By writing out the definition of the super differential, we see that
\begin{align}
\forall  v \in C,
\forall  w\in T_p\N,  \quad \rho_\Su (\exp_p(w)) - \rho_\Su(p) \leq  \langle v, w \rangle + o(|w|).
\nonumber
\end{align}
In particular, taking $w= u \gamma'(t)$, we have
\begin{align}
\forall  v \in C,
\quad \rho_\Su (\exp_p(u  \gamma'(t))) - \rho_\Su(p) \leq  
u \langle v, \gamma'(t) \rangle + o(|u \gamma'(t)|).
\nonumber
\end{align}
Thanks to the definition of the right derivative we have that \[ \lim _{u\searrow 0} \frac{\rho_\Su (\exp_p(u  \gamma'(t))) - \rho_\Su(p)}{u} = \frac{d}{du^+}_{\mid u=0}   \rho_{\Su} \left( \exp_{\gamma(t)}  u \gamma'(t) \right) , \] 
so that together with \eqref{eq:RightDerivativeRhoS}, we find 
\begin{equation}
\forall v \in C, \:  \frac{d}{dt^+} \rho_{\Su} (\gamma(t) ) \leq \left\langle v, \gamma'(t) \right\rangle.
\nonumber
\end{equation}

Combining this with \eqref{eq:DistanceDertivativeIsSquareNormGradient} yields 
\begin{equation}
\forall v \in C,\:  \left\langle \gamma'(t) , \gamma'(t)   \right\rangle \leq  \left\langle v, \gamma'(t) \right\rangle , 
\nonumber
\end{equation}
which can be reshuffled into 
\begin{equation}
\forall v \in C,  \: \left\langle v - \gamma'(t) , \gamma'(t)   \right\rangle \geq 0.
\nonumber
\end{equation}
We finally rewrite this identity as
\begin{equation}
\forall v \in C, \:  \left\langle v  , v   \right\rangle  -    \left\langle  \gamma'(t) , \gamma'(t)   \right\rangle =   \left\langle  v - \gamma'(t) , v - \gamma'(t)   \right\rangle + 2 \left\langle v - \gamma'(t) , \gamma'(t)   \right\rangle \geq 0.
\nonumber
\end{equation}
Since $\gamma'(t) \in C$  by  \eqref{eq:GammaPrimeInC}, this shows that $\gamma'(t)$ is the point in $C$ closest to $0$ which proves 
\eqref{eq:RightDerivativeIsProjOnSuperGradient}.
Because $\gamma$ is Lipschitz, it is differentiable almost everywhere
and, when it is, its derivative is equal to its right derivative, which gives \eqref{eq:DerivativeIsProjOnSuperGradientAE}.
\end{proof}

\begin{remark}
Since $t \mapsto \Phi_{\Su}(p, t)$ is $1$-Lipschitz, it is differentiable almost everywhere and is the integral of its derivative.
\end{remark}

%

Following the flow of the distance function for time $\tau$  decreases the distance to $\Su$ with $\tau$
 if one stays outside the cut locus. More precisely we have,
\begin{lemma}\label{lemma:DerivativeOfDistanceIsOneIffNotInCutLocus}
For any $p  \in \N \setminus \Su$:

\begin{equation}\label{eq:DerivativeOfDistanceIsOneIffNotInCutLocus}
\frac{d}{dt^+} \rho_{\Su} \left( \Phi_{\Su}(p, t) \right) 
= \left| \frac{d}{dt^+} \Phi_{\Su}(p, t)  \right|^2
\begin{cases} 
= 1 \quad \text{if} \quad    \Phi_{\Su}(p, t) \notin \operatorname{cl}_{\N}(\Su) \\
< 1 \quad \text{if} \quad    \Phi_{\Su}(p, t) \in \operatorname{cl}_{\N}(\Su)
\end{cases}
\end{equation}

\end{lemma}
\begin{proof}
The first equality is exactly Equation (4.6) in  \cite[Theorem 4.4]{albano2013singular} rewritten in our notation and for our setting.
By definition of the cut locus, $\Phi_{\Su}(p, t) \in \operatorname{cl}_{\N}(\Su)$ if and only if $ \Gamma_{\Su}(\Phi_{\Su}(p, t))$ contains at least two points.
Since $ \Gamma_{\Su}(\Phi_{\Su}(p, t))$ is a subset of the unit sphere in $T_{\Phi_{\Su}(p, t)}\N$, we get that the projection of $0$ on the convex hull of 
$ \Gamma_{\Su}(\Phi_{\Su}(p, t))$ is strictly inside the unit ball if and only if $\Phi_{\Su}(p, t)$ is on the cut locus.
The second equality/inequality  follows then by  Theorem \ref{theorem:ClosestDirectionIsMinusGradient} 
and \eqref{eq:RightDerivativeIsProjOnSuperGradient} of Lemma \ref{lemma:RightDerivativeOfFlowIsProjOnSupergradient}.
\end{proof}

We also have the following result which is reminiscent of part of Theorem 4.8 (12) of \cite{Federer} and improves a result that is implicit in the work of Kleinjohann \cite{kleinjohann1981nachste}.    
\TrajectoryAreMinimisingGeodesics*
\begin{proof}
For $t \in [0, \rho - \rho_{\Su}(p)]$ one has $\rho_{\Su}(\Phi_{\Su}(p, t)) \leq \rho_{\Su}(p) + t \leq \rho$, by \eqref{eq:DerivativeOfDistanceIsOneIffNotInCutLocus}. 
Therefore $\Phi_{\Su}(p, t)$ is not in the cut locus, so that 
there is a unique minimizing geodesic from  $\Phi_{\Su}(p, t)$ to $\Su$.
Moreover, again thanks to \eqref{eq:DerivativeOfDistanceIsOneIffNotInCutLocus},
the length of $\Phi_{\Su}(p, [0,t])$ is $t$ and $\rho_{\Su}(\Phi_{\Su}(p, t)) =  \rho_{\Su}(p) + t $, 
 so that the length of the  concatenation of the minimizing geodesic from $p$ to $\Su$
with the trajectory $\Phi_{\Su}(p, [0,t])$ is $\rho_{\Su}(p) + t = \rho_{\Su}(\Phi_{\Su}(p, t))$,
 which proves the claim.
\end{proof}

We also improve Kleinjohann's result on `Dilatationen' \cite[Satz 3.2 and 3.3]{kleinjohann1981nachste}. For this we recall the following notation.  
The complement of the open offset of $\Su$ is denoted:
\begin{equation}
\CplOffset^\rho(\Su) := \left\{ p \in \N  \mid \rho_\Su(p) \geq \rho \right\}
\tag{\ref{eq:ComplementOffset}}
\end{equation}
\reversedFlow*
\begin{proof}
Since $\rho'< \rchcl_{\N}(\Su)$, Lemma \ref{lemma:TrajectoryAreMinimisingGeodesics} yields that for any $p \in \CplOffset^{\rho'}(\Su) \setminus \Su$, there is a minimizing geodesic from 
 $\Phi_{\Su}(p,  \rho' - \rho_{Su}(p) )$ to $ \pi_{\Su} (p)$.

Equation \eqref{eq:DerivativeOfDistanceIsOneIffNotInCutLocus} in Lemma \ref{lemma:DerivativeOfDistanceIsOneIffNotInCutLocus},
gives that for $t \in [0,  \rho' - \rho_{Su}(p) ]$, $\frac{d}{dt^+}_{\mid t=t'} \rho_{\Su} \left( \Phi_{\Su}(p, t) \right) =1$.
 It follows that $\rho_{\Su} ( \Phi_{\Su}(p, \rho' - \rho_{\Su}(p) ) ) = \rho'$, so that 
 $\Phi_{\Su}(p, \rho' - \rho_{Su}(p) ) \in \partial  \CplOffset^{\rho'}(\Su)$
  and $\rho_{\CplOffset^{\rho'}(\Su)} (p) \leq  \rho' - \rho_{\Su}(p)$. 
By the triangle inequality $\rho_{\Su} ( p) + \rho_{\CplOffset^{\rho'}(\Su)} (p) \geq \rho'$ and hence we get
\begin{align}
\forall p \in \CplOffset^{\rho'}(\Su) \setminus \Su, \: \rho_{\Su} ( p)  + \rho_{\CplOffset^{\rho'}(\Su)} (p) = \rho'.
\nonumber
\end{align}
In other words we have that $\rho_\Su$ is minus $\rho_{\CplOffset^{\rho'}(\Su)}$ up to a constant. It now follows fromt the definition of the supergradient that  
\begin{equation}
d^+  \rho_{\CplOffset^{\rho'}(\Su)} (p) = - d^+  \rho_{\Su}  (p).
\label{eq:SuperDiffEqualsMinusSuperDiffOffset}
\end{equation}

Because $d^+  \rho_{\Su}  (p)$ is a single point, \eqref{eq:SuperDiffEqualsMinusSuperDiffOffset} yields that $d^+  \rho_{\CplOffset^{\rho'}(\Su)} (p)$ is also a single point. This in turn implies that $\rho_{\CplOffset^{\rho'}(\Su)}$ is differentiable at $p$. 
Hence, by Theorem \ref{theorem:ClosestDirectionIsMinusGradient}, $p$ is not in the cut locus of
$\CplOffset^{\rho'}(\Su)$ which gives  \eqref{eq:ReachReversed}.

The second claim of the lemma and \eqref{eq:HomotopyReversedFlow} follow then by considering the 
flow associated to the set $\CplOffset^{\rho'}(\Su)$  (or its restriction of he flow to $\Su^{\boxplus \rho}$).  
The flow is continuous with respect to both $p$ and $t$ and sends the complement $\CplOffset^{\rho'}(\Su)^c$ of $\CplOffset^{\rho'}(\Su)$ to $\partial \Su$.
\end{proof}

%
%
%

\subsection{Proof of the geometric interpretation of the superdifferential 
}\label{section:proofOfClosestDirectionIsMinusGradient}

Before proving Theorem \ref{theorem:ClosestDirectionIsMinusGradient} we need a few definitions and lemmas.
However, we first make some simple observations. 
First observe that, for any $\rho$ such that $0< \rho \leq \rho_{\Su}(p)$, one has $ \rho_{\Su^{\boxplus \rho_{\Su}(p) - \rho}} (p) = \rho$ and
\begin{equation}\label{eq:DefinitionGamma2}
\Gamma_{\Su}(p) =  \frac{1}{\rho} \exp_p^{-1} \left( B(p, \rho)  \cap \Su^{\boxplus \rho_{\Su}(p) - \rho} \right) = \Gamma_{\Su^{\boxplus \rho_{\Su}(p) - \rho}}(p) .
\end{equation}
Secondly, because $ \rho_{\Su^{\boxplus \rho_{\Su}(p) - \rho}} (x) =  \rho_{\Su}(x)  - \rho_{\Su}(p) + \rho$, 
 one has 
\begin{equation}\label{eq:GradinetSameToOffset}
d^+ \rho_{\Su} (p) = d^+  \rho_{\Su^{\boxplus \rho_{\Su}(p) - \rho}} (p)
\end{equation}
for any $p$ in $\N \setminus \Su^{\boxplus \rho_{\Su}(p) - \rho}$.

The proof below exploits the fact that the exponential map is locally a diffeomorphism, in particular we have: 
\SavedComment{\color{red} [I think it would be better to write $\rho_C$ for the radius of the lemma below, because if you would quantify it you would get the strong convexity radius]}
\begin{lemma}\label{lemma:InverseExpIsLipschitz}
Let $\N$ be  a complete smooth (at least $C^2$) Riemann manifold. For any $p \in \N$, there and $\rho>0$ and $\lambda>0$ such that:
\begin{itemize}
\item[(1)]  $2 \rho$ is smaller than the injectivity radius of $\N$,
\item[(2)]  for any $p \in \N$, the maps $\exp_p$ and $\exp_p^{-1}$ restricted to $B(p,2 \rho)$ are  $\lambda$-Lipschitz,
\item[(3)]  for any $p \in \N$, the maps $d \exp_p$ and $(d \exp_p)^{-1}$ restricted to $B(p,2 \rho)$ are  uniformly continuous.
\end{itemize}
\end{lemma}
\begin{proof}
The injectivity radius of $\N$ is known to be positive. 
 Taking $2\rho$ smaller than the injectivity radius, we get that the exponential map is a $C^1$ diffeomorphism on  $B(p,2 \rho)$.
Because $B(p,2 \rho)$ is compact, conditions $(2)$ and $(3)$ follow.
\end{proof}

\begin{lemma}\label{lemma:DerivativeDistanceToPointAntUnifContDeriv}
Let $\N$ be  a  complete Riemann manifold  and $p,q \in \N$ such that  $0<d(p,q) \leq \rho$ where $\rho$ is defined in Lemma \ref{lemma:InverseExpIsLipschitz}.
Let $v \in T_p\N$ , be a unit vector such that $\{ \exp_p(tv) \mid t \in [0,d(p,q)] \}$ is the unique minimizing geodesic
between $p$ and $q$  i.e. {$d(p,q)v= \exp_p^{-1} (q)$}. Then:
\begin{equation}\label{eq:DistanceToOnePoint1}
\frac{d}{dt}_{\mid t=0} d(q, \exp_p (tw) ) = - \langle v, w \rangle
\end{equation}
In other words one has:
\begin{equation}\label{eq:DistanceToOnePoint2}
\left( d\, \rho_{\{q\}} (p) \right)^{\ast} = -v
\end{equation}
Moreover, there is $\eta>0$ such that the map
$t \mapsto  d(q, \exp_p (tw) )$ is $C^1$ on $(-\eta,\eta)$ and there is  a uniform modulus of continuity independent of $q$ on its derivative, formally:
\begin{align}
\forall p \in \N, \forall \epsilon>0, \exists \alpha >0 \mid \forall q \in B(p, \rho), \forall t_1,t_2 \in (-\eta,\eta), \nonumber \\
|t_1-t_2| < \alpha \Rightarrow  \left| \frac{d}{dt}_{\mid t= t_1} d(q, \exp_p (tw) ) - \frac{d}{dt}_{\mid t= t_2} d(q, \exp_p (tw) ) \right| < \epsilon  \label{eq:UniformContinuityDerivativeDistance}
\end{align}

such that  derivative of the function $t \mapsto  d(q, \exp_p (tw) )$ is 
uniformly continuous on $(-\eta, \eta)$,
\end{lemma}
\SavedComment{
\color{red} 
[Alternative for the first part: As was noted in\footnote{In \cite{FuelForTheFire} an integral is considered, but the integration can be simply omited.} Equation (3.24) of \cite{FuelForTheFire} we have the following:
If $d_\N(q,p)$ is less than half the injectivity radius and, if the upper bound on the sectional curvatures $\Lambda_u$ is positive, $d_\N(q,p)<\frac{\pi}{2 \sqrt{\Lambda_u}}$, then 
\[
d(q, \exp_p (tw) )^2 = |tw - \exp_p^{-1}(q) |^2 + \mathcal{O} (|tw |^2),  
\]
where we take the norm in the tangent space of $p$. Writing $v=\exp_p^{-1}(q)$, taking the square root and differentiating, yields \eqref{eq:DistanceToOnePoint1}.

The Rauch comparison theorem \cite[Theorem IX.2.3]{chavel2006riemannian} is a quantified version of the second part of the lemma, see also \cite[Section 3]{dyer2015riemannian}  ]
}
\begin{proof}
Because the exponential map preserves radial distances
\begin{equation}
 d(q, \exp_p (tw) )= | \exp_q^{-1} \exp_p (tw)  |,
\nonumber
 \end{equation}
so that
\begin{equation}
 d(q, \exp_p (tw) )^2 = \langle \exp_q^{-1} \exp_p (tw) , \exp_q^{-1} \exp_p (tw) \rangle.
\label{eq:distanceSq}
 \end{equation}
Since $t\mapsto \exp_{p} t v$ and $t\mapsto \exp_{p} t  (d \exp_q)_{p}^{-1} (-v)$ both parametrize the geodesic between $p$ and $q$, one also has 
\begin{equation}
q=\exp_{p} d(p,q) v \Rightarrow  p = \exp_{q}  d(p,q)  (d \exp_q)_{p}^{-1} (-v)
\label{eq:rewrite_p}
 \end{equation}
We find that 
\begin{equation} 
\exp_q^{-1} \exp_p (tw)_{\mid t=0}  =  \exp_q^{-1} (p) = (d \exp_q)_{p}^{-1} (-d(p,q) v ),
\label{eq:rewriteInvEXP}
\end{equation}
\SavedComment{\color{red} [With regard to notation I would be more inclined to put $|_{t=0}$ at the end, but I am not feeling very strongly about this either.]}
where the first equality is due to the fact that $\exp_p (0 w)=\exp_p (0)= p$ and the second equality follows from \eqref{eq:rewrite_p} and the fact that differentiation is linear.

Differentiating \eqref{eq:distanceSq} we get
\begin{align*}
\frac{d}{dt}_{\mid t=0} d(q, \exp_p (tw) )^2 &= \frac{d}{dt}_{\mid t=0} \langle   \exp_q^{-1} \exp_p (tw)  , \exp_q^{-1} \exp_p (tw) \rangle \\
&= 2 \langle  (d \exp_q)_{p}^{-1} (-d(p,q) v ), (d \exp_q)_{p}^{-1} w \rangle \\
&= 2 d(p,q)  \langle   -v , w \rangle ,
\end{align*}
\SavedComment{\color{red} [I don't get why you first shift $d(p,q)$ forward and then backwards again.]}
where the second equality holds thanks to $\frac{d}{dt}_{\mid t=0}  \exp_p (tw)  =w$ and \eqref{eq:rewriteInvEXP} and the last equality holds by applying Gauss Lemma \cite[Lemma 3.5 of Chapter 3]{do1992riemannian} to $v',w'$, 
with $v' =  (d \exp_q)_{p}^{-1} (-v )$  and $w' = (d \exp_q)_{p}^{-1} w$. We stress that with this notation one has $(-v) = (d \exp_q)_{p}(v')$
and $w =  (d \exp_q)_{p}(w')$.

\SavedComment{\Andre{I get the feeling that we could write the same proof in a simpler way..}}

For $0< \eta< \rho$  and $t\in(-\eta,\eta)$, $\exp_p (tw)$ remains in $B(p,2\rho)$.
Therefore, by Lemma \ref{lemma:InverseExpIsLipschitz}, $t \mapsto d(q, \exp_p (tw) )= | \exp_q^{-1} \exp_p (tw)  |$ is the composition of $C^1$ functions whose  modulii of continuity 
are independent of $q$ and  $\exp_p (tw)$ and this gives us \eqref{eq:UniformContinuityDerivativeDistance}.
\end{proof}

Following Federer's definition in \cite[Theorem 4.8 (2)]{Federer} we introduce the following in the Riemannian setting 
\begin{definition}[Set of closest points]
If $\Su$ is a closed subset of the complete Riemann manifold $\N$ and $p$ a point in $\N \setminus \Su$, the {\em set of closest points to $\Su$ at $p$},
denoted $\tilde{\Gamma}_{\Su}(p)$ is the set of points in $\Su$ closests to $p$:
\begin{equation}\label{eq:DefinitionGammaTilde}
\tilde{\Gamma}_{\Su}(p)=  \{ q \in \Su \mid d(p,q) =\rho_{\Su}(p) \}
\end{equation}
\end{definition}

We have the following semi-continuity of $\tilde{\Gamma}_{\Su}$:
\begin{lemma}\label{lemma:SemiContinuityGammaTilde}
If $\Su$ is a closed subset of a complete Riemann manifold $\N$, then
\begin{equation}\label{eq:SemiContinuityGammaTilde2}
\forall p\in   \N \setminus \Su, \:  \forall \varepsilon>0, \: \exists \alpha >0 \mid    B \left( p, \rho_{\Su}(p) +  \alpha  \right)  \cap  \Su  \subseteq  \tilde{\Gamma}_{\Su}(p)^{\boxplus  \varepsilon} 
 \end{equation}
and
\begin{equation}\label{eq:SemiContinuityGammaTilde}
\forall p\in   \N \setminus \Su, \:  \forall \varepsilon>0, \: \exists \alpha >0 \mid d(p',p) < \alpha \Rightarrow \tilde{\Gamma}_{\Su}(p') \subseteq  \tilde{\Gamma}_{\Su}(p)^{\boxplus \varepsilon} .
\end{equation}
\end{lemma}
\begin{proof}
Consider the sequence of sets
\begin{equation}
K_{n} = \left(   B \left( p, \rho_{\Su}(p) + \frac{1}{n}\right)  \cap  \Su  \right)  \setminus  \tilde{\Gamma}_{\Su}(p)^{\boxplus \circ  \varepsilon}
\nonumber
\end{equation}
where $\tilde{\Gamma}_{\Su}(p)^{\boxplus \circ  \varepsilon}$  denotes the ``open offset'' of $\tilde{\Gamma}_{\Su}(p)$, that is 
\[\tilde{\Gamma}_{\Su}(p)^{\boxplus \circ  \varepsilon} := \left \{   y \in \N \mid d \left( y,\tilde{\Gamma}_{\Su}(p) \right) < \varepsilon    \right \}.\]

The sets $K_n$ are compact sets and 
\begin{equation}
\bigcap_{n \in \Ninteger} K_n =  \left(   B \left( p, \rho_{\Su}(p) \right)  \cap  \Su  \right)   \setminus  \tilde{\Gamma}_{\Su}(p)^{\boxplus \circ  \varepsilon} = \emptyset .
\nonumber
\end{equation}
It follows from Cantor's intersection theorem \cite[Section 48]{munkrestopology} that for some $n$ one has $K_n = \emptyset$, so that taking $2 \alpha = \frac{1}{n}$ gives 
$\left(   B \left( p, \rho_{\Su}(p) + 2 \alpha  \right)  \cap  \Su  \right)   \setminus  \tilde{\Gamma}_{\Su}(p)^{\boxplus \circ  \varepsilon} = \emptyset$, 
that is,
\begin{equation}
    B \left( p, \rho_{\Su}(p) + 2 \alpha  \right)  \cap  \Su  \subseteq  \tilde{\Gamma}_{\Su}(p)^{\boxplus \circ  \varepsilon}.  
 \end{equation}
This already gives \eqref{eq:SemiContinuityGammaTilde2}.  
For any $p' \in B(p, \alpha)$, one has $\rho_{\Su} (p') \leq \rho_{\Su}(p) + \alpha$, which gives
\begin{equation}
\tilde{\Gamma}_{\Su}(p') =  B(p', \rho_{\Su} (p')) \cap \Su   \subseteq  B(p',  \rho_{\Su}(p) + \alpha) \cap \Su   \subseteq  B(p,  \rho_{\Su}(p) + 2 \alpha) \cap \Su  \subseteq  \tilde{\Gamma}_{\Su}(p)^{\boxplus \circ  \varepsilon} . 
 \end{equation}
This is precisely \eqref{eq:SemiContinuityGammaTilde}.
\end{proof}

To prove the main result we need a result from functional analysis.  It seems likely results similar to the following have been proven, however we have not been able to find a specific reference. 
The phrasing of the following elementary lemma is tailored to the proof of the next one.
\begin{lemma}\label{lemma:DerivativeOfInf}
Let $\eta >0$ and a set of continuous functions $F_X = \{f_x: (-\eta,\eta)\rightarrow \R \mid x \in X\}$ indexed by a set $X$
such that:
\begin{itemize}
\item[(a)] every $f_x \in F_X$ is $C^{1}$ smooth on  $(-\eta,\eta)$ and the family of derivatives function $\{f'_x \mid x \in X\}$ is uniformly equicontinuous, formally:
\begin{align}
\forall \varepsilon>0, \exists \alpha>0 \mid  \quad & \forall x \in X, \: \forall t_1,t_2 \in  (-\eta,\eta),  \nonumber \\ 
&|t_1-t_2|  < \alpha \Rightarrow \left| f'(t_2)- f'(t_1) \right| < \varepsilon,  \label{eq:UniformEquiContinuousDerivatives}
 \end{align}

\item[(b)]  the derivatives are uniformly bounded:
\begin{equation} \label{eq:UniformBoundDerivatives}
\forall x \in X, \forall t \in (-\eta,\eta),  f'_x(t) \in [-1,1].
 \end{equation}

\item[(c)]  We denote by $\MinAtZero \subseteq X$ the  set of indices of functions whose value at $0$ is minimal and we assume that $\MinAtZero$ is not empty: 
\begin{equation}
\MinAtZero :=  \left\{ y \in X \mid f_y(0) = \inf_{x\in X} f_x(0) \right\} \neq \emptyset,
\label{eq:ConditionC}
 \end{equation}

\item[(d)] we assume moreover that:
\begin{align}
\forall \varepsilon>0, &\: \exists \alpha>0 \mid  \forall x \in X \nonumber \\ 
&f_x(0) - \inf_{x'\in X} f_{x'}(0) < \alpha \Rightarrow \exists y\in \MinAtZero \mid \left| f'_x(0) -  f'_y(0) \right| < \varepsilon.\label{eq:WhenNearMinimumNearMinimumDerivative}
 \end{align}
 \end{itemize}
 
 Then the function $t \mapsto \inf_{x\in X} f_x(t)$ has a right derivative at $t=0$ and
\begin{equation}
\frac{d}{dt^+}_{\mid t=0} \: \inf_{x\in X} f_x(t)  = \inf_{x \in \MinAtZero} \:  \frac{d}{dt}_{\mid t=0}  f_x(t) .
\nonumber
 \end{equation}
\end{lemma}

\begin{proof}
We define the minimal function $f:(-\eta,\eta) \rightarrow \R$  by:
\begin{equation}
f(t) :=  \inf_{x\in X} f_x(t) 
 \end{equation}

We denote the uniform modulus of equicontinuity of \eqref{eq:UniformEquiContinuousDerivatives} by $\alpha  \mapsto \epsilon(\alpha)$, that is 
\begin{align}
\forall \alpha >0,  \quad & \forall x \in X, \: \forall t_1,t_2 \in  (-\eta,\eta),  \nonumber \\ 
&|t_1-t_2|  < \alpha  \Rightarrow \left| f'(t_2)- f'(t_1) \right| < \epsilon(\alpha),  
\label{eq:UniformEquiContinuousDerivatives2}
 \end{align}
where  $\alpha  \mapsto \epsilon(\alpha)$ is not decreasing and  $\lim_{\alpha \rightarrow 0} \epsilon(\alpha) = 0$.
We stress that we use the notation $\epsilon$ for the modulus of equicontinuity instead of the frequently used $\varepsilon$. 


For any $x\in X$, one has, by the fundamental theorem of calculus, that
\begin{equation}
f_x(t) = f_x(0) + \int_0^t f'_x(u) du = f_x(0) + f'_x(0) t + \int_0^t (f'_x(u) - f'_x(0)) du.
\nonumber
\end{equation}
This means that the first order Taylor expansion can be expressed with a uniform remainder, that is a remainder bounded by
$t \mapsto  \epsilon(t) t = o(t)$ which is independent of $x\in X$
\begin{equation}\label{eq:FirstOrderTaylorWithUniformRest}
\forall x \in X, \left| f_x(t) - \Big( f_x(0) + f'_x(0) t \Big) \right| <  \int_0^t  | f'_x(u) - f'_x(0) |  du <  \epsilon(t) t.
\end{equation}
Let us define
\begin{equation}\label{eq:DefinitionA}
a := \inf_{y \in \MinAtZero}  f'_y(0), 
 \end{equation}
and observe that, by assumption (b), one has $a\in [-1,1]$.


Because of this definition we have that, for any $\varepsilon>0$ there is $y\in \MinAtZero$ such that $f'_y(0) < a + \betaAp$.

Since $f(0) = f_y(0)$ (due to Assumption (c) or \eqref{eq:ConditionC}) and $f(t) \leq f_y(t)$ (by definition of the infimum),
{and since the remainder $t\mapsto \epsilon(t) t$ in \eqref{eq:FirstOrderTaylorWithUniformRest} is independent of $x$, } we get
\begin{equation}
f(t) - f(0) < (a+ \betaAp) t +  \epsilon(t) t.
\nonumber
 \end{equation}
Since this holds for any $\betaAp>0$, we get the following upper bound on $f(t) - f(0)$:
\begin{equation}\label{eq:UpperBoundRightDerivative}
f(t) - f(0) \leq a t +  \epsilon(t) t.
 \end{equation}

Using $f(0) =  \inf_{x'\in X} f_{x'}(0) $ in Assumption (d) or rather \eqref{eq:WhenNearMinimumNearMinimumDerivative} gives that 
\begin{align}
\forall \varepsilon>0, &\: \exists \alpha>0 \mid  \forall x \in X \nonumber \\ 
&f_x(0) -  f(0) < \alpha \Rightarrow \exists y\in \MinAtZero \mid \left| f'_x(0) -  f'_y(0) \right| < \varepsilon. 
\nonumber
 \end{align}
Changing notation this means that 
for any $\betaAp'>0$ there is $\alpha>0$ such that for any $x\in X$:
\begin{equation}\label{eq:WhenNearMinimumNearMinimumDerivative2}
f_x(0) - f(0) < \alpha \Rightarrow \exists y\in \MinAtZero \mid \: \left| f'_x(0) -  f'_y(0) \right| < \betaAp'.
 \end{equation}
We now assume that the quantity $a\in [-1,1]$, defined by \eqref{eq:DefinitionA}, is restricted to $a \in (-1,1]$.
We treat the cases 
\begin{itemize} 
\item[\textbf{(1)}] $f_x(0) - f(0) \geq \alpha$ and  
\item[\textbf{(2)}] $f_x(0) - f(0) < \alpha$,
\end{itemize} 
separately. 
Once we have treated these two cases we'll consider the special setting where $a=1$. 

In case \textbf{(1)} 
we use the Assumption (b), that is the uniform bound (by $1$) on the derivatives as described in \eqref{eq:UniformBoundDerivatives} to conclude that $|f_x(t) - f(0)| \leq |\int_0^t f'_x(\tau) \ud \tau|  \leq \int_0^t |f'_x(\tau)| \ud \tau \leq t$. Combining this with $f_x(0) - f(0) \geq \alpha$ (the assumption \textbf{(1)}) yields $f_x(t) - f(0) \geq \alpha - t $.
This means that for any $a>-1$ and in particular the $a$ defined by \eqref{eq:DefinitionA}, we have,
 \begin{equation}\label{eq:LowerBoundRightDerivativeFirstCase}
\forall t \in \left [0, \frac{\alpha}{1+a} \right), \:   f_x(t) - f(0) > a t.
 \end{equation}
In case \textbf{(2)}, that is when $f_x(0) - f(0) < \alpha$, combining  
\eqref{eq:DefinitionA}  and \eqref{eq:WhenNearMinimumNearMinimumDerivative2} yields $f'_x(0) \geq a - \betaAp'$. Using \eqref{eq:FirstOrderTaylorWithUniformRest} now gives
 \begin{equation}\label{eq:LoweBoundRightDerivativeSecondCase}
f_x(t) - f(0) \geq (a - \betaAp')t  - \epsilon(t) t.
 \end{equation}
The inequalities both \eqref{eq:LowerBoundRightDerivativeFirstCase} and \eqref{eq:LoweBoundRightDerivativeSecondCase} imply, or put differently in both cases we have, 
 \begin{equation}\label{eq:LoweBoundRightDerivativeFx}
\forall \betaAp'>0 , \exists \alpha'>0 \mid \forall x \in X, \: t\in [0, \alpha'] \Rightarrow  f_x(t) - f(0) \geq (a - \betaAp')t  - \epsilon(t) t
 \end{equation}
So that, since $f(t) = \inf_{x\in X} f_x(t)$:
 \begin{equation}\label{eq:LoweBoundRightDerivative1}
\forall \betaAp'>0 , \exists \alpha'>0 \mid t\in [0, \alpha'] \Rightarrow  f(t) - f(0) \geq (a - \betaAp')t  - \epsilon(t) t
 \end{equation}
 which gives:
\begin{equation}\label{eq:LoweBoundRightDerivative}
f(t) - f(0) \geq a t  - o(t)
 \end{equation}
{Up to this point we assumed $a>-1$. However, from Assumption (b) we know that, for any $x\in X$, $f_x$ is $1$-Lipschitz.
  Therefore, $f(t) =  \inf_{x\in X} f_x(t)$ is also $1$-Lipschitz, and we have that $\forall t\in [0, \eta), f(t)-f(0) \geq -t$. 
 It follows that \eqref{eq:LoweBoundRightDerivative} holds as well in the case $a=-1$.}

Combining \eqref{eq:UpperBoundRightDerivative} and \eqref{eq:LoweBoundRightDerivative} yields the desired expression for the right derivative 
 \begin{equation}
\frac{d}{dt^+}_{\mid t=0}  f(t) = a.
\nonumber
 \end{equation}
\end{proof}

Having established our preparatory result on right derivatives, we now concentrate on our core technical lemma.

\begin{lemma}\label{lemma:SemiContinuityGamma}
If $\Su$ is a closed subset of the complete Riemann manifold $\N$,
$p$ is a point in $\N \setminus \Su$, 
and $w$ is a unit vector $T_p\N$, then:
\begin{equation}\label{eq:RightDerivativeAlongW}
\frac{d}{dt^+}_{\mid t=0} \rho_{\Su}( \exp_p(t w) ) = \inf_{v \in \Gamma_{\Su}(p) } - \langle v, w \rangle =  \min_{v \in \Gamma_{\Su}(p) } - \langle v, w \rangle . 
\end{equation}
\end{lemma}

\begin{proof}
We first observe that we can localize the problem. More precisely we have the following: 
The distance $d( \exp_p(t w) , \Su)$ to $\Su$ and the distance $d( \exp_p(t w) , \Su^{\boxplus \rho_{\Su}(p) - \rho})$ to its offset differ by the constant $\rho_{\Su}(p) - \rho$. Thanks to \eqref{eq:DefinitionGamma2} the directions of shortest paths to $\Su$ remain the same. This means that for any $0< \rho \leq \rho_{\Su}(p)$, we can replace $\Su$ by $\Su' =\Su^{\boxplus \rho_{\Su}(p) - \rho}$ without loss of generality. 

In particular we can restrict ourselves to a neighbourhood that is smaller than $\rho$ where $\rho$ is as chosen in Lemma \ref{lemma:InverseExpIsLipschitz}.


Thanks to the semi-continuity of the directions of shortest paths to $\Su$, see \eqref{eq:SemiContinuityGammaTilde}, we have that for any $\varepsilon >0$ there is some $\alpha>0$ such that for any $t \in [0,\alpha]$
\begin{equation}
 \tilde{\Gamma}_{\Su'}( \exp_p(t w)) \subseteq  \tilde{\Gamma}_{\Su'}(p)^{\boxplus \varepsilon} . 
\end{equation}

We apply Lemma \ref{lemma:DerivativeOfInf} where, for $x\in  X := \tilde{\Gamma}_{\Su'}(p)^{\boxplus \varepsilon}$, we define:
\begin{equation}
f_x(t) := d( \exp_p(t w) , x), 
\end{equation}
and one has, for $t\in [0, \alpha]$:
\begin{equation}
f(t) := d\left( \exp_p(t w) , \Su' \right) = d\left( \exp_p(t w) ,   \tilde{\Gamma}_{\Su'}(p)^{\boxplus \varepsilon} \right) = \inf_{x\in X} f_x(t).
\end{equation}

We verify that this family of functions satisfies the conditions of Lemma  \ref{lemma:DerivativeOfInf}.

%

\begin{itemize} 
\item Condition  $(a)$ is given by  Lemma \ref{lemma:DerivativeDistanceToPointAntUnifContDeriv},
and the uniform equicontinuity specifically by  \eqref{eq:UniformContinuityDerivativeDistance}.
\item Condition $(b)$ follows from the fact that, since $w$ is a unit  vector,
 both $t\mapsto  \exp_p(t w) $ and   $y \mapsto d(y,x)$ are $1$-Lipschitz  functions.
\item Condition $(c)$ follows, since, because $\tilde{\Gamma}_{\Su'}(p)$ is not empty:
\begin{equation}
\MinAtZero = \{ x \in \tilde{\Gamma}_{\Su'}(p)^{\boxplus \varepsilon} \mid d(p,x) = \rho_{\Su'}(p) \} = \tilde{\Gamma}_{\Su'}(p) \neq \emptyset.
\nonumber
\end{equation}
\item Establishing Condition $(d)$ is significantly more involved.  
Thanks to Lemma \ref{lemma:DerivativeDistanceToPointAntUnifContDeriv} one has:

\begin{equation}
\frac{d}{dt}_{\mid t=0} d(x, \exp_p (tw) ) = - \langle v, w \rangle,
\tag{\ref{eq:DistanceToOnePoint1}} 
\end{equation}

where $v = \exp_p^{-1} (x)$. Because we localized the problem to a neighbourhood whose size is given in Lemma \ref{lemma:InverseExpIsLipschitz}, the map $\exp_p^{-1}$ is continuous in $B(p, 2 \rho)$. This means that 
there is $\varepsilon'>0$ such that 
\[
d(y,x) < \varepsilon' \Rightarrow  | \langle \exp_p^{-1} (x), w \rangle - \langle \exp_p^{-1} (y), w \rangle | < \varepsilon,
\] 
or, by \eqref{eq:DistanceToOnePoint1} 
\begin{equation}\label{eq:XtoDerivativeContinue}
d(x,y) < \varepsilon' \Rightarrow \left|   \frac{d}{dt}_{\mid t=0} d( \exp_p(t w) , x)    -        \frac{d}{dt}_{\mid t=0} d( \exp_p(t w) , y)       \right| < \varepsilon.
\end{equation}
We now use Lemma \ref{lemma:SemiContinuityGammaTilde} where \eqref{eq:SemiContinuityGammaTilde2} reads 
\begin{equation}\tag{\ref{eq:SemiContinuityGammaTilde2}} 
\forall p\in   \N \setminus \Su', \:  \forall \varepsilon'>0, \: \exists \alpha >0 \mid    B \left( p, \rho_{\Su}(p) +  \alpha  \right)  \cap  \Su  \subseteq  \tilde{\Gamma}_{\Su}(p)^{\boxplus  \varepsilon'} .
 \end{equation}
Since $\inf_{x'\in X} f_{x'}(0) = \rho_{\Su'}(p)$ and $\MinAtZero = \tilde{\Gamma}_{\Su'}(p)$, we get:
\begin{align}
\forall \varepsilon >0, \: \exists \alpha>0 \mid  \forall x \in X  & &\nonumber \\
f_x(0) - \rho_{\Su'}(p) < \alpha    
&\Rightarrow \exists y \in \tilde{\Gamma}_{\Su'}(p)\mid d(x,y)  < \varepsilon' & \text{by } \eqref{eq:SemiContinuityGammaTilde2},     \nonumber \\
&\Rightarrow \exists y \in \tilde{\Gamma}_{\Su'}(p)\mid \left| f'_x(0) -  f'_y(0) \right| < \varepsilon & \text{by }  \eqref{eq:XtoDerivativeContinue}.
 \end{align}
\end{itemize} 

We get as a result of Lemma  \ref{lemma:DerivativeOfInf} that
\begin{align*}
\frac{d}{dt^+}_{\mid t=0} d( \exp_p(t w) , \Su) &= \inf_{y \in \tilde{\Gamma}_{\Su}(p)} \frac{d}{dt}_{\mid t=0} d(y,  \exp_p(t w) ) &\\
&=  \inf_{v \in \Gamma_{\Su}(p) } - \langle v, w \rangle.  & \text{by } \eqref{eq:DistanceToOnePoint1}
\end{align*}
The ``$\inf$'' becomes a ``$\min$'' in \eqref{eq:RightDerivativeAlongW} because $\Gamma_{\Su}(p)$ is compact and therefore the minimum is attained.
\end{proof}

%

We can now finally establish the main result of this appendix. 


\begin{proof}[Proof of Theorem \ref{theorem:ClosestDirectionIsMinusGradient}]
By the definition of the superdifferential and in particular \eqref{eq:GenetralizedGradientRiemann}, one has:
\begin{align}
d^+ \rho_{\Su} (p) &= \left\{v \in T_p\N \middle|   \forall  w\in T_p\N,  \: \rho_{\Su} (\exp_p(w)) - \rho_{\Su}(p) \leq  \langle v, w \rangle + o(|w|)  \right\}
\nonumber   \\
 &=  \bigcap_{w\in T_p\N}  \left\{v \in T_p\N \middle|    \: \rho_{\Su}(\exp_p(w)) - \rho_{\Su}(p) \leq  \langle v, w \rangle + o(|w|)  \right\}   
\nonumber \\
 &=  \bigcap_{w\in T_p\N}  \left\{v \in T_p\N \middle |    \: \frac{d}{dt^+}_{\mid t=0} \rho_{\Su}( \exp_p(t w) )  \leq  \langle v, w \rangle )  \right\}   
\nonumber \\
&=  \bigcap_{w\in T_p\N}  \left\{v \in T_p\N \middle|    \:  \inf_{u \in \Gamma_{\Su}(p) } - \langle u, w \rangle  \leq  \langle v, w \rangle )  \right\}    \quad( \text{by Lemma } \ref{lemma:SemiContinuityGamma})
\nonumber 
\\
&=  \bigcap_{\substack{w\in T_p\N \\ |w|=1 }} \left\{v \in T_p\N \middle|   \:   \langle -v, w \rangle   \leq  \sup_{u \in \Gamma_{\Su}(p) } \langle u, w \rangle   \right\} . \label{SuperGRadientAsConvexHull}
\end{align}
Note that $\left\{-v \in T_p\N \mid   \:   \langle -v, w \rangle   \leq  \sup_{u \in \Gamma_{\Su}(p) } \langle u, w \rangle   \right\} $ is the intersection of all
half-spaces orthogonal to $w$, with $w$ pointing outward and containing $\Gamma_{\Su}(p)$. 
With this observation 
\eqref{SuperGRadientAsConvexHull} is precisely \eqref{eq:GenetralizedGradientIsMinusClosetsDirections}.

\end{proof}

\end{document}